\documentclass[journal]{IEEEtran}
\usepackage{macrosN}
\allowdisplaybreaks
\begin{document}

\title{Fidelity-Based Smooth Min-Relative Entropy: Properties and Applications}

\author{Theshani~Nuradha,~\IEEEmembership{Graduate Student Member,~IEEE,} and~
       Mark~M.~Wilde,~\IEEEmembership{Fellow,~IEEE}%
\thanks{T. Nuradha and M.~M. Wilde are with the school of Electrical and Computer Engineering, Cornell University, Ithaca, New York 14850, USA. (e-mail: pt388@cornell.edu, wilde@cornell.edu)}%
\thanks{Authors acknowledge the support from the National Science Foundation under Grant No.~2315398}%
}

\maketitle

\begin{abstract}
The fidelity-based smooth min-relative entropy is a distinguishability measure that has appeared in a variety of contexts in prior work on quantum information, including resource theories like thermodynamics and coherence. Here we provide a comprehensive study of this quantity. First we prove that it satisfies several  basic properties, including the  data-processing inequality. We also establish connections between the fidelity-based smooth min-relative entropy and other widely used information-theoretic quantities, including smooth min-relative entropy and smooth sandwiched R\'enyi relative entropy, of which the sandwiched R\'enyi relative entropy and smooth max-relative entropy are special cases. After that, we use these connections to establish the second-order asymptotics of the fidelity-based smooth min-relative entropy and all smooth sandwiched R\'enyi relative entropies, finding that the first-order term is the quantum relative entropy and the second-order term involves the quantum relative entropy variance. Utilizing the properties derived, we also show how the fidelity-based smooth min-relative entropy provides one-shot bounds for operational tasks in general resource theories in which the target state is mixed, with a particular example being randomness distillation. The above observations then lead to second-order expansions of the upper bounds on distillable randomness, as well as the precise second-order asymptotics of the distillable randomness of particular classical--quantum states. 
Finally, we establish semi-definite programs for smooth max-relative entropy and smooth conditional  min-entropy, as well as a bilinear program for the fidelity-based smooth min-relative entropy, which we subsequently use to explore the tightness of a bound relating the last to the first.
\end{abstract}


\tableofcontents

\section{Introduction}

\subsection{Background}

Distinguishability plays a fundamental role across all fields of sciences. The core toolbox in this regard involves distinguishability measures. In quantum information theory, these distinguishability measures then lead to information measures including mutual information and conditional entropy, as well as entanglement measures (see \cite[Chapters 4 and 5]{khatri2020principles} for a review). Furthermore, they also appear in resource theories as conversion rates \cite{chitambar2019quantum}.

The min-relative entropy is one such distinguishability measure \cite{datta2009min}, defined for a pure state $\psi\coloneqq | \psi\rangle\! \langle \psi| $ and a positive semi-definite operator $\sigma$  as
\begin{align}
     D_{\min}(\psi \Vert \sigma) & \coloneqq -\log_2 \operatorname{Tr}[ \psi \sigma ] \label{eq:min-rel-ent-pure-1}\\
     & = -\log_2 F(\psi, \sigma ). \label{eq:min-rel-ent-pure-2}  
\end{align}
In \eqref{eq:min-rel-ent-pure-2} above, we made use of the fidelity \cite{uhlmann1976transition}, defined generally for two positive semi-definite operators $\omega$ and $\tau$ as
\begin{equation}
F(\omega, \tau) \coloneqq \left \Vert \sqrt{\omega} \sqrt{\tau }\right \Vert_1^2.  
\label{eq:fidelity-def}
\end{equation}
In the case that $\sigma$ is a state, we can interpret the expression $\operatorname{Tr}[ \psi \sigma ]$ in \eqref{eq:min-rel-ent-pure-1} as the probability that the first measurement outcome occurs when performing the measurement $\{\psi,I-\psi\}$ on the state $\sigma$. Alternatively, we can interpret the expression $F(\psi, \sigma )$ in \eqref{eq:min-rel-ent-pure-2} as the fidelity between the states $\psi$ and $\sigma$. Thus, in the first case, we are interpreting $\psi$ as a measurement operator, and in the second case, we are interpreting $\psi$ as a state.

Given the above, there are at least two ways of generalizing the min-relative entropy when $\rho$ is a general state. The first approach, originally introduced in \cite{datta2009min} as the min-relative entropy, defines it as
\begin{equation}
\label{eq:min-rel-ent-def}
    D_{\min}(\rho \Vert \sigma) \coloneqq - \log_{2} \Tr[\Pi_\rho \sigma],
\end{equation}
 where $\Pi_\rho$ denotes the projection onto the support of $\rho$. Clearly, this definition generalizes the expression in \eqref{eq:min-rel-ent-pure-1}, interpreting $\Pi_\rho$ as a measurement operator.
As discussed in \cite{datta2009min}, this interpretation is directly linked with the operational meaning of the min-relative entropy in asymmetric hypothesis testing, as the minimum Type II error exponent if the Type I error probability is constrained to be equal to zero. This quantity has been further interpreted in a resource-theoretic manner as the maximum number of exact bits of asymmetric distinguishability that can be distilled from the pair $(\rho,\sigma)$ \cite{wang2019resource}. Given the strong link between hypothesis testing and information theory \cite{blahut74}, the min-relative entropy finds further use as the basic quantity underlying optimal rates at which zero-error distillation is possible \cite{wang2019resource,wang19channels}.

The second generalization of min-relative entropy to a general state $\rho$ employs the formula in \eqref{eq:min-rel-ent-pure-2}, and is defined as follows \cite{dupuis2014generalizedDmin}:
\begin{equation}
    D_{\min,F}(\rho\Vert\sigma) \coloneqq -\log_2 F(\rho,\sigma).
    \label{eq:FB-min-rel-ent-def}
\end{equation}
In the above definition and throughout, we use the extra subscript $F$ to indicate that this generalization is based on fidelity, and we refer to it as the $F$-min-relative entropy. This quantity is also equal to the sandwiched R\'enyi relative entropy of order 1/2 \cite{muller2013quantum,wilde2014strong}. It is not known to have an operational meaning in hypothesis testing; however, it has appeared in a variety of previous works in quantum information theory \cite{dupuis2014generalizedDmin,faist2016quantumDmin,zhao2019oneDmin,ramakrishnan2023moderate,RT22}.

In realistic experimental scenarios, it is pertinent to allow for approximations in terms of a smoothing parameter \cite{RW04,RennerThesis}, which characterizes the error that can occur in an experiment. We can then consider smoothed versions of the quantities in~\eqref{eq:min-rel-ent-def} and \eqref{eq:FB-min-rel-ent-def}.

Let us first consider smoothing the quantity in \eqref{eq:min-rel-ent-def}. The approach employed finds its roots naturally in asymmetric hypothesis testing, given the operational scenario discussed above. With this in mind, if we relax the aforementioned Type~I error probability constraint, such that it is allowed to be larger than zero, then we arrive at the smooth min-relative entropy with smoothing parameter~$\varepsilon \in [0,1]$ \cite{BD10,brandao2011one,WR12}:
\begin{equation}
\label{eq:hypothesis-testing-relative-entropy}
    D^\varepsilon_{\min}(\rho \Vert \sigma) \coloneqq -\log_{2} \inf_{0 \leq \Lambda \leq I} \left\{ \Tr[\Lambda \sigma]: \Tr[\Lambda \rho] \geq 1- \varepsilon \right\}, 
\end{equation}
with $\{ \Lambda, I- \Lambda\}$ being the measurement that distinguishes between $\rho$ and $\sigma$. This quantity is referred to as the smooth min-relative entropy in \cite{bu2018max,LiuZi-WenPhd18, wang2019resource,wang19channels} and as the hypothesis testing relative entropy in \cite{WR12} and many other papers, including \cite{tomamichel2013hierarchy,li2014second,dupuis2014generalizedDmin,datta2014second,DTW14,datta2016second,khatri2020principles}. Considering that \cite{dupuis2014generalizedDmin, wang2019resource}
\begin{equation}
\lim_{\varepsilon \to 0}  D^\varepsilon_{\min}(\rho \Vert \sigma)=D_{\min} (\rho \Vert \sigma),    
\end{equation}
it is clear that $D^\varepsilon_{\min}(\rho \Vert \sigma)$ is a smoothed version of $D_{\min}(\rho\Vert\sigma)$.

Let us now consider smoothing the quantity in \eqref{eq:FB-min-rel-ent-def}. Since $\rho$ is a state, the idea when smoothing is to search for a nearby subnormalized state $\widetilde{\rho}$, such that it satisfies $F(\widetilde{\rho}, \rho) \geq 1- \varepsilon$ for a fixed smoothing parameter $\varepsilon \in [0,1]$, and then replace $\rho$ with $\widetilde{\rho}$ when comparing with $\sigma$. This reasoning naturally leads to the fidelity-based smooth min-relative entropy:
\begin{equation}
\label{eq:smooth-min-relative-entropy}
    D_{\min,F}^{\varepsilon}(\rho\Vert\sigma)\coloneqq-\log_{2}\inf_{\widetilde
{\rho}\in\mathcal{D}_{\leq}}\left\{  F(\widetilde{\rho},\sigma):F(\widetilde
{\rho},\rho)\geq1-\varepsilon\right\}, 
\end{equation}
where $\mathcal{D}_{\leq}$ denotes the set of sub-normalized states (see \cref{def:smooth-min-relative-fidelity}, as well as \cref{rem:subnorm-choice} for the choice of sub-normalized states). In what follows, we refer to it more simply as the smooth $F$-min-relative entropy. This quantity has been considered in several prior works \cite{dupuis2014generalizedDmin,faist2016quantumDmin,zhao2019oneDmin,ramakrishnan2023moderate,RT22}. It is interesting to compare the expressions involved in \eqref{eq:hypothesis-testing-relative-entropy} and \eqref{eq:smooth-min-relative-entropy}, where we observe that the main difference is that $D^{\varepsilon}_{\min}$ compares $\rho$ and $\sigma$ to a measurement operator $\Lambda$ via a trace overlap, whereas $D^{\varepsilon}_{\min,F}$ compares $\rho$ and $\sigma$ to a subnormalized state $\widetilde{\rho}$ via the fidelity.

By building on the recent observations of \cite{lami2022upper}, one contribution of the present paper is that the smooth $F$-min-relative entropy finds use in operational tasks such as randomness distillation. In this task, the goal is to distill a state that is close in fidelity to a maximally classically correlated state, which is a mixed state. As such, the approach to smoothing taken in \eqref{eq:smooth-min-relative-entropy} is more relevant in this scenario than that in \eqref{eq:hypothesis-testing-relative-entropy} and can be used to obtain upper bounds on the one-shot distillable randomness of a bipartite state. More generally, and as discussed in \cite{lami2022upper}, we suspect that the ideas put forward here will find use in quantum resource transformations in which the target state is a mixed state, and we provide some evidence in \cref{sec:apps-gen-res} that this is the case.

More broadly, the main goal of the present paper is to provide a comprehensive study of the fidelity-based smooth min-relative entropy in \eqref{eq:smooth-min-relative-entropy}, which, as indicated above, could be useful for understanding the fundamental limits of resource transformations in general resource theories.

\subsection{Contributions}

In this paper, we first derive several properties of the fidelity-based smooth min-relative entropy. In particular, we prove that it satisfies data processing (\cref{thm:data-processing}), as well as scaling, super-additivity, monotonicity, etc.~(\cref{thm:other-properties}).
We note here that the data-processing inequality was already established in \cite[Theorem~3]{RT22}, but here we provide an independent proof.
Then we proceed to establish its connections with other quantum information-theoretic quantities including sandwiched R\'enyi relative entropy and its smooth variants, smooth max-relative entropy, and smooth min-relative entropy. 

Next, with the assistance of the derived connections, we provide a second-order asymptotic analysis for the smooth $F$-min-relative entropy (\cref{thm:second-order-Dmin}).  There we find that the first-order term is the quantum relative entropy and the second-order term involves the quantum relative entropy variance. In addition, we derive the second-order behaviour of the smooth sandwiched R\'enyi relative entropy (\cref{cor:second-order-smooth-sandwiched-Renyi}). This corollary indicates that, in the asymptotic i.i.d.~setting and up to the second order, there is no difference between all of the smooth sandwiched R\'enyi relative entropies for all $\alpha >1$: they are all equivalent to the smooth max-relative entropy in this setting. Similarly, in the asymptotic i.i.d.~setting and up to the second order, there is no difference between all of them for $\alpha \in [1/2,1)$: they are all equivalent to the smooth min-relative entropy in this setting.

Furthermore, we show how the smooth $F$-min-relative entropy provides one-shot bounds for operational tasks in general resource-theoretic settings, with a particular analysis focusing on randomness distillation from bipartite states. We derive second-order expansions of the upper bounds on the LOCC-assisted distillable randomness (\cref{thm:Upper-Bound-Second-Order-general}), as well as the precise second-order asymptotics of the distillable randomness of particular classical--quantum states (\cref{thm:second-order-expansion-CR-some-cq}).

Moreover, we provide a method to compute the smooth $F$-min-relative entropy by means of a bilinear program (\cref{prop:compute-smooth-min-relative-entropy}). We also provide semi-definite programs (SDPs) for smooth max-relative entropy and smooth conditional min-entropy (Propositions~\ref{prop: SDP for smooth max normalized} and \ref{prop:SDP-for-smooth-max-with-sub-normalized}), which may be of independent interest.

\subsection{Organization}

The rest of our paper is organized as follows.  In \cref{Sec:background}, we introduce notation and preliminaries. The focus of \cref{Sec:Properties} is on deriving some basic properties of the smooth $F$-min-relative entropy, including data processing. Connections to other quantum information-theoretic quantities are established in \cref{Sec:Connections}. In \cref{Sec:Second-order-asymptotic}, we study the second-order asymptotics of the smooth $F$-min-relative entropy. We explore how smooth $F$-min-relative entropy provides bounds in operational tasks related to general resource theories in \cref{Sec: Randomness Distillation}. In \cref{Sec:Computational-analysis} we provide methods to compute the smooth $F$-min-relative entropy and related quantities, including the smooth max-relative entropy. 
Finally, \cref{Sec:Conclusion} provides concluding remarks and future directions.

\section{Preliminaries} \label{Sec:background}

\subsection{Basic Concepts and Notation}

We begin by reviewing basic concepts from quantum information theory and refer the reader to \cite{khatri2020principles} for more details. A quantum system~$R$ is identified with a finite-dimensional Hilbert space~$\mathcal{H}_R$. We denote the set of linear operators acting on $\mathcal{H}_R$ by $\mathcal{L}(\mathcal{H}_R)$. The support of a linear operator $X \in \mathcal{L}(\mathcal{H}_R)$ is defined to be the orthogonal complement of its kernel, and we denote it by $\operatorname{supp}(X)$.
Let $\operatorname{T}(X)$ 
denote the transpose of $X$. 
The partial transpose of $C_{AB} \in \mathcal{L}(\mathcal{H}_A \otimes \mathcal{H}_B)$ on the system~$A$ is represented as $\T_A(C_{AB})$. Let $\Tr\!\left[C_{AB} \right]$ denote the trace of~$C_{AB}$, and let $\Tr_A \!\left[C_{AB}\right]$ denote the partial trace of $C$ over the system $A$. We use the standard notation $C_A \equiv \Tr_{B}[C_{AB}]$ and $C_B \equiv \Tr_{A}[C_{AB}]$ to denote the marginals of $C_{AB}$. The trace norm  of an operator $B$ is defined as $\left\|B\right\|_1 \coloneqq \Tr[\sqrt{B^\dagger B} ]$. For Hermitian operators $A$ and $B$, the notation $A \geq B$ indicates that $A-B$ is a positive semi-definite (PSD) operator, while $A > B$ indicates that $A-B$ is a positive definite operator.

A quantum state $\rho_R\in\mathcal{L}(\mathcal{H}_R)$ of system $R$ is a  PSD, unit-trace operator acting on $\mathcal{H}_R$. We denote the set of all density operators acting on $\mathcal{H}_R$ as $\mathcal{D}(\mathcal{H}_R)$ (we also refer to the set of density operators by $\cD$ when there is no ambiguity regarding the underlying Hilbert space). A rank-one state $\rho_R$  is called pure, and in this case there exists a state vector $| \psi \rangle \in \mathcal{H}_R$ such that $\rho_R= | \psi \rangle\!\langle \psi | $. Otherwise,
$\rho_R$ is called a mixed state. By the spectral decomposition theorem, every state can be written as a convex combination
of pure, orthogonal states. 
A quantum channel $\mathcal{N}: \mathcal{L}(\mathcal{H}_A ) \to \mathcal{L}(\mathcal{H}_B)$ is a linear, 
completely positive and trace-preserving (CPTP) map from $\mathcal{L}(\mathcal{H}_A)$ to $\mathcal{L}(\mathcal{H}_B)$. We denote the Hilbert--Schmidt adjoint of $\mathcal{N}$ by $\mathcal{N}^\dagger$. A measurement of a quantum system $R$ is described by a
positive operator-valued measure (POVM) $\{M_y\}_{y \in \mathcal{Y}}$, which is defined to be a collection of PSD operators  satisfying $\sum_{y \in \mathcal{Y}} M_y= I_{R}$, where $I_{R}$ is the identity operator and $\mathcal{Y}$ is a finite alphabet. The Born rule dictates that, when applying the above POVM to a state $\rho$, the probability of observing the outcome $y$ is given by $\Tr\!\left[M_y \rho \right]$.

\subsection{Divergences}

First, let us recall the definition of the  fidelity-based smooth
min-relative entropy, which is the main distinguishability measure of interest in our paper. 

\begin{definition}[Fidelity-based smooth min-relative entropy] \label{def:smooth-min-relative-fidelity}
Fix $\varepsilon\in\left[  0,1\right]  $. The fidelity-based smooth
min-relative entropy of a state $\rho$ and a PSD operator
$\sigma$ is defined as
\begin{equation}
D_{\min,F}^{\varepsilon}(\rho\Vert\sigma)\coloneqq-\log_{2}\inf_{\widetilde
{\rho}\in\mathcal{D}_{\leq}}\left\{  F(\widetilde{\rho},\sigma):F(\widetilde
{\rho},\rho)\geq1-\varepsilon\right\}  , \label{eq:smooth-dmin}%
\end{equation}
where the fidelity of PSD 
operators $\omega$ and $\tau$ is
defined in  \eqref{eq:fidelity-def}
and $\mathcal{D}_{\leq}$ denotes the set of subnormalized states; i.e.,%
\begin{equation}
\mathcal{D}_{\leq}\coloneqq\left\{  \omega:\omega\geq0,\operatorname{Tr}%
[\omega]\leq1\right\}  .
\end{equation}
\end{definition}
Hereafter, we simply abbreviate this quantity as the smooth $F$-min-relative entropy.
Recalling the definition in \cite[Eq.~(8)]{dupuis2014generalizedDmin}, observe that when $\rho$ is a state and $\sigma$ is a PSD operator, there is no difference between  \cref{def:smooth-min-relative-fidelity} and the definition given in \cite[Eq.~(8)]{dupuis2014generalizedDmin}.

We call a distinguishability measure $\boldsymbol{D}(\cdot \Vert \cdot)$ a generalized divergence \cite{SW12} if it satisfies the data-processing inequality; i.e., for every channel $\cN$, state $\rho$, and PSD operator $\sigma$, 
\begin{equation}\label{eq:generalized-divergence}
    \boldsymbol{D}(\rho \Vert \sigma) \geq \boldsymbol{D}\!\left(\cN(\rho) \Vert \cN(\sigma) \right).
\end{equation}

Fix $\alpha \in (0,1) \cup (1, \infty)$. The sandwiched R\'enyi relative entropy of a state $\rho$ and a PSD operator $\sigma$ is defined as \cite{muller2013quantum, wilde2014strong}
\begin{multline}
  \widetilde{D}_\alpha(\rho\Vert \sigma)  \coloneqq  \\
  \begin{cases}\frac{1}{\alpha-1} \log_{2} \widetilde{Q}_\alpha(\rho\Vert \sigma)  & \mbox{if} \ \alpha \in (0,1), \mbox{or} \\ & \alpha \in (1, \infty), \ \supp(\rho)\subseteq\supp(\sigma),\\
+\infty & \mbox{otherwise},
\end{cases}\label{eq:sandwiched-renyi-def}
\end{multline}
where
\begin{equation}
    \widetilde{Q}_\alpha(\rho\Vert \sigma) \coloneqq \Tr\!\left[ \left( \sigma^{\frac{1-\alpha}{2\alpha}}\rho \sigma^{\frac{1- \alpha}{2\alpha}} \right)^\alpha \right].
\end{equation}
It is a generalized divergence for $\alpha \in [1/2,1)\cup(1,\infty)$ \cite{FL13} (see also \cite{W18opt,W18optISIT}), and satisfies the following $\alpha$-monotonicity property \cite{muller2013quantum}:
\begin{equation}
    0 < \alpha \leq \beta \quad \Rightarrow \quad \widetilde{D}_\alpha(\rho\Vert \sigma) \leq \widetilde{D}_\beta(\rho\Vert \sigma) .
    \label{eq:sandwiched-alpha-mono}
\end{equation}
For $\alpha=1/2$, observe that
\begin{equation} 
  \widetilde{D}_{1/2}(\rho\Vert \sigma) = -\log_{2} F(\rho,\sigma) = D_{\min,F}(\rho\Vert \sigma).
  \label{eq:alpha-1/2-D_min,F}
\end{equation}
The special case of $\alpha\to 1$ reduces to the quantum relative entropy \cite{muller2013quantum, wilde2014strong}:
\begin{equation}
\lim_{\alpha \to 1} \widetilde{D}_\alpha(\rho \Vert \sigma) = D(\rho \| \sigma) \label{eq:limit-alpha-1},
\end{equation}
the latter defined as \cite{U62}
\begin{equation}
    D(\rho \| \sigma) \coloneqq  \Tr\!\left [\rho (\log_2 \rho - \log_2 \sigma) \right]
    \label{eq:relative-entropy-def}
\end{equation}
if $\supp(\rho)\subseteq\supp(\sigma)$ and as $+\infty$ otherwise.
The relative entropy variance 
$V(\rho\Vert \sigma)$ is defined as \cite{tomamichel2013hierarchy,li2014second}
\begin{equation}
   V(\rho\Vert \sigma) \coloneqq \Tr\! \left[ \rho \left( \log_{2} \rho - \log_{2} \sigma \right)^2 \right] -\left( D(\rho \Vert \sigma) \right)^2. 
   \label{eq:rel-ent-var}
\end{equation}

The Petz--R\'enyi relative entropy is defined for $\alpha\in(0,1)\cup(1,\infty)$, a state $\rho$, and a PSD operator $\sigma$ as \cite{P85,P86}
\begin{multline}
  D_\alpha(\rho\Vert \sigma)  \coloneqq  \\
  \begin{cases}\frac{1}{\alpha-1} \log_{2} Q_\alpha(\rho\Vert \sigma)  & \mbox{if} \ \alpha \in (0,1), \mbox{or} \\ & \alpha \in (1, \infty), \ \supp(\rho)\subseteq\supp(\sigma),\\
+\infty & \mbox{otherwise},
\end{cases}
\end{multline}
where
\begin{equation}
    Q_\alpha(\rho\Vert \sigma) \coloneqq \Tr[  \rho^\alpha \sigma^{1- \alpha}  ].
\end{equation}
It is a generalized divergence for $\alpha \in (0,1)\cup(1,2]$ \cite{P85,P86}.

The max-relative entropy of a state $\rho$ and a PSD operator~$\sigma$ is defined as \cite{datta2009min}
\begin{equation}
D_{\max}(\rho\Vert\sigma)\coloneqq \log_{2}\inf_{\lambda\geq0}\left\{
\lambda:\rho\leq\lambda\sigma\right\}  .
\end{equation}
It is known from  \cite{muller2013quantum} that
\begin{equation} \label{eq: dmax infity}
D_{\max}(\rho\Vert\sigma)=\lim_{\alpha\rightarrow\infty}\widetilde{D}_{\alpha
}(\rho\Vert\sigma).
\end{equation}
The smooth max-relative entropy of a state $\rho$ and a PSD
operator $\sigma$ is defined for $\varepsilon\in\left[  0,1\right]  $ as \cite{datta2009min} (see also \cite{tomamichel2015quantum})
\begin{equation}
D_{\max}^{\varepsilon}(\rho\Vert\sigma)\coloneqq\inf_{\widetilde{\rho}%
\in\mathcal{D}_{\leq}}\left\{  D_{\max}(\widetilde{\rho}\Vert\sigma
):F(\widetilde{\rho},\rho)\geq1-\varepsilon\right\}  .
\label{eq:smooth-dmax-def}%
\end{equation}
We also define the following variant
\begin{equation} \label{eq:smooth-max-normalized-to-hypothesis}
    \widehat{D}_{\max}^{\varepsilon}(\rho\Vert\sigma)\coloneqq\inf_{\widetilde{\rho}%
\in\mathcal{D}}\left\{  D_{\max}(\widetilde{\rho}\Vert\sigma
):F(\widetilde{\rho},\rho)\geq1-\varepsilon\right\} , 
\end{equation} 
and note that
\begin{equation}
    D^\varepsilon_{\max}(\rho \Vert \sigma) \leq \widehat{D}^\varepsilon_{\max}(\rho \Vert \sigma), \label{eq:smooth-max-with-sub-and-not-normalized}
\end{equation}
which follows since $\mathcal{D} \subseteq \cD_{\leq}$.

\section{Properties of Smooth \texorpdfstring{$F$}{F}-Min-Relative Entropy} \label{Sec:Properties}

In this section, we derive some basic properties satisfied by the smooth $F$-min-relative entropy, including data processing, scaling, super-additivity, monotonicity, etc. Then in subsequent sections, we utilize these properties to obtain bounds on operational quantities arising in general resource-theoretic settings, in particular on the net rate of the distillable randomness of a bipartite state.

Before deriving these properties, let us first observe that we can always restrict the constraint in the definition of $D_{\min,F}^{\varepsilon}$ to be an equality constraint, by following the same line of reasoning from \cite[Appendix~B]{kaur2017upper}. 
\begin{remark}[Inequality constraint in the definition of smooth $F$-min-relative entropy]
\label{rem:ineq-to-eq-constraint}
For $\varepsilon \in [0,1)$,
the smooth $F$-min-relative entropy in \eqref{eq:smooth-dmin} can be rewritten as 
\begin{equation}
D_{\min,F}^{\varepsilon}(\rho\Vert\sigma)=-\log_{2}\inf_{\widetilde{\rho}%
\in\mathcal{D}_{\leq}}\left\{  F(\widetilde{\rho},\sigma):F(\widetilde{\rho
},\rho)=1-\varepsilon\right\}  .
\end{equation}
Indeed, if $\widetilde{\rho}$ is such that $F(\widetilde{\rho},\rho
)>1-\varepsilon$, then we can choose a positive constant $c=\left(
1-\varepsilon\right)  /F(\widetilde{\rho},\rho)\in\left(  0,1\right)  $ such
that $F(\rho^{\prime},\rho)=1-\varepsilon$, where $\rho^{\prime}%
=c\widetilde{\rho}$ and $\rho^{\prime}\in\mathcal{D}_{\leq}$. Furthermore, we
also have that $F(\widetilde{\rho},\sigma)>F(\rho^{\prime
},\sigma)$, so that the objective function only decreases under this change.
\end{remark}

\subsection{Data Processing}

In this section, we prove the data-processing inequality for the smooth $F$-min relative entropy. As noted previously, this finding was already established in \cite[Theorem~3]{RT22}, but here we provide an independent proof.
Before establishing the data-processing inequality, we prove the unitary invariance of the smooth $F$-min-relative entropy, which assists in proving the data-processing inequality.

\begin{lemma}[Unitary Invariance] \label{lem: unitary invariance}
   The smooth $F$-min-relative entropy $D_{\min,F}^{\varepsilon}$ is invariant under the action
of a unitary channel $\mathcal{U}$, i.e.,%
\begin{equation}
D_{\min,F}^{\varepsilon}(\rho\Vert\sigma)=D_{\min,F}^{\varepsilon}(\mathcal{U}%
(\rho)\Vert\mathcal{U}(\sigma)). \label{eq:unitary-inv}%
\end{equation} 
for all $\varepsilon\in\left[  0,1\right)  $, every state $\rho$, and PSD operator $\sigma$.
\end{lemma}

\begin{IEEEproof}
Let $\widetilde{\rho}\in\mathcal{D}_{\leq}$ satisfy
$F(\widetilde{\rho},\rho)\geq1-\varepsilon$. Then, by the unitary invariance
of fidelity, we conclude that
\begin{equation}
F(\mathcal{U}(\widetilde{\rho}),\mathcal{U}(\rho))=F(\widetilde
{\rho},\rho).     
\end{equation}
Thus, $F(\mathcal{U}(\widetilde{\rho}),\mathcal{U}(\rho
))\geq1-\varepsilon$, and%
\begin{align}
-\log_{2} F(\widetilde{\rho},\sigma)  &  =-\log_{2} F(\mathcal{U}(\widetilde{\rho
}),\mathcal{U}(\sigma))\\
&  \leq D_{\min,F}^{\varepsilon}(\mathcal{U}(\rho)\Vert\mathcal{U}(\sigma)).
\end{align}
Since the inequality holds for every $\widetilde{\rho}\in\mathcal{D}_{\leq}$
satisfying $F(\widetilde{\rho},\rho)\geq1-\varepsilon$, we conclude that%
\begin{equation} \label{eq: one direction unitary}
D_{\min,F}^{\varepsilon}(\rho\Vert\sigma)\leq D_{\min,F}^{\varepsilon}%
(\mathcal{U}(\rho)\Vert\mathcal{U}(\sigma)).
\end{equation} 
To prove the opposite inequality, let $\widehat{\rho}\in\mathcal{D}_{\leq}$
satisfy $F(\widehat{\rho},\mathcal{U}(\rho))\geq1-\varepsilon$. Then, by
unitary invariance of fidelity,%
\begin{align}
F(\widehat{\rho},\mathcal{U}(\rho))  &  =F(\mathcal{U}^{\dag}(\widehat{\rho
}),(\mathcal{U}^{\dag}\circ\mathcal{U})(\rho))\\
&  =F(\mathcal{U}^{\dag}(\widehat{\rho}),\rho).
\end{align}
Thus, $F(\mathcal{U}^{\dag}(\widehat{\rho}),\rho)\geq1-\varepsilon$, and%
\begin{align}
-\log_{2} F(\widehat{\rho},\mathcal{U}(\sigma))  &  =-\log_{2} F(\mathcal{U}^{\dag
}(\widehat{\rho}),(\mathcal{U}^{\dag}\circ\mathcal{U})(\sigma))\\
&  =-\log_{2} F(\mathcal{U}^{\dag}(\widehat{\rho}),\sigma)\\
&  \leq D_{\min,F}^{\varepsilon}(\rho\Vert\sigma).
\end{align}
Since the inequality holds for every $\widehat{\rho}\in\mathcal{D}_{\leq}$
satisfying $F(\widehat{\rho},\mathcal{U}(\rho))\geq1-\varepsilon$, we conclude
that%
\begin{equation} \label{eq: other direction unitary}
D_{\min,F}^{\varepsilon}(\mathcal{U}(\rho)\Vert\mathcal{U}(\sigma))\leq D_{\min,F}^{\varepsilon}(\rho\Vert\sigma).
\end{equation}
Then \eqref{eq:unitary-inv} follows from \eqref{eq: one direction unitary} and \eqref{eq: other direction unitary}.
\end{IEEEproof}

\medskip

Now, we are ready to present and prove the data-processing inequality for the smooth $F$-min-relative entropy.

\begin{theorem} \label{thm:data-processing}
The smooth $F$-min-relative entropy obeys the data-processing inequality:%
\begin{equation}
D_{\min,F}^{\varepsilon}(\rho\Vert\sigma)\geq D_{\min,F}^{\varepsilon}%
(\mathcal{N}(\rho)\Vert\mathcal{N}(\sigma)),\label{eq:DP-ineq-channels}%
\end{equation}
for all $\varepsilon\in\left[  0,1\right)  $, every state $\rho$, PSD operator $\sigma$, and channel~$\mathcal{N}$.
\end{theorem}

\begin{IEEEproof}
First, let us show that data processing holds under the partial trace channel.
That is, for every bipartite state $\rho_{AB}$ and bipartite positive
semi-definite operator $\sigma_{AB}$:%
\begin{equation}
D_{\min,F}^{\varepsilon}(\rho_{AB}\Vert\sigma_{AB})\geq D_{\min,F}^{\varepsilon
}(\rho_{A}\Vert\sigma_{A}). \label{eq:dp-partial-trace}%
\end{equation}
To this end, let $\widetilde{\rho}_{A}\in\mathcal{D}_{\leq}$ satisfy
$F(\widetilde{\rho}_{A},\rho_{A})\geq1-\varepsilon$. Then, by Uhlmann's
theorem \cite{uhlmann1976transition} there exists $\widetilde{\rho}_{AB} \in \mathcal{D}_{\leq}$ such that $F(\widetilde{\rho
}_{AB},\rho_{AB})=F(\widetilde{\rho}_{A},\rho_{A})$. Thus, $F(\widetilde{\rho
}_{AB},\rho_{AB})\geq1-\varepsilon$, and%
\begin{align}
-\log_{2} F(\widetilde{\rho}_{A},\sigma_{A})  &  \leq-\log_{2} F(\widetilde{\rho}%
_{AB},\sigma_{AB})\\
&  \leq D_{\min,F}^{\varepsilon}(\rho_{AB}\Vert\sigma_{AB}).
\end{align}
The first inequality follows from the data-processing inequality for fidelity, and the second follows from the definition of $D_{\min,F}^{\varepsilon}$.
Since this inequality holds for every $\widetilde{\rho}_{A}\in\mathcal{D}%
_{\leq}$ satisfying $F(\widetilde{\rho}_{A},\rho_{A})\geq1-\varepsilon$, we
conclude \eqref{eq:dp-partial-trace}.

Note that the same argument used for $D_{\min,F}^{\varepsilon}(\rho\Vert
\sigma)\leq D_{\min,F}^{\varepsilon}(\mathcal{U}(\rho)\Vert\mathcal{U}(\sigma))$ (\cref{lem: unitary invariance})
holds whenever $\mathcal{U}$ is an isometric channel, due to the isometric invariance of fidelity. Since the embedding
$\omega\rightarrow\omega\otimes|0\rangle\!\langle0|$ is an isometric channel,
we conclude that%
\begin{equation}
D_{\min,F}^{\varepsilon}(\rho\Vert\sigma)\leq D_{\min,F}^{\varepsilon}(\rho
\otimes|0\rangle\!\langle0|\Vert\sigma\otimes|0\rangle\!\langle
0|).\label{eq:ancilla-ineq-1}%
\end{equation}

Now let $\widetilde{\rho}_{SE}\in\mathcal{D}_{\leq}$ satisfy $F(\widetilde
{\rho}_{SE},\rho_{S}\otimes|0\rangle\!\langle0|_{E})\geq1-\varepsilon$. By
applying Lemma~\ref{lem:fid-overlap-0-lem}, we conclude that%
\begin{equation}
F(\widetilde{\rho}_{SE},\rho_{S}\otimes|0\rangle\!\langle0|_{E})=F(\widetilde
{\rho}_{S}^{0},\rho_{S}),
\end{equation}
where $\widetilde{\rho}_{S}^{0}:=\langle0|_{E}\widetilde{\rho}_{SE}%
|0\rangle_{E}$. Thus, $\widetilde{\rho}_{S}^{0}$ is a subnormalized state that
satisfies%
\begin{equation}
F\!\left(  \widetilde{\rho}^{0}_S,\rho_S\right)  =F(\widetilde{\rho}_{SE}%
,\rho_S\otimes|0\rangle\!\langle0|)\geq1-\varepsilon.
\end{equation}
Furthermore, again applying Lemma~\ref{lem:fid-overlap-0-lem}, we conclude
that%
\begin{equation}
F(\widetilde{\rho}_{SE},\sigma_{S}\otimes|0\rangle\!\langle0|_{E}%
)=F(\widetilde{\rho}_{S}^{0},\sigma_{S}).
\end{equation}
Thus, it follows that
\begin{align}
 \!\!\!\!\!\!\!\! -\log_{2}F(\widetilde{\rho}_{SE},\sigma_{S}\otimes|0\rangle\!\langle
0|_{E})
&  =-\log_{2}F\!\left(  \widetilde{\rho}^{0},\sigma\right)  \\
&  \leq D_{\min,F}^{\varepsilon}(\rho\Vert\sigma).
\end{align}
Since the inequality holds for every $\widetilde{\rho}_{SE}\in\mathcal{D}%
_{\leq}$ satisfying $F(\widetilde{\rho}_{SE},\rho\otimes|0\rangle
\!\langle0|)\geq1-\varepsilon$, we conclude that%
\begin{equation}
D_{\min,F}^{\varepsilon}(\rho\otimes|0\rangle\!\langle0|\Vert\sigma
\otimes|0\rangle\!\langle0|)\leq D_{\min,F}^{\varepsilon}(\rho\Vert
\sigma).\label{eq:ancilla-ineq-2}%
\end{equation}
Putting together \eqref{eq:ancilla-ineq-1} and \eqref{eq:ancilla-ineq-2}, we conclude
that%
\begin{equation}
D_{\min,F}^{\varepsilon}(\rho\otimes|0\rangle\!\langle0|\Vert\sigma
\otimes|0\rangle\!\langle0|)=D_{\min,F}^{\varepsilon}(\rho\Vert\sigma
).\label{eq:ancilla-eq}%
\end{equation}

Finally, since, by the Stinespring dilation theorem \cite{Sti55} (see also \cite{khatri2020principles}), every channel can be
realized in terms of

\begin{enumerate}
\item the map $\omega\rightarrow\omega\otimes|0\rangle\!\langle0|$,

\item a unitary channel, and

\item a partial trace,
\end{enumerate}

\noindent we conclude the desired inequality in \eqref{eq:DP-ineq-channels},
after putting together \eqref{eq:unitary-inv}, \eqref{eq:dp-partial-trace},
and \eqref{eq:ancilla-eq}.
\end{IEEEproof}

\begin{remark}[On the choice of subnormalized states]
\label{rem:subnorm-choice}
It is only this last step (i.e., proving \eqref{eq:ancilla-ineq-2}) in which we required the assumption of smoothing
over subnormalized states in the definition of the smooth $F$-min-relative entropy, rather than smoothing over normalized states.
\end{remark}

We used the following lemma in the proof of \cref{thm:data-processing}. 
\begin{lemma}
\label{lem:fid-overlap-0-lem}Let $\omega_{SE}$ be a bipartite PSD operator, and let $\sigma_{S}$ be a PSD
operator. Then
\begin{equation}
F(\omega_{SE},\sigma_{S}\otimes|0\rangle\!\langle0|_{E})=F(\omega_{S}%
^{0},\sigma_{S}),
\end{equation}
where $\omega_{S}^{0}:=\langle0|_{E}\omega_{SE}|0\rangle_{E}$.
\end{lemma}

\begin{IEEEproof}
Consider that%
\begin{align}
&  \sqrt{F}(\omega_{SE},\sigma_{S}\otimes|0\rangle\!\langle0|_{E})\nonumber\\
&  =\operatorname{Tr}\!\left[  \sqrt{\sqrt{\sigma_{S}\otimes|0\rangle
\!\langle0|_{E}}\omega_{SE}\sqrt{\sigma_{S}\otimes|0\rangle\!\langle0|_{E}}%
}\right]  \\
&  =\operatorname{Tr}\!\left[  \sqrt{\sqrt{\sigma_{S}}\otimes|0\rangle
\!\langle0|_{E}\omega_{SE}\sqrt{\sigma_{S}}\otimes|0\rangle\!\langle0|_{E}%
}\right]  \\
&  =\operatorname{Tr}\!\left[  \sqrt{\sqrt{\sigma_{S}}\langle0|_{E}\omega
_{SE}|0\rangle_{E}\sqrt{\sigma_{S}}\otimes|0\rangle\!\langle0|_{E}}\right]  \\
&  =\operatorname{Tr}\!\left[  \sqrt{\sqrt{\sigma_{S}}\langle0|_{E}\omega
_{SE}|0\rangle_{E}\sqrt{\sigma_{S}}}\otimes|0\rangle\!\langle0|_{E}\right]  \\
&  =\operatorname{Tr}\!\left[  \sqrt{\sqrt{\sigma_{S}}\langle0|_{E}\omega
_{SE}|0\rangle_{E}\sqrt{\sigma_{S}}}\right]  \\
&  =\sqrt{F}(\langle0|_{E}\omega_{SE}|0\rangle_{E},\sigma_{S})\\
&  =\sqrt{F}\!\left(  \omega_{S}^{0},\sigma_{S}\right)  ,
\end{align}
concluding the proof.
\end{IEEEproof}
\medskip

With Theorem~\ref{thm:data-processing} in hand, it thus follows that the smooth $F$-min-relative entropy is a particular kind of generalized divergence (here recall \eqref{eq:generalized-divergence}). Hence it possesses some basic properties satisfied by generalized divergences, as listed next. 
\begin{corollary}
  Let $\rho$  be a state and $\sigma$ a PSD operator. The smooth $F$-min-relative entropy satisfies the following properties:
 \begin{enumerate}
     \item Isometric invariance: For every isometry $V$,
     \begin{equation}
         D^\varepsilon_{\min,F}(\rho \Vert \sigma) = D^\varepsilon_{\min,F}( V\rho V^\dagger \Vert V \sigma V^\dagger).
     \end{equation}
     \item Stability: For every state $\tau$, 
     \begin{equation}
          D^\varepsilon_{\min,F}(\rho \Vert \sigma) = D^\varepsilon_{\min,F}(\rho \otimes \tau \Vert \sigma \otimes \tau).
     \end{equation}
 \end{enumerate}
\end{corollary}
The proof directly follows from  \cite[Proposition~7.14]{khatri2020principles}, along with \cref{thm:data-processing}.

\subsection{Other Properties}

In this section, we derive some other properties of the smooth $F$-min-relative entropy.

\begin{theorem} \label{thm:other-properties}
    For all $\varepsilon \in [0,1)$, every state $\rho$, and PSD operator $\sigma$, the smooth $F$-min-relative entropy $D_{\min,F}^{\varepsilon}(\rho\Vert\sigma)$ satisfies the following properties: 
\begin{enumerate}
    \item Scaling: For $c > 0$, we have 
    \begin{equation}
       D_{\min,F}^{\varepsilon}(\rho\Vert c\sigma) =D_{\min,F}^{\varepsilon}(\rho\Vert\sigma) + \log_{2}\!\left( \frac{1}{c}\right).
    \end{equation}
    
    \item Monotonicity: For $\varepsilon \leq  \varepsilon' \in [0,1)$, 
    \begin{equation}
        D_{\min,F}^{\varepsilon}(\rho\Vert\sigma) \leq D_{\min,F}^{\varepsilon'}(\rho\Vert\sigma).
    \end{equation}
    
    \item Superadditivity: For $\varepsilon_1, \varepsilon_2 \in[0,1)$, states $\rho_1$ and $\rho_2$,  and PSD operators $\sigma_1$ and $\sigma_2$, we have 
    \begin{multline}
      D_{\min,F}^{\varepsilon_1}(\rho_1\Vert\sigma_1)  + D_{\min,F}^{\varepsilon_2}(\rho_2\Vert\sigma_2) \\
      \leq  D_{\min,F}^{\varepsilon'}(\rho_1 \otimes \rho_2 \Vert\sigma_1 \otimes \sigma_2),
      \label{eq:superadd-smooth-f-min}
    \end{multline}
    where $\varepsilon' \coloneqq \varepsilon_1 + \varepsilon_2 - \varepsilon_1 \varepsilon_2$. 
    By monotonicity, we can also choose $\varepsilon'= \varepsilon_1 + \varepsilon_2$.
    
    \item Convexity in the second argument: Let $\{\sigma_i\}_i$ be a set of PSD operators,  and let $\{p_i\}_i$ be a probability distribution. Then
    \begin{equation}
        D_{\min,F}^{\varepsilon}(\rho  \Vert  \overline{\sigma})  \leq \sum_{i} p_i  D_{\min,F}^{\varepsilon}(\rho\Vert\sigma_i),
    \end{equation}
    where
    \begin{equation}
        \overline{\sigma} \coloneqq \sum_{i}  p_i \sigma_i.
    \end{equation}
    
    \item Non-negativity: For $\varepsilon \in[0,1)$, we have
    \begin{equation}
        D_{\min,F}^{\varepsilon}(\rho\Vert\sigma) \geq \log_{2}\!\left( \frac{1}{1-\varepsilon}\right) \geq 0,
    \end{equation}
    with the first inequality saturated if $\rho = \sigma$.
    
    \item If $0 \leq \sigma \leq \sigma'$, then 
    \begin{equation}
        D^\varepsilon_{\min,F}(\rho \Vert \sigma) \geq  D^\varepsilon_{\min,F}(\rho \Vert \sigma').
    \end{equation}

    \item Zero-error bounds: For $\varepsilon\in[0,1)$,
    \begin{equation}
        D_{\min,F}(\rho\Vert\sigma) \leq  D^\varepsilon_{\min,F}(\rho \Vert \sigma) 
        \label{eq:zero-err-lim-explicit-lower}
    \end{equation} 
    and for $\varepsilon \in [0,F(\rho,\widehat{\sigma})]$, with $\widehat{\sigma}\coloneqq \frac{\sigma}{\Tr[\sigma]}$,
    \begin{equation}
        D^\varepsilon_{\min,F}(\rho \Vert \sigma) 
        \leq 
           -\log_2\!\left[1-  g(\varepsilon,\rho,\sigma) \right]
 - \log_2 \Tr[\sigma],
 \label{eq:zero-err-lim-explicit-upper}
    \end{equation}
    where
    \begin{equation}
        g(\varepsilon,\rho,\sigma) \coloneqq\\
        \left( \sqrt{\varepsilon}\sqrt{F(\rho,\widehat{\sigma})}+\sqrt{1- F(\rho, \widehat{\sigma})}
\sqrt{1-\varepsilon} \right)^2.
    \end{equation}
    As such,
    \begin{equation} 
        \lim_{\varepsilon \to 0} D^\varepsilon_{\min,F}(\rho \Vert \sigma) = D_{\min,F}(\rho\Vert\sigma) .
        \label{eq:zero-err-lim}
    \end{equation}
\end{enumerate}
\end{theorem}
\begin{IEEEproof}
    
    \underline{Property 1:} Consider that
\begin{align}
& D_{\min,F}^{\varepsilon}(\rho\Vert c\sigma) \cr
& =\sup_{\widetilde{\rho}\in\mathcal{D}_{\leq}}\left\{  -\log_{2}%
F(\widetilde{\rho},c\sigma):F(\widetilde{\rho},\rho)\geq1-\varepsilon\right\}
\\
& =\sup_{\widetilde{\rho}\in\mathcal{D}_{\leq}}\left\{  -\log_{2}\left(
cF(\widetilde{\rho},\sigma)\right)  :F(\widetilde{\rho},\rho)\geq
1-\varepsilon\right\}  \\
& =\sup_{\widetilde{\rho}\in\mathcal{D}_{\leq}}\left\{  -\log_{2}%
F(\widetilde{\rho},\sigma)-\log_{2}c:F(\widetilde{\rho},\rho)\geq
1-\varepsilon\right\}  \\
& =\sup_{\widetilde{\rho}\in\mathcal{D}_{\leq}}\left\{  -\log_{2}%
F(\widetilde{\rho},\sigma):F(\widetilde{\rho},\rho)\geq1-\varepsilon\right\}
-\log_{2}c\\
& =D_{\min,F}^{\varepsilon}(\rho\Vert\sigma)-\log_{2}c.
\end{align}
This concludes the proof.

  \underline{Property 2:}
  Let $\rho' \in \cD_{\leq}$ satisfy $F(\rho,\rho') \geq 1-\varepsilon$. Then  $F(\rho,\rho') \geq 1-\varepsilon'$ since $\varepsilon' \geq \varepsilon$.
This leads to 
\begin{equation}
    -\log_{2}F(\rho',\sigma) \leq D_{\min}^{\varepsilon'}(\rho\Vert\sigma).
\end{equation}
The above inequality holds for every $\rho' \in \cD_{\leq}$ such that $F(\rho,\rho') \geq 1-\varepsilon$. Thus, we have
\begin{equation}
   D_{\min,F}^{\varepsilon}(\rho\Vert\sigma) \leq D_{\min,F}^{\varepsilon'}(\rho\Vert\sigma).
\end{equation}

\underline{Property 3:} In this proof, we use the multiplicativity of fidelity with respect to tensor products:
\begin{equation}
    F(\rho_1 \otimes \rho_2, \sigma_1 \otimes \sigma_2)= F(\rho_1, \sigma_1) F(\rho_2, \sigma_2).
\end{equation}
Let $\rho'_1, \rho'_2 \in \cD_{\leq}$ satisfy 
$F(\rho_1,\rho'_1) \geq 1-\varepsilon_1$ and $F(\rho_2,\rho'_2) \geq 1-\varepsilon_2$. 
 Consider that
\begin{align}
    & \!\!\!\!\!\!\!\!
    - \log_{2}  F(\rho'_1 , \sigma_1 ) -\log_{2}  F( \rho'_2, \sigma_2) \label{eq:start-fid-add-prod}\\
     &= -\log_{2}  F(\rho'_1 \otimes \rho'_2, \sigma_1 \otimes \sigma_2) \\
    & \leq D_{\min}^{\varepsilon'}(\rho_1 \otimes \rho_2 \Vert\sigma_1 \otimes \sigma_2),
\end{align}
where the last inequality follows because 
\begin{align}
    F(\rho_1' \otimes \rho_2', \rho_1 \otimes \rho_2) & = F(\rho_1' , \rho_1 ) \cdot F(\rho_2' , \rho_2 )\\
    & \geq  
    (1-\varepsilon_1) (1-\varepsilon_2) \\
    & = 1-\varepsilon',
\end{align}
with $\varepsilon'$ as stated just after \eqref{eq:superadd-smooth-f-min}.
Then, supremizing \eqref{eq:start-fid-add-prod} over $\rho'_1$ and~$\rho'_2$ satisfying $F(\rho_1', \rho_1) \geq 1- \varepsilon_1$ and $F(\rho_2', \rho_2) \geq 1- \varepsilon_2$, respectively, we arrive at the  desired inequality in \eqref{eq:superadd-smooth-f-min}. 

\underline{Property 4:}
 Let $\rho' \in \cD_{\leq}$ satisfy $F(\rho,\rho') \geq 1-\varepsilon$.
 By concavity of fidelity \cite[Theorem~3.60]{khatri2020principles} (while rescaling for a subnormalized state $\rho'$), we have that
 \begin{equation}
     F(\rho', \overline{\sigma} ) \geq \sum_{i} p_i F(\rho',\sigma_i).
     \label{eq:fid-concave}
 \end{equation}
 Then consider that
 \begin{align}
     -\log_{2}  F(\rho', \overline{\sigma} ) &
     \leq -\log_{2} \!\left[ \sum_{i} p_i F\!\left(\rho',  \sigma_i \right)\right] \\
     & \leq \sum_{i} p_i \left[-\log_{2}  F\!\left(\rho',  \sigma_i \right)\right]  \\
     &\leq  \sum_{i} p_i  D_{\min,F}^{\varepsilon}(\rho\Vert\sigma_i) ,
 \end{align}
where the first inequality follows from \eqref{eq:fid-concave} and monotonicity of $-\log_2$, the second inequality from the convexity of $-\log_{2}$, and the last due to $F(\rho,\rho') \geq 1-\varepsilon$.
Lastly, by optimizing over all $\rho'$ satisfying the required condition, we conclude the proof. 

\underline{Property 5:} By the data-processing inequality derived in \cref{thm:data-processing} and choosing the quantum channel $\cN(X)= \Tr[X] \omega$, where $X$ is a linear operator, and $\omega$ is a quantum state, we have 
\begin{equation} \label{eq: negativity data proce}
    D_{\min,F}^{\varepsilon}(\rho\Vert\sigma) \geq  D_{\min,F}^{\varepsilon}(\omega\Vert\omega).
\end{equation}
Then with the constraint $F(\widetilde{\rho},\omega)=1-\varepsilon$ (recall Remark~\ref{rem:ineq-to-eq-constraint}), it follows that $D_{\min,F}^{\varepsilon}(\omega\Vert\omega)=-\log_{2}(1-\varepsilon)$ for every state~$\omega$. Combining that with \eqref{eq: negativity data proce} concludes the proof.

\underline{Property 6:}
Let $\widetilde{\rho} \in \mathcal{D}_{\leq}$ be such that $F(\widetilde{\rho}, \rho) \geq 1- \varepsilon$. By  \cite[Proposition~7.33]{khatri2020principles}, $\widetilde{D}_{1/2}(\rho \Vert \sigma) \geq \widetilde{D}_{1/2}(\rho \Vert \sigma') $ for $\sigma \leq \sigma'$. Then, by \eqref{eq:alpha-1/2-D_min,F},  we have that
\begin{equation*}
    -\log_2 F(\widetilde{\rho}, \sigma) \geq  -\log_2 F(\widetilde{\rho}, \sigma').
\end{equation*}
Supremizing over $\widetilde{\rho} \in \mathcal{D}_{\leq}$ such that $F(\widetilde{\rho}, \rho) \geq 1- \varepsilon$, we arrive at the desired inequality. 

\underline{Property 7:}
We first prove \eqref{eq:zero-err-lim-explicit-lower}. For all $\varepsilon \in[0,1)$, we choose $\widetilde{\rho} = \rho$ and thus have that $F(\widetilde{\rho}, \rho)=1 \geq 1-\varepsilon$. This leads to 
\begin{equation}
    -\log_2 F(\rho, \sigma) \leq D^\varepsilon_{\min,F}(\rho \Vert \sigma).
\end{equation}

Now we prove \eqref{eq:zero-err-lim-explicit-upper} and first do so in the case that $\sigma$ is a state. Let $\rho_1 \in \cD_{\leq}$ satisfy $F({\rho_1}, \rho) =1-\varepsilon$. Now construct the following normalized states:
\begin{align}
     \rho'_1 & \coloneqq \rho_1 \oplus (1-\Tr[\rho_1]), \\
    \rho' & \coloneqq \rho \oplus 0 ,\\
     \sigma' & \coloneqq \sigma \oplus 0,
\end{align}
so that
\begin{align}
     F(\rho'_1,\sigma') &= F(\rho_1,\sigma), \\
     F(\rho'_1,\rho') &= F(\rho_1,\rho), \\
     F(\rho',\sigma') &= F(\rho,\sigma).
\end{align}
By assumption,  $\varepsilon$ satisfies $0 \leq \varepsilon \leq F(\rho, \sigma)$, which implies that $\varepsilon+(1-F(\rho,\sigma)) \leq 1$.
We can thus apply 
the refined triangular inequality for the sine distance
$\sqrt{1-F}$ \cite[Proposition~3.16]{tomamichel2015quantum} to find that
\begin{align}
  & \sqrt{1-F({\rho_1},\sigma)}\notag \\
  & \leq\sqrt{1-F({\rho_1},\rho)}\sqrt{F(\rho,\sigma)}\nonumber\\
&  \qquad+\sqrt{1-F(\rho, \sigma)}\sqrt{F({\rho_1},\rho)}\\
&  =\sqrt{\varepsilon}\sqrt{F(\rho,\sigma)}+\sqrt{1- F(\rho, \sigma)}%
\sqrt{1-\varepsilon}.
\end{align} 
After an algebraic manipulation of the inequality above, we arrive at 
\begin{multline}
    -\log_2 F({\rho_1}, \sigma)   \leq \\
     -\log_2\!\left(1- \left( \sqrt{\varepsilon}\sqrt{F(\rho,\sigma)}+\sqrt{1- F(\rho, \sigma)}%
\sqrt{1-\varepsilon} \right)^2  \right).
\end{multline}
This inequality holds for all ${\rho_1}$ satisfying $F({\rho_1}, \rho) =1-\varepsilon$. Then supremizing over all such candidates belonging to $\cD_{\leq}$, 
we conclude that 
\begin{multline}
    D^\varepsilon_{\min,F} (\rho \Vert \sigma)  \leq \\
     -\log_2\!\left(1- \left( \sqrt{\varepsilon}\sqrt{F(\rho,\sigma)}+\sqrt{1- F(\rho, \sigma)}%
\sqrt{1-\varepsilon} \right)^2  \right).
\label{eq:zero-err-lim-explicit-upper-state}
\end{multline}
This concludes the proof of \eqref{eq:zero-err-lim-explicit-upper} when $\sigma$ is a state.

When $\sigma$ is a general PSD operator, we can write $\sigma= c \left(\frac{1}{c} \sigma \right)$ with $c \coloneqq \Tr[\sigma]$, so that $ \frac{1}{c} \sigma $ is a state. We then apply \eqref{eq:zero-err-lim-explicit-upper-state} to $\frac{1}{c} \sigma$ and use the scaling property (Property~1 of \cref{thm:other-properties}) to conclude \eqref{eq:zero-err-lim-explicit-upper} in the general case.

To conclude \eqref{eq:zero-err-lim}, we take the following limits of \eqref{eq:zero-err-lim-explicit-lower} and \eqref{eq:zero-err-lim-explicit-upper}, respectively:
\begin{align}\label{eq: liminf}
    D_{\min,F}(\rho\Vert \sigma) & \leq \liminf_{\varepsilon \to 0}  \ D^\varepsilon_{\min,F}(\rho \Vert \sigma), \\
 \label{eq: limsup}
    \limsup_{\varepsilon \to 0} D^\varepsilon_{\min,F} (\rho \Vert \sigma) & \leq   -\log_2 F\!\left({\rho}, \frac{\sigma}{c}\right) - \log_2 c\\
    & =   -\log_2 F({\rho}, \sigma) \\
    & =  D_{\min,F}(\rho\Vert \sigma). 
\end{align}
This concludes the proof.
\end{IEEEproof}

\section{Relating Smooth \texorpdfstring{$F$}{F}-Min-Relative Entropy to Other Distinguishability Measures} \label{Sec:Connections}

In this section, we establish connections between the smooth $F$-min-relative entropy and other information-theoretic quantities such as the sandwiched R\'enyi relative entropy, its smooth versions, as well as the smooth max- and min-relative entropies. 
Then, we utilize these connections  in subsequent sections to establish the second-order asymptotics  of the smooth $F$-min-relative entropy. 

\subsection{Relation to Sandwiched R\'enyi Relative Entropy and its Smoothed Variants}

The smoothed sandwiched R\'enyi relative entropies were defined recently  in \cite{RT22}, and we recall their definitions here.

\begin{definition} Let $\rho$ be a state,  and let $\sigma$ be a PSD operator. Fix  $\varepsilon \in [0,1]$ and $\alpha \in (0,1) \cup (1,\infty)$. The smooth sandwiched R\'enyi relative entropy is defined for $\alpha > 1$ as 
   \begin{equation}
   \label{eq:smooth-Frenyi-def}
 \widetilde{D}_{\alpha}^{\varepsilon}(\rho\Vert\sigma) \coloneqq 
   \inf_{\widetilde
{\rho}\in\mathcal{D}_{\leq}}\left\{ \widetilde{D}_{\alpha}(\widetilde{\rho}\Vert\sigma) :F(\widetilde
{\rho},\rho)\geq1-\varepsilon\right\}.
\end{equation} 
and for $\alpha \in (0, 1)$ as
 \begin{equation}
   \label{eq:smooth-Frenyi-def alpha less than 1}
 \widetilde{D}_{\alpha}^{\varepsilon}(\rho\Vert\sigma) \coloneqq 
   \sup_{\widetilde
{\rho}\in\mathcal{D}_{\leq}}\left\{ \widetilde{D}_{\alpha}(\widetilde{\rho}\Vert\sigma) :F(\widetilde
{\rho},\rho)\geq1-\varepsilon\right\}.
\end{equation} 
In the above definitions, we use precisely the mathematical expression in \eqref{eq:sandwiched-renyi-def} for evaluating $\widetilde{D}_{\alpha}(\widetilde{\rho}\Vert\sigma)$, even though we only defined it in \eqref{eq:sandwiched-renyi-def} for (normalized) states.
\end{definition}

{Note that the smooth sandwiched R\'enyi relative entropy for $\alpha=1/2$ is equivalent to the smooth $F$-min relative entropy (recall \cref{def:smooth-min-relative-fidelity}).}

\begin{remark}[Inequality constraint in smooth sandwiched R\'enyi relative entropy definition]
\label{rem:ineq-to-eq-smooth-sand}
Note that the  smooth sandwiched R\'enyi relative entropy can be rewritten for $\alpha > 1$ as 
\begin{equation}
\widetilde{D}_{\alpha}^{\varepsilon}(\rho\Vert\sigma) \coloneqq 
   \inf_{\widetilde
{\rho}\in\mathcal{D}_{\leq}}\left\{ \widetilde{D}_{\alpha}(\widetilde{\rho}\Vert\sigma) :F(\widetilde
{\rho},\rho)= 1-\varepsilon\right\},
\end{equation}
and for $\alpha \in (0,1)$ as
\begin{equation}
\widetilde{D}_{\alpha}^{\varepsilon}(\rho\Vert\sigma) \coloneqq 
   \sup_{\widetilde
{\rho}\in\mathcal{D}_{\leq}}\left\{ \widetilde{D}_{\alpha}(\widetilde{\rho}\Vert\sigma) :F(\widetilde
{\rho},\rho)= 1-\varepsilon\right\},
\end{equation}
Indeed, let us first consider when $\alpha > 1$.
If $\widetilde{\rho}$ is such that $F(\widetilde{\rho},\rho
)>1-\varepsilon$, then we can set $c=\left(  1-\varepsilon\right)
/F(\widetilde{\rho},\rho)\in\left(  0,1\right)  $ such that $F(\rho
^{\prime},\rho)=1-\varepsilon$, where $\rho^{\prime}%
=c\widetilde{\rho}$ and $\rho^{\prime}\in\mathcal{D}_{\leq}$.
Furthermore, we also have that $\widetilde{D}_{\alpha}(\widetilde{\rho}\Vert\sigma)>\widetilde{D}_{\alpha}(\rho^{\prime}\Vert\sigma)$, so that the objective function only
decreases under this change. The statement for $\alpha \in (0,1)$ follows from a similar argument.
\end{remark}

\begin{theorem}
\label{thm:connection-to-smooth-renyi}
Let $\rho$ be a state and $\sigma$ a positive
semi-definite operator. Let $\varepsilon_{1},\varepsilon_{2}\in\left[
0,1\right]  $ be such that $\varepsilon_{1}+\varepsilon_{2}\leq1$, and let 
\begin{equation}
\varepsilon^{\prime}\coloneqq\left[  \sqrt{\varepsilon_{1}}\sqrt
{1-\varepsilon_{2}}+\sqrt{1-\varepsilon_{1}} \sqrt{\varepsilon_{2}}\right]
^{2},
\end{equation}
so that $\varepsilon^{\prime}\in\left[  0,1\right]  $. Then for $\alpha \in (1/2,1)$ and $ \beta = \frac{\alpha}{2\alpha - 1} >1$, we have that
\begin{equation}
\label{eq:smooth-sand-ineq}
\widetilde{D}_{\beta}^{\varepsilon_{1}}(\rho\Vert\sigma)+\frac{\beta}{\beta-1}\log_{2}\!\left(  \frac
{1}{1-\varepsilon^{\prime}}\right)  \geq \widetilde{D}_{\alpha}^{\varepsilon_{2}}(\rho
\Vert\sigma).
\end{equation}

\end{theorem}
\begin{IEEEproof}
    Let $\rho_{1}$ be optimal for $\widetilde{D}_{\beta}^{\varepsilon_{1}}(\rho\Vert\sigma)$,
and let $\rho_{2}$ be optimal for $\widetilde{D}_{\alpha}^{\varepsilon_{2}}(\rho\Vert
\sigma)$. Then it follows from Remark~\ref{rem:ineq-to-eq-smooth-sand} that
\begin{equation}
F(\rho_{i},\rho)=1-\varepsilon_{i},
\end{equation}
for $i\in\left\{  1,2\right\}  $.
Let
\begin{align}
     \rho'_1 & \coloneqq \rho_1 \oplus (1-\Tr[\rho_1] ) \oplus 0,\\
     \rho'_2 & \coloneqq \rho_2 \oplus 0 \oplus  (1-\Tr[\rho_2] ),\\
     \rho' & \coloneqq \rho \oplus 0 \oplus 0.
\end{align}
Note that $\rho'_1$ and $ \rho'_2$ are normalized states satisfying the equalities
\begin{align}
    F(\rho'_1,\rho'_2) & = F(\rho_1,\rho_2), \\
     F(\rho'_i,\rho') & = F(\rho_i,\rho),
\end{align}
for $i \in \{1,2\}$.
Then, by applying the refined triangular inequality for the sine distance $\sqrt{1-F}$ of normalized states \cite[Proposition~3.16]{tomamichel2015quantum}, along with the assumption that $\varepsilon_{1}+\varepsilon_{2}\leq1$), 
we arrive at 
\begin{align}
 & \sqrt{1-F(\rho_{1},\rho_{2})} \notag \\ & \leq\sqrt{1-F(\rho_{1},\rho)}\sqrt{F(\rho,\rho_{2})} \nonumber \\
 & \qquad +\sqrt{1-F(\rho,\rho_{2})}\sqrt{F(\rho_{1},\rho)} \\
&  =\sqrt{\varepsilon_{1}}\sqrt{1-\varepsilon_{2}}+\sqrt{\varepsilon_{2}}%
\sqrt{1-\varepsilon_{1}}.
\end{align}
Then it follows that%
\begin{align}
F(\rho_{1},\rho_{2}) &  \geq1-\left[  \sqrt{\varepsilon_{1}}\sqrt
{1-\varepsilon_{2}}+\sqrt{\varepsilon_{2}}\sqrt{1-\varepsilon_{1}}\right]
^{2}\\
&  =1-\varepsilon^{\prime} \label{eq: relation between rho 1 and 2}.
\end{align}

To arrive at the desired inequality in \eqref{eq:smooth-sand-ineq}, let us recall the following inequality from  \cite[Lemma~1]{wang2019resource}. Let
$\rho_{1}$ and $\rho_{2}$ be subnormalized states, and let $\sigma$ be a positive
semi-definite operator such that $\operatorname{supp}(\rho_{1})\subseteq
\operatorname{supp}(\sigma)$. For $\alpha\in(1/2,1)$ and $\beta=\alpha
/(2\alpha-1)>1$,%
\begin{align}
\widetilde{D}_{\beta}(\rho_{1}\Vert\sigma)-\widetilde{D}_{\alpha}(\rho
_{2}\Vert\sigma)  &  \geq\frac{\alpha}{1-\alpha}\log_{2}F(\rho_{1},\rho_{2})\\
&  =\frac{\beta}{\beta-1}\log_{2}F(\rho_{1},\rho_{2}).
\end{align}
We note here that the proof of \cite[Lemma~1]{wang2019resource} was only given therein for states, but it is clear by inspection that the same proof holds for subnormalized states.
This implies that
\begin{align}
\widetilde{D}_{\beta}(\rho_{1}\Vert\sigma)+\frac{\beta}{\beta-1}\log_{2}
\!\left(  \frac{1}{F(\rho_{1},\rho_{2})}\right)   &  \geq\widetilde
{D}_{\alpha}(\rho_{2}\Vert\sigma)\label{eq:take-limit-inf-1}.
\end{align}
Then, by the inequality in \eqref{eq: relation between rho 1 and 2}, we have 
\begin{equation}
    \widetilde{D}_{\beta}(\rho_{1}\Vert\sigma)+\frac{\beta}{\beta-1}\log_{2}
\left(  \frac{1}{1-\varepsilon'}\right)  \geq \widetilde{D}_{\alpha}(\rho_{2}\Vert\sigma).
\label{eq:final-step-pseudo-cont-smooth}
\end{equation}
Lastly, we conclude the proof by noting the assumption that  $\rho_{1}$ and $ \rho_{2}$ are optimal for $\widetilde{D}_{\beta}^{\varepsilon_{1}}(\rho\Vert\sigma)$
and $\widetilde{D}_{\alpha}^{\varepsilon_{2}}(\rho\Vert
\sigma)$, respectively.
\end{IEEEproof}

\medskip 
Applying \cref{thm:connection-to-smooth-renyi} and the limits $\varepsilon_1 \to 0$ and $\alpha \to 1/2$ (or alternatively directly from \eqref{eq:final-step-pseudo-cont-smooth} with $\varepsilon_1 =0$ and taking the limit $\alpha \to 1/2$ while employing the $\alpha$-monotonicity of~$\widetilde{D}_{\alpha}$), we arrive at the following corollary: 

\begin{corollary}\label{prop: Dmin and smooth renyi}
For all $\varepsilon\in[0,1)$, every state $\rho$, PSD
operator~$\sigma$, and $\beta>1$, the following inequality holds%
\begin{equation}
\widetilde{D}_{\beta}(\rho\Vert\sigma)+\frac{\beta}{\beta-1}\log_{2}
\!\left(  \frac{1}{1-\varepsilon}\right)  \geq D_{\min,F}^{\varepsilon}%
(\rho\Vert\sigma).
\end{equation}
    
\end{corollary}

By applying \eqref{eq: dmax infity} and \cref{prop: Dmin and smooth renyi}, we arrive at the following inequality:
\begin{equation}
D_{\max}(\rho\Vert\sigma)+\log_{2}\!\left(  \frac{1}{1-\varepsilon}\right)
\geq D_{\min,F}^{\varepsilon}(\rho\Vert\sigma).
\end{equation}

We can also arrive at an inequality relating the smooth $F$-min relative entropy to the smooth max-relative entropy. 
By taking the limits $\beta \to \infty$ and $\alpha \to 1/2$ in \cref{thm:connection-to-smooth-renyi}, while employing the $\alpha$-monotonicity of $\widetilde{D}_{\beta}$, we obtain the following inequality: 
\begin{corollary}
    \label{prop:connect-smooth-max}
    Let $\rho$ be a state and $\sigma$ a PSD operator. Let $\varepsilon_{1},\varepsilon_{2}\in\left[
0,1\right]  $ be such that $\varepsilon_{1}+\varepsilon_{2}\leq1$, and let 
\begin{equation}
\varepsilon^{\prime}\coloneqq\left[  \sqrt{\varepsilon_{1}}\sqrt
{1-\varepsilon_{2}}+\sqrt{1-\varepsilon_{1}} \sqrt{\varepsilon_{2}}\right]
^{2},
\end{equation}
so that $\varepsilon^{\prime}\in\left[  0,1\right]  $. Then
\begin{equation}
{D}_{\max}^{\varepsilon_{1}}(\rho\Vert\sigma)+\log_{2}\!\left(  \frac
{1}{1-\varepsilon^{\prime}}\right)  \geq D_{\min,F}^{\varepsilon_{2}}(\rho
\Vert\sigma).
\label{eq:smooth-min-max-ineq}
\end{equation}
\end{corollary}

We note here that the inequality in \eqref{eq:smooth-min-max-ineq}
 appeared in \cite[Lemma~III.8]{ramakrishnan2023moderate} via a different proof strategy. This kind of inequality was interpreted in a resource-theoretic manner in \cite{wang2019resource,wang19channels,W21second,takagi2022one}---it would thus be interesting if it were possible to do so for \eqref{eq:smooth-sand-ineq} and \eqref{eq:smooth-min-max-ineq}.

\subsection{Relation to Smooth Min-Relative Entropy}

\begin{definition}
Recall that the smooth min-relative entropy is defined for
$\varepsilon\in\left[  0,1\right]  $, a state $\rho$, and a PSD operator $\sigma$ as%
\begin{equation}
D_{\min}^{\varepsilon}(\rho\Vert\sigma)\coloneqq-\log_{2}\inf_{\Lambda\geq
0}\left\{  \operatorname{Tr}[\Lambda\sigma]:\operatorname{Tr}[\Lambda\rho
]\geq1-\varepsilon,\Lambda\leq I\right\}  .
\end{equation}
As mentioned before, it is also known as the hypothesis testing relative entropy.
\end{definition}

\begin{proposition}
\label{prop:connect-hypo}For every $\varepsilon\in(0,1)$, state $\rho$, and
PSD $\sigma$, the following inequality holds%
\begin{equation}
D_{\min}^{\varepsilon}(\rho\Vert\sigma)\leq D_{\min,F}^{\varepsilon}(\rho
\Vert\sigma)+\log_{2}\!\left(  \frac{1}{1-\varepsilon}\right)  .
\end{equation}

\end{proposition}

\begin{IEEEproof}
Let $\Lambda$ be an arbitrary measurement operator satisfying
$\operatorname{Tr}[\Lambda\rho]\geq1-\varepsilon$. By the gentle measurement
lemma from \cite[Eq.~(9.202)]{Wbook17}, we know that $\operatorname{Tr}[\Lambda\rho]\geq1-\varepsilon$ implies
that%
\begin{equation}
F(\widetilde{\rho},\rho)\geq1-\varepsilon,
\label{eq:GML-fid}
\end{equation}
where $\widetilde{\rho}=\frac{1}{\operatorname{Tr}[\Lambda\rho]}\sqrt{\Lambda
}\rho\sqrt{\Lambda}$. 
An alternative way to see the inequality in
\eqref{eq:GML-fid}\ is by the proof in \cref{lem:gentle-meas-alt-proof} in Appendix~\ref{app:supp-lems}.

Now we should relate $\operatorname{Tr}[\Lambda\sigma]$
to $F(\widetilde{\rho},\sigma)$. Consider from \cref{lem:alt-fid-flip-Lambda} in Appendix~\ref{app:supp-lems} that
\begin{align}
  F(\widetilde{\rho},\sigma) & =  \frac{1}{\Tr[\Lambda \rho]}F(\rho,\sqrt{\Lambda}\sigma\sqrt{\Lambda}) \\
&  \leq \frac{1}{1-\varepsilon}F(\rho,\sqrt{\Lambda}\sigma\sqrt{\Lambda})\\
&  \leq\frac{1}{1-\varepsilon}\operatorname{Tr}[\Lambda\sigma].
\end{align}
The last inequality follows from data processing for fidelity under the trace
channel. Then%
\begin{align}
-\log_{2}\operatorname{Tr}[\Lambda\sigma] &  \leq-\log_{2}\!\left(
F(\widetilde{\rho},\sigma)\left(  1-\varepsilon\right)  \right)  \\
&  =-\log_{2}F(\widetilde{\rho},\sigma)+\log_{2}\!\left(  \frac{1}%
{1-\varepsilon}\right)  \\
&  \leq D_{\min,F}^{\varepsilon}(\rho\Vert\sigma)+\log_{2}\!\left(  \frac
{1}{1-\varepsilon}\right)  .
\end{align}
Since this holds for an arbitrary measurement operator satisfying
$\operatorname{Tr}[\Lambda\rho]\geq1-\varepsilon$, we conclude that%
\begin{equation}
D_{\min}^{\varepsilon}(\rho\Vert\sigma)\leq D_{\min,F}^{\varepsilon}(\rho
\Vert\sigma)+\log_{2}\!\left(  \frac{1}{1-\varepsilon}\right)  ,
\end{equation}
which is the desired statement.
\end{IEEEproof}

\medskip 

The following was established in \cite{zhao2019oneDmin}, but we give an alternative proof for it.

\begin{proposition}
\label{prop:connect-hypo-alt}For every $\varepsilon\in(0,1)$, state $\rho$, and PSD operator $\sigma$, the following inequality holds%
\begin{equation}
D_{\min}^{\varepsilon}(\rho\Vert\sigma)\leq D_{\min,F}^{\varepsilon(2-\varepsilon)}(\rho
\Vert\sigma)  .
\end{equation}

\end{proposition}
\begin{IEEEproof}
The proof is similar to the proof of Proposition~\ref{prop:connect-hypo}. See Appendix~\ref{app:alt-proof-min-hypo-bnd} for details.
\end{IEEEproof}

\medskip 
We note here that \cref{prop:connect-hypo} derived in this work is useful in obtaining the second-order asymptotics presented in \cref{Sec:Second-order-asymptotic}.

\begin{remark}[Lower bound for smooth $F$-min-relative entropy via Petz--R\'enyi relative entropy]
    For $\alpha\in(0,1)$, by \cite[Proposition~3]{qi2018applications}, we have 
\begin{equation}
    D^\varepsilon_{\min} (\rho \Vert \sigma) \geq \frac{\alpha}{\alpha-1} \log_2 \! \left( \frac{1}{\varepsilon}\right) + D_\alpha(\rho \Vert \sigma).
\end{equation}
Then by  \cref{prop:connect-hypo}, we find that 
\begin{equation}
    D^\varepsilon_{\min,F}(\rho \Vert \sigma) + \log_2 \! \left( \frac{1}{1-\varepsilon}\right) \geq D_\alpha(\rho \Vert \sigma) + \frac{\alpha}{\alpha-1} \log_2 \! \left( \frac{1}{\varepsilon}\right) .
\end{equation}
It is an open question to establish a tighter lower bound on the smooth $F$-min-relative entropy in terms of the Petz--R\'enyi relative entropy of order $\alpha\in(0,1)$.
\end{remark}

\section{Second-Order Asymptotics} \label{Sec:Second-order-asymptotic}

In this section, we establish the second-order asymptotics of the smooth $F$-min-relative entropy (\cref{thm:second-order-Dmin}), as well as for the smooth sandwiched R\'enyi relative entropy (\cref{cor:second-order-smooth-sandwiched-Renyi}). Before doing so, let us first define a set of quantities that are needed in what follows. 
The cumulative distribution function of a standard normal random variable and its inverse are respectively given by 
\begin{equation}
    \Phi(a) \coloneqq \frac{1}{\sqrt{2 \pi}} \int_{-\infty}^a dx \ \mathrm{exp}\left( \frac{-x^2}{2}\right),
\end{equation}
\begin{equation}
    \Phi^{-1}(\varepsilon) \coloneqq \sup \{ a \in \RR \  | \ \Phi(a) \leq \varepsilon \}.
\end{equation}
For $\varepsilon \in(0,1)$, recall that 
\begin{equation}
    \Phi^{-1}(1-\varepsilon) = - \Phi^{-1}(\varepsilon),
\end{equation}
and note that $\Phi^{-1}(\varepsilon) < 0$ for $\varepsilon < 1/2$ and $\Phi^{-1}(\varepsilon) > 0 $ for $\varepsilon > 1/2$.
For a state  $\rho$ and PSD operator $\sigma$ such that $\supp(\rho)\subseteq\supp(\sigma)$ and $V(\rho \Vert \sigma) >0$, the following second-order expansion is known:
\begin{multline} \label{eq:hypothesis-testing-second-order}
    \frac{1}{n} D^\varepsilon_{\min}\!\left(\rho^{\otimes n} \Vert \sigma^{\otimes n }\right) =\\
    D(\rho \Vert \sigma) + \sqrt{\frac{1}{n} V(\rho \Vert \sigma)} \ \Phi^{-1}(\varepsilon) + O\!\left( \frac{\log n}{n} \right).
\end{multline}
This expansion gives a refined understanding of Stein's lemma for asymmetric hypothesis testing \cite{tomamichel2013hierarchy,li2014second} and has been useful in various developments toward establishing second-order asymptotic characterizations of information-theoretic tasks (see, e.g., \cite{datta2014second,DTW14,wilde2017position}).
For the finite-dimensional scenario, \eqref{eq:hypothesis-testing-second-order} was proven in \cite{tomamichel2013hierarchy,li2014second}. For a state $\rho$ and PSD trace-class $\sigma$ acting on a separable Hilbert space, the inequality $\leq $ was established in \cite{datta2016second, kaur2017upper},
while the inequality $\geq$ was shown in \cite{li2014second, datta2016second, khabbazi2019union}.

\begin{theorem} \label{thm:second-order-Dmin}
     For a state  $\rho$, a PSD operator $\sigma$, and $\varepsilon\in(0,1)$, such that $\supp(\rho)\subseteq\supp(\sigma)$ and $V(\rho \Vert \sigma) >0$, the following second-order expansion holds:
\begin{multline}
    \frac{1}{n} D_{\min,F}^{\varepsilon}\!\left(\rho^{\otimes n} \Vert\sigma^{\otimes n}\right) = \\  D(\rho \Vert \sigma) + \sqrt{\frac{1}{n} V(\rho \Vert \sigma)} \ \Phi^{-1}(\varepsilon) + O\!\left( \frac{\log n}{n} \right).
\end{multline}

\end{theorem}

\begin{IEEEproof}
    For the lower bound, we apply  \cref{prop:connect-hypo} to find that
    \begin{align}
  & \frac{1}{n}D^\varepsilon_{\min,F} (\rho^{\otimes n} \Vert \sigma^{\otimes n})     \\
  & \geq \frac{1}{n} D^\varepsilon_{\min}(\rho^{\otimes n} \Vert \sigma^{\otimes n})
  - \frac{1}{n}\log_{2} \!\left(\frac{1}{1-\varepsilon}\right) \\ 
  & = D(\rho \Vert \sigma) + \sqrt{\frac{1}{n} V(\rho\Vert\sigma)} \ \Phi^{-1}(\varepsilon) + O\!\left( \frac{\log n}{n} \right)
  \label{eq:lower-bound-second-order-D-min},
\end{align}
where the last inequality follows from \eqref{eq:hypothesis-testing-second-order}.

For the upper bound, we use \cref{prop:connect-smooth-max}  for $\varepsilon
,\delta\in\left(  0,1\right)  $ such that $\varepsilon+\delta\in\left(
0,1\right)$ to find that 
\begin{multline}
    \frac{1}{n}D^\varepsilon_{\min,F} (\rho^{\otimes n}  \Vert \sigma^{\otimes n}) \\ 
    \leq \frac{1}{n} D_{\max}^{1-\varepsilon-\delta}(\rho^{\otimes n}\Vert\sigma^{\otimes n})+ \frac{1}{n}\log_{2} \!\left(
\frac{1}{1-f(\varepsilon,\delta)}\right), \label{eq:smooth-max-and-min-n-tensor}
\end{multline}
where
\begin{equation} \label{eq:f(eps,delta)}
f(\varepsilon,\delta):=\left[  \sqrt{\varepsilon}\sqrt{\varepsilon+\delta
}+\sqrt{1-\varepsilon-\delta}\sqrt{1-\varepsilon}\right]  ^{2},
\end{equation}
so that $f(\varepsilon,\delta)\in\left(  0,1\right)  $. Indeed, $f(\varepsilon
,\delta)$ can be understood as a classical fidelity of two binary random variables with parameters $\varepsilon$ and $\varepsilon+\delta$, which we know is $\leq1$.
We also note that
\begin{equation}
\frac{1}{1-f(\varepsilon,\delta)}=\frac{4\varepsilon\left(  1-\varepsilon
\right)  }{\delta^{2}}+\frac{2\left(  1-2\varepsilon\right)  }{\delta}%
-\frac{1}{4\varepsilon\left(  1-\varepsilon\right)  }+O\left(  \delta\right)
.
\end{equation}
Thus, when we choose $\delta=1/\sqrt{n}$ (and $n$ sufficiently large so that $\varepsilon + \delta < 1$) for the second-order expansion, we find
that%
\begin{equation} \label{eq: f function taylor}
\log_{2}\!\left(  \frac{1}{1-f(\varepsilon,\delta)}\right)  =O(\log n).
\end{equation}
Recall from \cite[Theorem~4]{ABJT19} that 
\begin{equation}
    D_{\max}^{\varepsilon}(\rho\Vert\sigma)  \leq D_{\min}^{1-\varepsilon}(\rho \Vert \sigma) + \log_{2}\!\left(\frac{1}{1-\varepsilon}\right). \label{eq:smooth-max-to-hypothesis}
\end{equation}
Note that the choice of $\varepsilon$ instead of $\sqrt{\varepsilon}$ as in \cite[Theorem~4]{ABJT19} is due to the fact that the smooth max-relative entropy is defined in \cite{ABJT19} in terms of the sine distance $\sqrt{1-F(\widetilde{\rho},\rho)}$  \cite{R02,R03,GLN04,R06}.
Combining the above inequalities together with \eqref{eq:smooth-max-and-min-n-tensor}, we arrive at 
\begin{align}
 &  \frac{1}{n}D^\varepsilon_{\min,F} (\rho^{\otimes n}  \Vert \sigma^{\otimes n}) \\
   & \stackrel{(a)} \leq \frac{1}{n} D_{\max}^{1-\varepsilon-\delta}(\rho^{\otimes n}\Vert\sigma^{\otimes n})+ \frac{1}{n}\log_{2} \!\left(
\frac{1}{1-f(\varepsilon,\delta)}\right)   \\
& \stackrel{(b)} \leq \frac{1}{n} D^{\varepsilon+ \delta}_{\min}(\rho^{\otimes n} \Vert \sigma^{\otimes n})  +  \frac{1}{n} \log_{2}\!\left(\frac{1}{\varepsilon+\delta}\right)  \notag \\
& \qquad +  \frac{1}{n}\log_{2} \!\left(  \frac{1}{1-f(\varepsilon,\delta)}  \right) \label{eq: connecting with hypothesis eq} \\
& \stackrel{(c)}\leq D(\rho \Vert \sigma) + \sqrt{\frac{1}{n} V(\rho\Vert\sigma)} \ \Phi^{-1}(\varepsilon + \delta) + O\!\left( \frac{\log n}{n} \right) \\
& \stackrel{(d)} \leq D(\rho \Vert \sigma) + \sqrt{\frac{1}{n} V(\rho\Vert \sigma)} \ \Phi^{-1}(\varepsilon) + O\!\left( \frac{\log n}{n} \right). \label{eq:upperbound-second-order}
\end{align}
In the above, (a) follows from \eqref{eq:smooth-max-with-sub-and-not-normalized}, (b) from \eqref{eq:smooth-max-and-min-n-tensor} and \eqref{eq:smooth-max-to-hypothesis}, (c) from \eqref{eq:hypothesis-testing-second-order}, and lastly (d) from the following argument given below.
As stated in \cite[Eq. (4)]{datta2014second} in the proof of \cite[Lemma~3.7]{datta2014second}:
for $f$ a continuously  differentiable function,  by Taylor's theorem,
\begin{equation}
    f \!\left(x \pm \frac{1}{\sqrt{n}} \right)= f(x) \pm \frac{1}{\sqrt{n}}f'(\eta),
\end{equation}
where $\eta \in \left( x-\frac{1}{\sqrt{n}},x \right)$ in the case $-$ and $\eta \in \left( x, x+\frac{1}{\sqrt{n}}\right)$ in the case $+$. By choosing $\delta=1/\sqrt{n}$, the function $\Phi^{-1}(\varepsilon + \delta)$ also satisfies the said property.

With the upper bound in \eqref{eq:upperbound-second-order}, together with the lower bound in \eqref{eq:lower-bound-second-order-D-min}, we conclude the proof. 
\end{IEEEproof}

\begin{corollary}[Second-order asymptotics of smooth sandwiched R\'enyi relative entropy] \label{cor:second-order-smooth-sandwiched-Renyi}
Let $\rho$ be a state and $\sigma$ a PSD operator. Fix $\varepsilon \in (0,1)$. The smooth sandwiched R\'enyi relative entropy of order $\alpha >1$ has the following second-order expansion
       \begin{multline}
    \frac{1}{n} \widetilde{D}_{\alpha}^{\varepsilon}\!\left(\rho^{\otimes n} \Vert\sigma^{\otimes n}\right) = \\  D(\rho \Vert \sigma) -\sqrt{\frac{1}{n} V(\rho \Vert \sigma)} \ \Phi^{-1}(\varepsilon) + O\!\left( \frac{\log n}{n} \right).
\end{multline} 
For $\alpha \in[1/2,1)$, the smooth sandwiched R\'enyi relative entropy has the following second-order expansion:
      \begin{multline}
    \frac{1}{n} \widetilde{D}_{\alpha}^{\varepsilon}\!\left(\rho^{\otimes n} \Vert\sigma^{\otimes n}\right) = \\  D(\rho \Vert \sigma) +\sqrt{\frac{1}{n} V(\rho \Vert \sigma)} \ \Phi^{-1}(\varepsilon) + O\!\left( \frac{\log n}{n} \right).
\end{multline} 
\end{corollary}

\begin{IEEEproof}
This follows from techniques similar to those in the proof of \cref{thm:second-order-Dmin}, as well as from the $\alpha$-monotonicity of the sandwiched R\'enyi relative entropy and Theorem~\ref{thm:connection-to-smooth-renyi}.
See Appendix~\ref{APP:Second-order-smooth-renyi} for a proof.
\end{IEEEproof}

\begin{remark}[Equivalence of smooth relative entropies up to the second-order]
    \cref{cor:second-order-smooth-sandwiched-Renyi} indicates that, in the asymptotic i.i.d.~setting and up to the second order, there is no difference between all of the smooth sandwiched R\'enyi relative entropies for all $\alpha >1$. That is, they are all equivalent to the smooth max-relative entropy. Similarly, in the asymptotic i.i.d.~setting and up to the second order, there is no difference between all of them for $\alpha \in [1/2,1)$: they are all equivalent to the smooth min-relative entropy. This is a unifying result that should find use in future works on quantum resource theories.
\end{remark}

\medskip
A sequence $\{a_n\}_n$ is called a moderate sequence if $a_n \to 0$ and $\sqrt{n} a_n \to \infty$ when $n \to \infty$.

\begin{proposition}[Moderate deviations]
\label{prop:mod-dev-smooth-F-min}
For a moderate sequence $\{a_n\}_n$ and $\varepsilon_{n}= e^{-n a^2_n}$,
the smooth $F$-min-relative entropy scales as follows:
    \begin{equation}
        \frac{1}{n} D^{\varepsilon_n}_{\min,F}\!\left( \rho^{\otimes n} \middle \Vert \sigma^{\otimes n} \right) = D(\rho \Vert \sigma) -\sqrt{2 V(\rho \Vert \sigma)} \ a_n + o(a_n).
    \end{equation}
\end{proposition}

\begin{IEEEproof}
This follows by utilizing the connections established in the last section along with the moderate deviation analysis for the smooth min-relative entropy in \cite{chubb2017moderate}. For completeness, we provide a proof in Appendix~\ref{App: moderate deviation}.
\end{IEEEproof}

\begin{remark}[Moderate deviations of smooth sandwiched R\'enyi relative entropy]
 Similar to \cref{prop:mod-dev-smooth-F-min}, we arrive at the following scaling for smooth sandwiched R\'enyi relative entropy: for a moderate sequence $\{a_n\}_n$ and $\varepsilon_{n}= e^{-n a^2_n}$,
  \begin{enumerate}
      \item $\alpha >1 \colon$ 
      \begin{equation}
        \frac{1}{n} \widetilde{D}^{\varepsilon_n}_{\alpha}\!\left( \rho^{\otimes n} \middle \Vert \sigma^{\otimes n} \right) = D(\rho \Vert \sigma) +\sqrt{2 V(\rho \Vert \sigma)} \ a_n + o(a_n).
    \end{equation}
    \item $\alpha \in [1/2,1) \colon$
    \begin{equation}
        \frac{1}{n} \widetilde{D}^{\varepsilon_n}_{\alpha}\!\left( \rho^{\otimes n} \middle \Vert \sigma^{\otimes n} \right) = D(\rho \Vert \sigma) - \sqrt{2 V(\rho \Vert \sigma)} \ a_n + o(a_n).
    \end{equation}
  \end{enumerate}  
 The proof follows by employing the techniques in the derivation of second-order asymptotics (in Appendix~\ref{APP:Second-order-smooth-renyi}) and proceeding as in Appendix~\ref{App: moderate deviation}.
\end{remark}

\section{Application to Randomness Distillation}

\label{Sec: Randomness Distillation}

\subsection{1W-LOCC-Assisted Randomness Distillation}

The distillable randomness of a bipartite state is a measure of classical correlations contained in that state. This quantity was first proposed and characterized in \cite{devetak2004distilling}, and later studied in various guises in \cite{OHHH02,D05,DHW05RI,KD07,MPZ18,MLA19,CNB22,lami2022upper}. 
It has applications to determining the fundamental limitations on experiments in which the goal is to distill randomness from bipartite states \cite{zhang2020expRand,shalm2021random}.
For completeness, we review the definition of a randomness distillation protocol  assisted by one-way local operations and classical communications (1W-LOCC)  (shown in \cref{fig:one way LOCC}). We follow the recent presentation of \cite{lami2022upper}. In Section~\ref{sec:RD-gen-LOCC}, we recall  the definition of randomness distillation assisted by general LOCC, as also considered in \cite{lami2022upper}. 

\begin{figure}
    \centering
    \includegraphics[width=\linewidth]{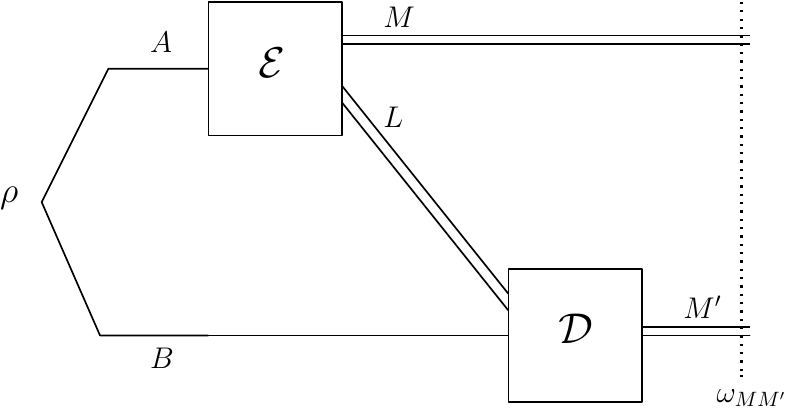}
    \caption{One way local operations and classical communication protocol from $A$ to $B$. First, the channel $\cE_{A \to ML}$ with classical outputs $L$ and $M$ is applied by Alice. Then, system $L$ is communicated to Bob via a noiseless classical channel. Bob applies the decoding channel $\cD_{LB \to M'}$ to get the classical output $M'$. At the end of the protocol, the output state shared by Alice and Bob is $\omega_{MM'}$, and it should be close to a maximally classically correlated state. In the above figure, classical systems are denoted by double lines.}
    \label{fig:one way LOCC}
\end{figure}

Let $\rho_{AB}$ be a bipartite state. The protocol  begins with Alice applying  a quantum channel $\cE_{A \to ML}$, with the output systems $L$ and $M$  classical. Then, system $L$ is communicated to Bob over a noiseless classical channel. Having system $L$, Bob acts with the decoding channel $\cD_{LB \to M'}$ on his systems. At the end of the protocol,  the final state is
\begin{equation}
    \omega_{MM'} \coloneqq (\cD_{LB \to M'} \circ \cE_{A \to ML})(\rho_{AB}).
\end{equation}
A $(d,\varepsilon)$ randomness distillation protocol   satisfies 
\begin{equation}
    F\!\left( \overline{\Phi}^d_{MM'}, \omega_{MM'} \right) \geq 1- \varepsilon, 
\end{equation}
where $ \overline{\Phi}^d_{MM'}$ is the maximally classically correlated state of rank $d$:
\begin{equation}
     \overline{\Phi}^d_{MM'} \coloneqq \frac{1}{d} \sum_{i=0}^{d-1} |i \rangle\!\langle i|_M \otimes  |i \rangle\!\langle i|_{M' }.
     \label{eq:max-class-corr-def}
\end{equation}

The one-shot distillable randomness of $\rho_{AB}$ is defined as 
\begin{multline}
    R^\varepsilon(\rho_{AB})  \coloneqq  \\ 
   \sup_{\substack{\cE_{A \to ML} \\ \cD_{LB \to M'}} } \left\{  \log_{2} d - \log_{2} d_L:  F\!\left(  \overline{\Phi}^d_{MM'},
   \omega_{MM'} \right) \geq 1- \varepsilon \right\}.
\end{multline}
Intuitively, $R^\varepsilon(\rho_{AB}) $ is the largest net number of  maximally classically correlated random bits that can be generated from the state $\rho_{AB}$. Here, we need to subtract the number of bits of classical communication used in the protocol, in order to rule out the possibility of distilling an infinite number of shared random bits.

\subsection{LOCC-Assisted Randomness Distillation} 

\label{sec:RD-gen-LOCC}

\begin{figure}
    \centering
    \includegraphics[width=\linewidth]{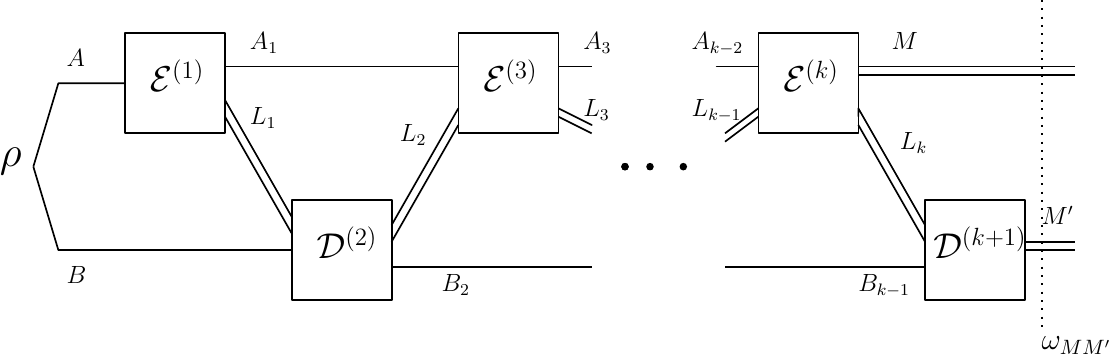}
    \caption{General local operations and classical communication protocol from $A$ to $B$: First the channel $\cE^{(1)}_{A \to A_1L_1}$ with classical output $L_1$ is applied by Alice. Then, system $L_1$ is communicated to Bob via a noiseless classical channel. Next Bob applies the channel $\cD^{(2)}_{L_1 B\to L_2 B_2}$, and the classical output $L_2$ is communicated to Alice. This procedure is continued for $k-1$ rounds. During the final round Alice performs $\cE^{(k)}_{A_{k-2} L_{k-1}\to M L_k}$ where both the output systems are classical and communicates $L_k$ to Bob. Bob completes the protocol by applying the channel $\cD^{(k+1)}_{L_k B_{k-1} \to M'}$. At the end of the protocol, the output state shared by Alice and Bob is $\omega_{MM'}$, and it should be close to a maximally classically correlated state.}
    \label{fig:general-LOCC}
\end{figure}

Here we review general LOCC protocols for randomness distillation \cite{lami2022upper} (shown in \cref{fig:general-LOCC}). A general LOCC-assisted randomness distillation protocol starts with Alice performing the channel $\cE^{(1)}_{A \to A_1 L_1}$, with system $L_1$ being classical and communicated to Bob. Then, Bob performs the channel $\cD^{(2)}_{L_1 B \to L_2 B_2}$, with system $L_2$ being classical and communicated to Alice. The above procedure continues for $k$ rounds. We denote the rest of the channels for Alice and Bob as 
$\{\cE^{(i)}_{ A_{i-2} L_{i-1} \to A_{i} L_{i} } \}_i$ for $i \in \{3,5, \ldots\}$ and $\{\cD^{(i)}_{L_{i-1} B_{i-2} \to L_{i} B_{i}} \}_i$ for $i \in \{4, \ldots\}$, respectively. Without loss of generality, we consider the last two channels of the protocol to be $\cE^{(k)}_{A_{k-2} L_{k-1}\to M L_{k}} $ and 
$\cD^{(k+1)}_{L_{k} B_{k-1} \to M'} $. The state shared by Alice and Bob at the end of this protocol is
\begin{equation}
    \omega_{MM'} \coloneqq  \left ( \cD^{(k+1)} \circ \cE^{(k)} \circ \ldots \circ \cD^{(2)} \circ \cE^{(1) }             \right)\left(\rho_{AB}\right).
\end{equation}
The above-mentioned protocol has $\varepsilon$ error if it satisfies 
\begin{equation}
p_{\mathrm{err}}(\cP^k ) \coloneqq  1- F\!\left( \omega_{MM'}, \overline{\Phi}^d_{MM'}\right) \leq \varepsilon, 
\label{eq:error-constraint-1W-dist-rand}
\end{equation}
where $\cP^k$ corresponds to the protocol described above.
With that, the one-shot distillable randomness from $\rho_{AB}$ assisted by LOCC is defined as 
\begin{multline}
     R^\varepsilon_{ \leftrightarrow}(\rho_{AB})  \coloneqq  \\ 
  \sup_{k \in \NN,\cP^k} \left\{  \log_{2} d - \sum_{i=1}^k  \log_{2} d_{L_i}: p_{\mathrm{err}}(\cP^k ) \leq  \varepsilon \right\}.
\end{multline}

Since general LOCC assistance contains 1W-LOCC assistance as a special case, we obtain the following bound for every state $\rho_{AB}$ and $\varepsilon\in[0,1]$:
\begin{equation} \label{eq: inequality-1W-and-General}
    R^\varepsilon(\rho_{AB}) \leq R^\varepsilon_{ \leftrightarrow}(\rho_{AB}).
\end{equation}

\subsection{\texorpdfstring{$\Gamma$}{Gamma}-Upper Bound on One-Shot Distillable Randomness}

Let us recall the definition of the classical correlation measure $\gamma$, defined recently in \cite{lami2022upper} for a PSD bipartite operator~$\sigma_{AB}$ as
\begin{equation}
    \gamma(\sigma_{AB}) \coloneqq \inf_{K_A, L_B, V_{AB} \in \mathrm{Herm}} \left\{  \begin{array}[c]{c}
    \operatorname{Tr}[K_A \otimes L_B]: \\ \T_B(V_{AB} \pm \sigma_{AB} )\geq 0, \\ K_A \otimes L_B \pm V_{AB} \geq 0
    \end{array} \right\}.
    \label{eq:def-class-corr-gamma-meas}
\end{equation}
{Some intuition for this quantity was not discussed in \cite{lami2022upper}, and so we provide some briefly here. For a classical correlation measure, it is desirable for it to indeed measure correlations, meaning that it should be equal to zero for a product state, be greater than zero for a state that is not product, and should not increase under the action of local channels. For this purpose, mutual information and its variants are helpful, but they measure quantum correlations in addition to classical correlations \cite{GPW05}. Mutual information can be written as the minimum quantum relative entropy between the state of interest and the set of product states \cite[Exercise~11.8.2]{Wbook17}:
\begin{equation}
    I(A;B)_{\rho} = \inf_{\substack{\sigma_A \in \mathcal{D}(\mathcal{H}_A),\\ \sigma_B \in \mathcal{D}(\mathcal{H}_B)}} D(\rho_{AB} \Vert \sigma_A \otimes \sigma_B).
\end{equation}
Using the max-relative entropy in place of the quantum relative entropy, one can define the max-mutual information as \cite{CBR14}
\begin{align}
    &I_{\max}(A;B)_{\rho} \notag \\  & \coloneqq \inf_{\substack{\sigma_A \in \mathcal{D}(\mathcal{H}_A),\\ \sigma_B \in \mathcal{D}(\mathcal{H}_B)}} D_{\max}(\rho_{AB} \Vert \sigma_A \otimes \sigma_B) \\
    & = \log_2 \inf_{K_A, L_B\geq 0} \{\operatorname{Tr}[K_A \otimes L_B] : \rho_{AB} \leq K_A \otimes L_B\}.
    \label{eq:max-MI-rewrite}
\end{align}
As mentioned above, this quantity measures quantum correlations in addition to classical correlations; for example, it is equal to $2\log_2 d$ for a maximally entangled state of Schmidt rank~$d$, as we prove in Appendix~\ref{app:max-MI-props} (we also prove the equality above there). Given that the classical correlations contained in or the distillable randomness extractable from such a state is equal to $\log_2 d$, the mutual information is twice as large as it should be for this case. To address this problem, recall that the logarithmic negativity of a bipartite state is defined as follows \cite{ZHSL98,Vidal2002}:
\begin{equation}
    \log_2 \left\Vert T_B (\rho_{AB})\right\Vert_1 =
\log_2 \inf_{\substack{M_{AB},\\
N_{AB}\geq 0}} \left\{  \begin{array}[c]{c}\operatorname{Tr}[M_{AB}+N_{AB}]:\\
M_{AB}\geq T_B (\rho_{AB}), \\
N_{AB}\geq - T_B (\rho_{AB})\end{array}\right\},
\end{equation}
where the second expression follows from \cite[Proposition~3.53]{khatri2020principles}.
It is equal to $\log_2 d $ for a maximally entangled state and equal to zero for a maximally classically correlated state of rank $d$. This latter property is also undesirable for a classical correlation measure, which should be equal to $\log_2 d$ for such a state. The basic idea behind the measure in \eqref{eq:def-class-corr-gamma-meas} is to combine features of the max-mutual information and logarithmic negativity into a single measure (the features being the optimization over product positive semi-definite operators and the use of the partial transpose). The resulting measure is equal to $\log_2 d $ for both a maximally entangled state of Schmidt rank~$d$ and a maximally classically correlated state of rank~$d$, and it satisfies a number of desirable properties expected of a classical correlation measure, as shown in \cite{lami2022upper}.}

We can then use the definition in \eqref{eq:def-class-corr-gamma-meas} and the general construction in \cite[Eq.~(4)]{lami2022upper} to define the following classical correlation measure for a bipartite state $\rho_{AB}$, relevant for us here:
\begin{equation}
\Gamma^\varepsilon_{\min,F}(A;B)_\rho \coloneqq \inf_{\substack{\sigma_{AB} \geq 0: \\ \gamma(\sigma_{AB}) \leq 1} } D^\varepsilon_{\min,F}(\rho_{AB} \Vert \sigma_{AB}).
\label{eq:def-Gamma-min-F-smooth}
\end{equation}

\begin{theorem} \label{thm:distillable-randomness-with-gamma-set}
    Fix $\varepsilon\in(0,1)$. The following bound holds for the  one-shot LOCC-assisted distillable randomness of a bipartite state $\rho_{AB}$: 
    \begin{equation}
         R_{\leftrightarrow}^\varepsilon(\rho_{AB})  \leq   \Gamma^\varepsilon_{\min,F}(A;B)_\rho .
         \label{eq:gen-LOCC-up-bnd-Gamma}
    \end{equation}
\end{theorem}

\begin{IEEEproof}
Let us begin by proving the bound
\begin{equation}
    R^\varepsilon(\rho_{AB})  \leq   \Gamma^\varepsilon_{\min,F}(A;B)_\rho ,
    \label{eq:1wlocc-up-bnd}
\end{equation}
and then we discuss afterward how to generalize it to get \eqref{eq:gen-LOCC-up-bnd-Gamma}. 
    The proof of \eqref{eq:1wlocc-up-bnd} follows similarly to the proof of  \cite[Eq.~(41)]{lami2022upper}, by some properties satisfied by $D^\varepsilon_{\min,F}(\cdot\Vert \cdot)$, which include data processing (see \cref{thm:data-processing}) and scaling (see Property~1 of \cref{thm:other-properties}).  These properties of smooth
    $F$-min-relative entropy result in $\Gamma^\varepsilon_{\min,F}(A;B)_\rho$ satisfying the properties presented in \cite{lami2022upper}; out of those we use Proposition~1 on symmetry, Proposition~2 on data processing under local channels, and Proposition~7 on scaling.

    By Lemma~9 of \cite{lami2022upper}, we have 
    \begin{equation} \label{eq: gamma 1/d bound}
        \sup_{\substack{\tau_{MM'} \geq 0: \\  \gamma(\tau_{MM'}) \leq 1}}  F\!\left(  \overline{\Phi}^d_{MM'}, \tau_{MM'}\right)  \leq \frac{1}{d}.
    \end{equation}
    Then, considering the constraint $F\!\left(  \overline{\Phi}^d_{MM'}, \omega_{MM'}\right) \geq 1- \varepsilon$ from \eqref{eq:error-constraint-1W-dist-rand}, we arrive at 
    \begin{align}
        \log_{2} d & \leq -\log_{2}  F\!\left(  \overline{\Phi}^d_{MM'},
        \tau_{MM'} \right) \\
        & \leq    D^\varepsilon_{\min,F}(\omega_{MM'} \Vert \tau_{MM'} ).
    \end{align}
    The above inequality holds for all $\tau_{MM'} \geq 0$ satisfying $\gamma(\tau_{MM'}) \leq 1$. Thus we have 
     \begin{align}
        \log_{2} d & \leq \inf_{\substack{\tau_{MM'} \geq 0: \\ \gamma(\tau_{MM'}) \leq 1}} D^\varepsilon_{\min,F}(\omega_{MM'} \Vert \tau_{MM'} ) \\
        & =  \Gamma^\varepsilon_{\min,F}(M;M')_\omega.
    \end{align}
By the data-processing inequality for the smooth-min relative entropy (\cref{thm:data-processing}) applied to the channel $\cD_{BL \to M'}$, as well as Proposition~2 of \cite{lami2022upper}, we get 
\begin{equation}
    \Gamma^\varepsilon_{\min,F}(M;M')_\omega \leq  \Gamma^\varepsilon_{\min,F}(M;BL)_{\cE(\rho)}.
\end{equation}
Then applying Proposition 7 of \cite{lami2022upper} with the assistance of the scaling property of $D^\varepsilon_{\min,F}(\cdot \Vert \cdot)$ (\cref{thm:other-properties}), we find that 
\begin{equation}
    \Gamma^\varepsilon_{\min,F}(M;BL)_{\cE(\rho)} \leq \log_{2} d_L +  \Gamma^\varepsilon_{\min,F}(LM;B)_{\cE(\rho)}.
\end{equation}
Next, again applying data processing under the local channel $\cE_{A \to LM}$, we conclude that 
\begin{equation}
    \Gamma^\varepsilon_{\min,F}(LM;B)_{\cE(\rho)} \leq  \Gamma^\varepsilon_{\min,F}(A;B)_{\rho}. 
\end{equation}
Putting everything together, we arrive at
\begin{equation}
    \log_{2} d - \log_{2} d_L \leq  \Gamma^\varepsilon_{\min,F}(A;B)_{\rho}.
\end{equation}
Since the above inequality holds for an arbitrary $(d,\varepsilon)$ randomness distillation protocol, we conclude the desired bound in \eqref{eq:1wlocc-up-bnd}.

The proof of \eqref{eq:gen-LOCC-up-bnd-Gamma} then follows by iterating the same reasoning as in the proof above, while going backward through the protocol $\cP^k$ defined in \cref{sec:RD-gen-LOCC}. See the end of the proof of \cite[Theorem~11]{lami2022upper} for similar reasoning.
\end{IEEEproof}

\begin{remark}[Comparison with other existing bounds]
{Note that the bound provided in the previous work (Theorem~11 of \cite{lami2022upper}) is a consequence of \cref{thm:distillable-randomness-with-gamma-set}. In particular, this previous work established the following bound:  For $\alpha > 1$
\begin{equation}
R^\varepsilon_{\leftrightarrow}(\rho_{AB})  \leq \inf_{\substack{ \sigma_{AB} \geq 0 ,\\ \gamma(\sigma_{AB}) \leq 1}} \widetilde{D}_\alpha( \rho_{AB} \Vert \sigma_{AB}) + \frac{\alpha}{\alpha-1} \log_2 \!\left( \frac{1}{1-\varepsilon} \right).
\end{equation}
Together with the relationship between the sandwiched R\'enyi relative entropy and the smooth $F$-min relative entropy derived in \cref{prop: Dmin and smooth renyi}, it can be seen that the bound derived in \cref{thm:distillable-randomness-with-gamma-set} implies the previous bound. 
}  
\end{remark}

\subsection{Smooth \texorpdfstring{$F$}{F}-Min- and Min-Mutual-Information Bounds on One-Shot Distillable Randomness} 

In this section, we define the smooth $F$-min-mutual information of a state and establish several of its properties, thus justifying it as a correlation measure for bipartite states. We also discuss how it leads to an alternative upper bound on the one-shot distillable randomness.

Let us define the following mutual-information-like correlation measure, which we call the smooth $F$-min-mutual information:
\begin{equation}
    I^\varepsilon_{\min,F}(A;B)_\rho \coloneqq   \inf_{\substack{\sigma_{A} , \sigma_{B} \geq 0: \\ \Tr[\sigma_A] \leq 1,  \Tr[\sigma_B] \leq 1} } \hspace{-4mm} D^\varepsilon_{\min,F}(\rho_{AB} \Vert \sigma_{A} \otimes \sigma_B).
\end{equation}
Since every subnormalized product state $\sigma_A \otimes \sigma_B$ satisfies $\gamma(\sigma_A \otimes \sigma_B) \leq 1$ (as a consequence of Propositions~4 and 6 of \cite{lami2022upper}), we conclude that the following bound holds for every state $\rho_{AB}$:
\begin{equation}
    \Gamma^\varepsilon_{\min,F}(A;B)_\rho \leq I^\varepsilon_{\min,F}(A;B)_\rho .
\end{equation}
Next, we prove various properties of $I^\varepsilon_{\min,F}(A;B)_\rho$.
\begin{lemma}\label{lem: div with prod}
    Let $\rho_{AB}$ be a state. Then $I^\varepsilon_{\min,F}(A;B)_\rho$ satisfies the following properties. 
    \begin{enumerate}
        \item Symmetry: 
        \begin{equation}
            I^\varepsilon_{\min,F}(A;B)_\rho =I^\varepsilon_{\min,F}(B;A)_\rho
        \end{equation}
        \item Data processing under local channels: 
        Let $\cN_{A \to A'}$ and $\cM_{B \to B'}$ be quantum channels, then 
        \begin{equation}
            I^\varepsilon_{\min,F}(A;B)_\rho \geq  I^\varepsilon_{\min,F}(A';B')_\omega,
        \end{equation}
        where $\omega_{AA'} \coloneqq (\cN_{A \to A'} \otimes \cM_{B \to B'})(\rho_{AB})$.
        \item Classical communication bound: Let $\rho_{XAB}$ be a tripartite state: 
        \begin{equation}
            \rho_{XAB} \coloneqq \sum_{x} p(x) |x \rangle\!\langle x|_X \otimes \rho^x_{AB},
        \end{equation}
        where $\{p(x) \}_x$ is a probability distribution and $\{ \rho^x_{AB}\}_x$ is a set of states. Then, we have 
        \begin{equation}
             I^\varepsilon_{\min,F}(AX;B)_\rho \leq \log_2 d_X + I^\varepsilon_{\min,F}(A;BX)_\rho.
        \end{equation}
    \end{enumerate}
\end{lemma}

\begin{IEEEproof}
   \underline{\textit{Symmetry:}} This follows because $\sigma_B \otimes \sigma_A \geq 0$ if and only if $\sigma_A \otimes \sigma_B \geq 0$, and by the unitary invariance of $D^\varepsilon_{\min,F}(\cdot \Vert \cdot)$. In particular, by applying the unitary SWAP operation, we have 
\begin{equation}
    D^\varepsilon_{\min,F}(\rho_{AB} \Vert \sigma_A \otimes \sigma_B) = D^\varepsilon_{\min,F}(\rho_{BA} \Vert \sigma_B \otimes \sigma_A).
\end{equation}
Then, by definition of $ I^\varepsilon_{\min,F}(A;B)_\rho$, it satisfies symmetry. 

\medskip 
\underline{\textit{Data processing under local channels:}}
   By the data-processing inequality for the smooth $F$-min-relative entropy (Theorem~\ref{thm:data-processing}), we have 
   \begin{multline}
       D^\varepsilon_{\min,F}(\rho_{AB} \Vert \sigma_A \otimes \sigma_B) \geq  \\ D^\varepsilon_{\min,F}\!\left(\omega_{A'B'} \Vert \cN_{A \to A'}(\sigma_A) \otimes  \cM_{B \to B'}(\sigma_B)\right).
   \end{multline}
Since $\cN_{A \to A'}$ and $\cM_{B \to B'}$ are quantum channels, we find that 
$\cN_{A \to A'}(\sigma_A) \otimes  \cM_{B \to B'}(\sigma_B) \geq 0$, $\Tr\!\left[\cN_{A \to A'}(\sigma_A)\right] \leq 1$ and $\Tr\!\left[\cM_{B \to B'}(\sigma_B) \right] \leq 1$. 
With that, we conclude that
\begin{multline}
     D^\varepsilon_{\min,F}(\rho_{AB} \Vert \sigma_A \otimes \sigma_B) \geq \\
     \inf_{\substack{\tau_{A'}, \tau_{B'} \geq 0: \\ \Tr[\tau_{A'}] \leq 1,  \Tr[\tau_{B'}] \leq 1} } \hspace{-4mm} D^\varepsilon_{\min,F}(\omega_{A'B'} \Vert \tau_{A'} \otimes \tau_{B'}).
\end{multline}
Then, optimizing over $\sigma_{A},\sigma_{B} \geq 0, \Tr[\sigma_A] \leq 1,$ and $  \Tr[\sigma_B] \leq 1$, we arrive at the desired conclusion. 

\medskip
\underline{\textit{Classical communication bound:}}  
Fix $\sigma_A$ and $\sigma_{BX}$ such that  $\sigma_{A}, \sigma_{BX} \geq 0, \Tr[\sigma_A] \leq 1,  \Tr[\sigma_{BX}] \leq 1$.
By scaling of smooth $F$-min-relative entropy in \cref{thm:other-properties}, we have 
\begin{multline}
     \log_2 d_X +  D^\varepsilon_{\min,F}(\rho_{ABX} \Vert \sigma_A \otimes \sigma_{BX}) \\
     = D^\varepsilon_{\min,F}\!\left(\rho_{ABX} \middle \Vert \sigma_A \otimes \frac{\sigma_{BX}}{d_X}\right).
\end{multline}
Denote the completely dephasing channel as 
\begin{equation} \label{eq: completely dephasing channel with X}
     \overline{\Delta}_X(\cdot) \coloneqq \sum_x |x \rangle\!\langle x|(\cdot) |x \rangle\!\langle x|,
\end{equation}
and set $\sum_x \widetilde{\sigma}^x_{B} \otimes  |x \rangle\!\langle x| 
 \coloneqq \overline{\Delta}_X(\sigma_{BX})  $. 
Consider that
\begin{align}
     &D^\varepsilon_{\min,F}\!\left(\rho_{ABX} \middle \Vert \sigma_A \otimes \frac{\sigma_{BX}}{d_X}\right)  \cr
     &\geq     D^\varepsilon_{\min,F}\!\left(\overline{\Delta}_X(\rho_{ABX}) \middle \Vert \sigma_A \otimes \frac{\overline{\Delta}_X(\sigma_{BX})}{d_X}\right) \\
     &=  D^\varepsilon_{\min,F}\!\left(\rho_{AXB} \middle \Vert \sigma_A \otimes \frac{1}{d_X} \sum_x |x \rangle\!\langle x| \otimes \tilde{\sigma}^x_B \right) \\
     &\geq   D^\varepsilon_{\min,F}\!\left(\rho_{AXB} \middle \Vert \sigma_A \otimes \frac{I}{d_X} \otimes \sum_x  \tilde{\sigma}^x_B \right) \\
     & \geq  \inf_{\substack{\tau_{AX} \otimes \tau_{B} \geq 0: \\ \Tr[\tau_{AX}] \leq 1,  \Tr[\tau_B] \leq 1} } D^\varepsilon_{\min,F}(\rho_{AXB} \Vert \tau_{AX} \otimes \tau_{B}), 
\end{align}
where the first inequality follows from data processing under the dephasing channel, the equality from the unitary invariance of smooth $F$-min-relative entropy with the SWAP operator to interchange the $B$ and $X$ systems,
the next inequality from Property~6 of \cref{thm:other-properties} (since $|x \rangle\!\langle x| \leq I$ for every $x$), and the final inequality because $\Tr[\sum_{x} \tilde{\sigma}^x] \leq 1$ given that $\Tr[\sigma_{BX}] \leq 1$ and $\Tr[ \sigma_A \otimes I/d_X ] \leq 1$, given that $\Tr[\sigma_{A}] \leq 1$. 
\end{IEEEproof}

\medskip 
Motivated by the expression for the distillable randomness of classical--quantum states from  \cite{devetak2004distilling}, we present the following upper bound on the one-shot distillable randomness. Indeed, it can be understood as a one-shot generalization of the quantity from \cite[Eq.~(24)]{devetak2004distilling}.

\begin{proposition}
\label{prop:RD-With-the-set-of-product-states-or-sub-normalized-states}
    Fix $\varepsilon \in (0,1)$. The following upper bound holds for the  one-shot distillable randomness of a classical-quantum (cq) state 
   $\rho_{XB}$: 
    \begin{equation}
        R^\varepsilon(\rho_{XB})   \leq    I^\varepsilon_{\min,F}(X;B)_{\rho}.
        \label{eq:1st-DW-style-1-shot}
    \end{equation}
    Furthermore, for a bipartite quantum state $\rho_{AB}$, we have 
     \begin{equation}
         R^\varepsilon(\rho_{AB})   \leq \\
         \sup_{\cM: A \to X}   I^\varepsilon_{\min,F}(X;B)_{\omega},
         \label{eq:2nd-DW-style-1-shot}
     \end{equation}
     where $\omega_{XB}\coloneqq \mathcal{M}_{A\to X}(\rho_{AB})$ and the supremum is over every measurement channel $\mathcal{M}_{A\to X}$ that takes a quantum input $A$ and outputs a classical system $X$, i.e., of the form $\mathcal{M}_{A\to X}(\cdot) = \sum_x \Tr[M^x_A (\cdot)] |x\rangle\!\langle x|_X$, where $\{M^x_A\}_x$ is a POVM.
\end{proposition}

\begin{IEEEproof}
The first inequality in \eqref{eq:1st-DW-style-1-shot} follows from similar reasoning as in the proof of \cref{thm:distillable-randomness-with-gamma-set}. For that to follow, Lemmas~\ref{lem: div with prod} and \ref{lem: upperbound 1/d} establish the required properties for $I^\varepsilon_{\min,F}(A;B)_\rho$. Note that these lemmas are proved for a general bipartite state, and these also hold for the special case of cq states.

The second inequality in \eqref{eq:2nd-DW-style-1-shot} follows because the inequality in \eqref{eq:1st-DW-style-1-shot} holds for an arbitrary cq state formed after a measurement on the $A$ system, and every such one-way LOCC protocol for randomness distillation consists of a measurement on Alice's system as the first step.   
\end{IEEEproof}

\begin{remark}[Classical output state]
    Notice that for a fixed measurement channel $\cM$ acting on the system $A$, applying the completely dephasing channel $\overline{\Delta}_X$ in \eqref{eq: completely dephasing channel with X} we have 
    \begin{multline}
        D^\varepsilon_{\min,F}\!\left(\cM_{A \to X}(\rho_{AB}) \Vert \sigma_X \otimes \sigma_B \right) \geq   \\ D^\varepsilon_{\min,F}\!\left(\cM_{A \to X}(\rho_{AB}) \Vert \overline{\Delta}_X (\sigma_X) \otimes \sigma_B \right),
    \end{multline}
    which follows due to the data-processing inequality for the smooth $F$-min-relative entropy and $\overline{\Delta}_X \circ \cM_{A\to X} =\cM_{A\to X}$. This shows that the infimum in the definition of $I^\varepsilon_{\min,F}$ is achieved by a state that is classical on $X$.
\end{remark}

The following lemma was used in the proof of \cref{prop:RD-With-the-set-of-product-states-or-sub-normalized-states}:

\begin{lemma}\label{lem: upperbound 1/d}
    The following bound holds: 
    \begin{equation}
        \sup_{\substack{\sigma_{A} \otimes \sigma_{B} \geq 0: \\ \Tr[\sigma_A] \leq 1,  \Tr[\sigma_B] \leq 1} } F\!\left( \overline{\Phi}^d_{AB}, \sigma_A \otimes \sigma_B\right) \leq \frac{1}{d},
    \end{equation}
    where $\overline{\Phi}^d_{AB}$ is the maximally classically correlated state. 
\end{lemma}

\begin{IEEEproof}
We give two proofs for this statement. A first proof follows from the observation that every product state $\sigma_{A} \otimes \sigma_{B}$ satisfies $\gamma(\sigma_{A} \otimes \sigma_{B}) = 1$, by applying Proposition~4 of \cite{lami2022upper}. Then we obtain the inequality $\gamma(\sigma_{A} \otimes \sigma_{B}) \leq 1$ for subnormalized states by applying Proposition~6 of \cite{lami2022upper}. This then proves that the set of subnormalized product states is contained in the set  $\{\sigma_{AB} \geq 0 :  \gamma(\sigma_{AB}) \leq 1\}$, so that the desired statement follows from Lemma~9 of \cite{lami2022upper}.

    As an alternative proof, let $\sigma_A, \sigma_B$ satisfy the constraints $\sigma_{A} \otimes \sigma_{B} \geq 0, \Tr[\sigma_A] \leq 1$, and $\Tr[\sigma_B] \leq 1$.
By the data-processing inequality for fidelity, consider that 
\begin{equation}
     F\!\left( \overline{\Phi}^d_{AB}, \sigma_A \otimes \sigma_B\right)  \leq  F\!\left( \overline{\Phi}^d_{AB}, \overline{\Delta}_A(\sigma_A) \otimes \overline{\Delta}_B (\sigma_B)\right), 
\end{equation}
where $\overline{\Delta}$ is a  dephasing channel defined as
\begin{equation}
     \overline{\Delta}(\cdot) \coloneqq \sum_{m=0}^{d-1} |m \rangle\!\langle m|(\cdot) |m \rangle\!\langle m|.
\end{equation}
With that, we have
\begin{align}
  & \overline{\Delta}_A(\sigma_A) \otimes \overline{\Delta}_B (\sigma_B) \cr
  &= \sum_{m=0}^{d-1} |m \rangle\!\langle m|\sigma_A|m \rangle\!\langle m| \otimes \sum_{\ell=0}^{d-1} |\ell \rangle\!\langle \ell|\sigma_B |\ell\rangle\!\langle \ell| \\
  &=\sum_{m,\ell=0}^{d-1} \langle m|\sigma_A|m \rangle  \langle \ell|\sigma_B |\ell \rangle  |m \rangle\!\langle m| \otimes |\ell \rangle\!\langle \ell|.
\end{align}
Then, we obtain 
\begin{align}
    & F\!\left( \overline{\Phi}^d_{AB}, \overline{\Delta}_A(\sigma_A) \otimes \overline{\Delta}_B (\sigma_B)\right)  \notag \\ 
    &= \left( \sum_{m=0}^{d-1} \sqrt{\frac{1}{d} \langle m|\sigma_A|m \rangle  \langle m|\sigma_B |m\rangle}\right)^2 \\
    & \leq \frac{1}{d} \sum_{m=0}^{d-1}  \langle m|\sigma_A|m \rangle  \sum_{m=0}^{d-1}  \langle m|\sigma_B|m \rangle \\
    & = \frac{1}{d} \Tr[\sigma_A] \Tr[\sigma_B] \\
    & \leq \frac{1}{d},
\end{align}
where the first equality follows from the fidelity reducing to a classical fidelity, the first inequality by Cauchy--Schwarz, and last inequality from the assumptions that $\Tr[\sigma_A] \leq 1$ and $\Tr[\sigma_B] \leq 1$. 
Finally, we complete the proof of \cref{lem: upperbound 1/d} by supremizing over $\sigma_A$ and $\sigma_B$ satisfying the required constraints. 
\end{IEEEproof}

\medskip 
For a cq state with a uniform classical probability distribution, in what follows we derive a lower bound for the one-shot distillable randomness. We obtain this by devising an achievable protocol that makes use of position-based coding \cite{AJW17b} and the square-root measurement at the decoder. 
Let us recall that the smooth min-mutual information of a bipartite state $\rho_{AB}$ is defined for $\varepsilon \in [0,1]$ as \cite{WR12}
\begin{equation} \label{eq:smooth-min-MI}
    I^\varepsilon_{\min}(A;B)_\rho \coloneqq D^\varepsilon_{\min}(\rho_{AB} \Vert \rho_A \otimes \rho_B).
\end{equation}

\begin{proposition}[Lower bound]\label{prop:second-order-lower-bound-special-cq}Fix $\varepsilon \in (0,1)$ and $\eta \in (0, \varepsilon)$.
     For a cq state $\rho_{XB}$ of the form 
\begin{equation}
    \rho_{XB} \coloneqq \frac{1}{L} \sum_{x=1}^L |x \rangle\!\langle x|_X \otimes \rho_B^x,
\end{equation}
the one-shot distillable randomness of $\rho_{XB}$ is bounded from below as follows: 
\begin{equation}
    R^\varepsilon(\rho_{XB}) \geq \left\lfloor I^{\varepsilon-\eta}_{\min}(X;B)_\rho - \log_2 \!\left( \frac{4 \varepsilon}{\eta^2}  \right) \right\rfloor.
    \label{eq:lower-bnd-1-shot-dist-rand-uniform}
\end{equation}
\end{proposition}

\begin{IEEEproof}
   We prove this lower bound by devising an achievable one-way protocol from Alice to Bob.
   The proof below follows by employing the idea behind the protocol presented in the proof of \cite[Theorem~6]{khatri2019second}. 
     Prior to the initiation of the protocol, Alice and Bob share the state $\rho_{XB}$, Alice has access to the $X$ system, and Bob has access to the $B$ system.
   
   The protocol begins with Alice picking an index  $m\in \mathcal{M}$ uniformly at random and placing it in a classical register $M$. She then labels her $X$ system of $\rho_{XB}$ as $X_m$. She prepares $|\mathcal{M}|$ independent instances of the classical state
   \begin{equation}
       \rho_X = \frac{1}{L} \sum_{x=1}^L |x \rangle\!\langle x|_X
   \end{equation}
   and labels them as $X_1,\ldots, X_{m-1}, X_{m+1}, \ldots, X_{|\mathcal{M}|}$. She sends the registers $X_1, \ldots, X_{|\mathcal{M}|}$, in this order, over a classical channel to Bob, while keeping a copy of each of them in her laboratory (denote the copies by $X'_1, \ldots, X'_{|\mathcal{M}|}$). For a fixed value of $m$, the reduced state of Bob has the following form:
   \begin{equation}
       \rho_{X_1} \otimes \cdots \otimes \rho_{X_{m-1}} \otimes \rho_{X_m B} \otimes \rho_{X_{m+1}} \otimes \cdots \otimes \rho_{|\mathcal{M}|},
   \end{equation}
   and his goal is to employ a decoding measurement to figure out which $X$ system is correlated with his $B$ system.
   This reduced state has exactly the form considered in position-based coding \cite{AJW17b} (see also \cite{qi2018applications,wilde2017position} and in particular \cite[Eq.~(3.5)]{qi2018applications}). As such, at this point, we can invoke those results to conclude that as long as
   \begin{equation}
       \log_2 |\mathcal{M}| = \left\lfloor I^{\varepsilon-\eta}_{\min}(X;B)_\rho - \log_2 \!\left( \frac{4 \varepsilon}{\eta^2}  \right) \right\rfloor,
   \end{equation}
   it is possible for Bob to decode the index $m$ with an error probability $\leq \varepsilon$. 
    Furthermore, Bob can make use of the square-root measurement construction to perform the decoding, and due to the permutation symmetry of the protocol and measurement, it follows that the error probability in decoding each index $m$ is equal to the same fixed value $p_{\operatorname{err}}$, for some $p_{\operatorname{err}}\in [0,1]$ such that $p_{\operatorname{err}} \leq \varepsilon$ (see \cite[Eq.~(3.10)]{qi2018applications}). After performing the decoding, he places his measurement outcome in a classical register $M'$. Thus, the final state of registers $M$ and $M'$ at the end of the protocol is given by
   \begin{equation}
       \omega_{MM'} \coloneqq (1-p_{\operatorname{err}}) \overline{\Phi}_{MM'} + p_{\operatorname{err}} \frac{I_{MM'}- \overline{\Phi}_{MM'}}{|\mathcal{M}|-1},
   \end{equation}
   for which we have that
   \begin{equation}
       \frac{1}{2} \left \Vert \overline{\Phi}_{MM'} - \omega_{MM'}\right \Vert_1 = 1-F(\overline{\Phi}_{MM'} , \omega_{MM'}) = p_{\operatorname{err}} \leq \varepsilon.
   \end{equation}
   Since the state $\rho_X$ is uniform, the cost for Alice to communicate each of the registers $X_1, \ldots, X_{|\mathcal{M}|}$ over a classical channel is precisely equal to the amount of distillable randomness contained in the register pairs $(X_1, X_1')$, \ldots, $(X_{|\mathcal{M}|}, X_{|\mathcal{M}|}')$. Thus, the net number of random shared bits generated by this protocol is equal to $\log_2|\mathcal{M}|$.
\end{IEEEproof}

\subsection{Second-Order Expansions for Randomness Distillation}

\label{sec:2nd-order-expansions}

Now we show how to utilize  the second-order asymptotics of the smooth $F$-min-relative entropy, as presented in \cref{thm:second-order-Dmin}, to obtain upper and lower bounds on the rate at which randomness distillation is possible. Before doing so, let us 
 define the following quantities: 
\begin{equation} \label{def:first-order-term}
    \Gamma(A;B)_\rho \coloneqq \inf_{\substack{\sigma_{AB} \geq 0: \\ \gamma(\sigma_{AB}) \leq 1} } D(\rho_{AB} \Vert \sigma_{AB}),
\end{equation}
where $D(\rho_{AB} \Vert \sigma_{AB})$ is the quantum relative entropy from~\eqref{eq:relative-entropy-def}.
Then, denote $\Pi_\gamma \subseteq \{\sigma_{AB} \geq 0:\gamma(\sigma_{AB}) \leq 1\} $ as the set of PSD operators achieving the infimum in $\Gamma(A;B)_\rho$. From that, we define the following variance quantity: 
\begin{equation}
\label{def:second-order-term}
       V_{\Gamma}^\varepsilon(A;B)_\rho \coloneqq 
  \begin{cases} \inf_{\sigma_{AB} \in \Pi_\gamma}  V(\rho_{AB} \Vert \sigma_{AB})& \mbox{if} \ \varepsilon \geq \frac{1}{2},\\
 \sup_{\sigma_{AB} \in \Pi_\gamma}  V(\rho_{AB} \Vert \sigma_{AB}) & \mbox{if} \ \varepsilon < \frac{1}{2}.
\end{cases}
\end{equation}

\begin{theorem}[Upper bound]\label{thm:Upper-Bound-Second-Order-general}
 Fix $\varepsilon \in (0,1)$.   The LOCC-assisted distillable randomness of a bipartite state $\rho_{AB}$ is bounded from above as follows: 
    \begin{multline}
        \frac{1}{n}    R^\varepsilon_{\leftrightarrow}(\rho^{\otimes n}_{AB})   \leq \\
         \Gamma(A;B)_\rho + \sqrt{\frac{1}{n}V_{\Gamma}^\varepsilon(A;B)_\rho}\Phi^{-1}(\varepsilon) +  O\!\left( \frac{\log n}{n} \right).
    \end{multline} 
\end{theorem}

\begin{IEEEproof}
   Applying \cref{thm:distillable-randomness-with-gamma-set} and choosing $\sigma_{AB}$ to be an optimum in \eqref{def:second-order-term}, we find that
\begin{align}
      &\frac{1}{n}    R^\varepsilon_{\leftrightarrow}(\rho^{\otimes n}_{AB}) \notag   \\
     & \leq \frac{1}{n} \Gamma^\varepsilon_{\min,F}\! \left(A^n; B^n\right)_{\rho^{\otimes n}} \\ 
     & \leq  \frac{1}{n} D^\varepsilon_{\min,F}(\rho_{AB}^{\otimes n} \Vert \sigma_{AB}^{\otimes n}) \\ 
     & =\Gamma(A;B)_\rho + \sqrt{\frac{1}{n}V_{\Gamma}^\varepsilon(A;B)_\rho}\Phi^{-1}(\varepsilon) +  O\!\left( \frac{\log n}{n} \right).
\end{align}
The second inequality follows from the definition of $\Gamma^\varepsilon_{\min,F}$, and the last equality follows from the second-order expansion from \cref{thm:second-order-Dmin}.
\end{IEEEproof}

\medskip
In the above proof, once we restrict the second argument in $\Gamma^\varepsilon_{\min,F}\! \left(A^n; B^n\right)_{\rho^{\otimes n}}$ to be a tensor-power state, it follows from \cite[Lemma~63]{ChannelcodingRateYury} that the choices we have made are optimal.

\begin{remark}[Computational efficiency]
  We note here that we can further relax the upper bound in \cref{thm:Upper-Bound-Second-Order-general} to find an efficiently computable upper bound, by following an approach similar to that discussed in \cite[Section~VI]{lami2022upper}.
\end{remark}

Now let us recall the definition of the mutual information variance of a bipartite state $\omega_{AB}$ \cite[Section~2.2]{DTW14} as
\begin{equation}
        V(A;B)_\omega \coloneqq V(\omega_{AB} \Vert \omega_A \otimes \omega_B),
    \end{equation}
    where the relative entropy variance $V$ is defined in \eqref{eq:rel-ent-var}.

\begin{proposition}[Upper bound for cq states] \label{prop:Upper-Bound-Second-Order-CR}
     Fix $\varepsilon \in (0,1)$. The LOCC-assisted distillable randomness of a cq state $\rho_{XB}$ is bounded from above as follows: 
    \begin{multline} 
        \frac{1}{n}    R^\varepsilon_{\leftrightarrow}(\rho^{\otimes n}_{XB})   \leq \\ I(X;B)_\rho +  \sqrt{\frac{1}{n} V(X;B)_\rho} \ \Phi^{-1}(\varepsilon) 
        + O\!\left( \frac{\log n}{n} \right).
    \end{multline}
\end{proposition}

\begin{IEEEproof}
Applying \cref{thm:distillable-randomness-with-gamma-set}, consider that
\begin{align}
     & \frac{1}{n}    R^\varepsilon_{\leftrightarrow}(\rho^{\otimes n}_{XB})   \notag \\
     & \leq \frac{1}{n} \Gamma^\varepsilon_{\min,F}\! \left(X^n;B^n \right)_{\rho_{XB}^{\otimes n}} \\ 
     & \leq \frac{1}{n} D^\varepsilon_{\min,F}\! \left(\rho^{\otimes n}_{XB} \Vert (\rho_X \otimes \rho_B)^{\otimes n} \right)  \\
     & = D(\rho_{XB} \Vert \rho_X \otimes \rho_B) \notag \\ 
    & \quad +  \sqrt{\frac{1}{n} V(\rho_{XB} \Vert \rho_X \otimes \rho_B)} \ \Phi^{-1}(\varepsilon) + O\!\left( \frac{\log n}{n} \right) \\ 
    & = I(X;B)_\rho +\sqrt{\frac{1}{n} V(X;B)_{\rho}} \ \Phi^{-1}(\varepsilon)   + O\!\left( \frac{\log n}{n} \right),
\end{align}
where the second inequality follows because $\gamma(\rho_X \otimes \rho_B)=1$ for product states \cite[Proposition~4]{lami2022upper}, the first equality by applying \cref{thm:second-order-Dmin}, and the last equality from the definitions of mutual information and mutual information variance. 
\end{IEEEproof}

\begin{proposition}[Lower bound for uniform cq states] \label{prop:second-order-expansion-CR-some-cq-lower}
Fix $\varepsilon \in (0,1)$. For a cq state of the form 
\begin{equation}
    \rho_{XB} \coloneqq \frac{1}{M} \sum_{x=1}^M |x \rangle\!\langle x| \otimes \rho_B^x,
\end{equation}
 the 1W-LOCC-assisted distillable randomness of $\rho_{XB}$ satisfies the following:
    \begin{multline} 
        \frac{1}{n}    R^\varepsilon(\rho^{\otimes n}_{XB}) \geq  \\
        I(X;B)_\rho + \sqrt{\frac{1}{n} V(X;B)_\rho} \ \Phi^{-1}(\varepsilon) + O\!\left( \frac{\log n}{n} \right).
    \end{multline} 
\end{proposition}
\begin{IEEEproof}
   First by  applying \cref{prop:second-order-lower-bound-special-cq} for the state $\rho_{XB}^{\otimes n}$ and choosing $\eta = 1/\sqrt{n}$, consider that
   \begin{align}
       & \frac{1}{n}    R^\varepsilon(\rho^{\otimes n}_{XB})  \notag \\ 
       & \geq \frac{1}{n} I_{\min}^{\varepsilon -\eta}(X^n;B^n)_{\rho^{\otimes n}} - \frac{1}{n} \log \!\left( \frac{4 \varepsilon}{\eta^2}  \right) -\frac{1}{n}\\
       &= D(\rho_{XB} \Vert \rho_X \otimes \rho_B)  + \sqrt{\frac{1}{n} V(\rho_{XB} \Vert \rho_X \otimes \rho_B)} \ \Phi^{-1}(\varepsilon-\eta) \notag \\
       & \qquad \qquad + O\!\left( \frac{\log n}{n} \right) - \frac{1}{n} \log \!\left( \frac{4 \varepsilon}{\eta^2}  \right) \\
 & = I(X;B)_\rho + \sqrt{\frac{1}{n} V(X;B)_\rho} \ \Phi^{-1}(\varepsilon) + O\!\left( \frac{\log n}{n} \right),
   \end{align}
   where the first equality follows from \eqref{eq:lower-bnd-1-shot-dist-rand-uniform} and the second equality by \eqref{eq:hypothesis-testing-second-order}. For the latter, we choose $n$ sufficiently large so that  $\eta=1/\sqrt{n} \in(0,\varepsilon)$  and invoke a standard step in \cite[Footnote~6]{tomamichel2013hierarchy} applied to $\Phi^{-1}(\varepsilon- \eta)$, which follows from Taylor's theorem: for $f$ continuously differentiable, $c$ is a positive constant, and $n \geq n_0$, the following equality holds
   \begin{equation}
       \sqrt{n} f(x- c/ \sqrt{n})=\sqrt{n} f(x) - c f'(a)
   \end{equation}
   for some $a \in [x-c/\sqrt{n},x]$ (note that  the reasoning we used to arrive at \eqref{eq:upperbound-second-order}, in the proof of \cref{thm:second-order-Dmin} is a special case of this argument).
    The last equality follows from the definitions of mutual information and mutual information variance.  
\end{IEEEproof}

\begin{theorem}
\label{thm:second-order-expansion-CR-some-cq}
Fix $\varepsilon \in (0,1)$. For a cq state of the form 
\begin{equation} \label{eq:cq-form}
    \rho_{XB} \coloneqq \frac{1}{M} \sum_{x=1}^M |x \rangle\!\langle x| \otimes \rho_B^x,
\end{equation}
 the 1W-LOCC-assisted and LOCC-assisted distillable randomness of $\rho_{XB}$ satisfy the following:
    \begin{multline} 
        \frac{1}{n}    R^\varepsilon(\rho^{\otimes n}_{XB}) = \frac{1}{n}    R_{\leftrightarrow}^\varepsilon(\rho^{\otimes n}_{XB})  \\ = I(X;B)_\rho + \sqrt{\frac{1}{n} V(X;B)_\rho} \ \Phi^{-1}(\varepsilon) + O\!\left( \frac{\log n}{n} \right).
    \end{multline}     
\end{theorem} 

\begin{IEEEproof}
   First by applying \eqref{eq: inequality-1W-and-General}, we get 
   \begin{equation}
       \frac{1}{n}    R^\varepsilon(\rho^{\otimes n}_{XB})  \leq \frac{1}{n}    R_{\leftrightarrow}^\varepsilon(\rho^{\otimes n}_{XB}).
   \end{equation}
   Then, by applying \cref{prop:Upper-Bound-Second-Order-CR}, we arrive at the desired upper bound. By obtaining a matching lower bound from \cref{prop:second-order-expansion-CR-some-cq-lower}, we conclude the proof. 
\end{IEEEproof}

\medskip We suspect that \cref{thm:second-order-expansion-CR-some-cq} holds more generally for all cq states. To establish this finding, it seems that one would need to devise a protocol that has comparable performance to that given in the proof of \cref{prop:second-order-lower-bound-special-cq}.

\begin{remark}[Impact of feedback on distillable randomness]
\cref{thm:second-order-expansion-CR-some-cq} indicates that feedback does not improve the distillable randomness of a cq state of the form in \eqref{eq:cq-form}, even up to the second order for this class of states. This finding is in distinction to the findings of \cite{WNA20} for channel coding, in which it was shown that feedback can improve the classical communication rate up to the second order for channels with compound dispersion. The example in \cite[Remark~1]{WNA20} provides an interesting case study. If we choose the uniform distribution over the six channel input symbols, this leads to a fixed bipartite classical state shared between Alice and Bob, for which its distillable randomness cannot be improved by feedback, even up to the second order. However, in the channel coding setting, the sender can adjust the input distribution based on feedback from the receiver, and this is the mechanism underlying the improved performance found in~\cite{WNA20}.
\end{remark}

\begin{figure}
    \centering
    \includegraphics[width=\linewidth]{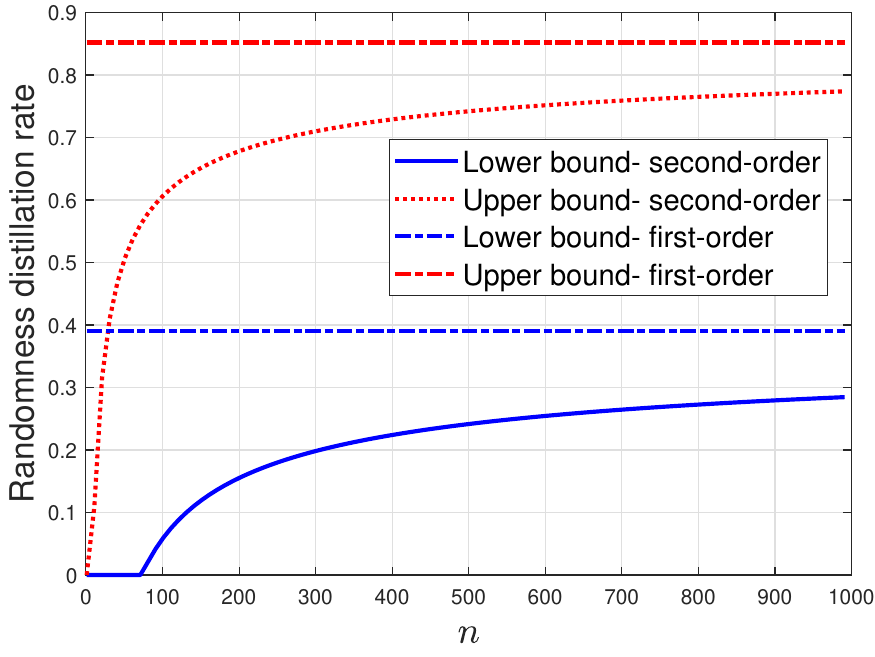}
    \caption{For fixed $\varepsilon=0.0001$, the plots depicts upper and lower bounds on the randomness distillation rate of an isotropic state $\rho_{AB}= (1-p)\Phi^d_{AB} +p \frac{I_{AB}}{d^2}$ with $p=0.3$ and $d=2$. The horizontal lines depict asymptotic values of the upper and lower bounds, while the curves depict the lower bound from \cref{prop:second-order-expansion-CR-some-cq-lower} and the upper bound from \cref{thm:Upper-Bound-Second-Order-general}.}
    \label{fig:secondOrder-isotropic}
\end{figure}

We now evaluate the upper bound from \cref{thm:Upper-Bound-Second-Order-general} and the lower bound from \cref{prop:second-order-expansion-CR-some-cq-lower} for a particular example. Suppose that Alice and Bob share an isotropic state of the form $\rho^{(d,p)}_{AB} \coloneqq (1-p)\Phi^d_{AB} + p \frac{I_{AB}}{d^2}$,
where
\begin{equation}
    \Phi^d_{AB} \coloneqq \frac{1}{d} \sum_{i,j} |i\rangle\!\langle j|_A \otimes |i\rangle\!\langle j|_B.
\end{equation}
One scheme for obtaining a lower bound on the distillable randomness of this state, as considered in \cite{lami2022upper}, is for Alice to  measure her system $A$ in the computational basis. Note that this measurement procedure has the same effect as applying a completely dephasing channel on system $A$. The resulting state is of the form in \eqref{eq:cq-form}, and so then \cref{prop:second-order-expansion-CR-some-cq-lower} applies for obtaining a lower bound on its distillable randomness. In \cref{fig:secondOrder-isotropic}, we show how the upper bound (from \cref{thm:Upper-Bound-Second-Order-general}) and the lower bound (from \cref{prop:second-order-expansion-CR-some-cq-lower}) on the LOCC-assisted distillable randomness vary within the finite $n$ regime for fixed $\varepsilon=0.0001$, when choosing $\rho_{AB}^{(d,p)}$ with $p=0.3$ and $d=2$.

\subsection{Applications to General Resource Theories}

\label{sec:apps-gen-res}

In this section, we discuss how our approach can be generalized beyond randomness distillation, to the distillation of mixed states in a general resource theory (see \cite{chitambar2019quantum} for a review). Let us consider the following general scenario: let~$\mathbb{O}$ denote the set of free channels, and let $\mathbb{F}$ denote the set of free states. Suppose that the goal of a protocol is to start from an arbitrary state $\rho$ and apply a free channel $\mathcal{M} \in \mathbb{O}$, in order to approximately distill a state from the set $\{\tau^d\}_{d \in \mathbb{Z}^+}$, where $d$ denotes the amount of the resource being distilled (a concrete example of such a set is the set of maximally classical correlated states, defined from \eqref{eq:max-class-corr-def}).  The one-shot distillable resource of $\rho$ is then defined for $\varepsilon\in[0,1]$ as follows:
\begin{equation}
    G^\varepsilon(\rho) \coloneqq \sup_{\mathcal{M} \in \mathbb{O}} \{\log_2 d : F(\mathcal{M}(\rho), \tau^d) \geq 1-\varepsilon   \}.
\end{equation} 

 Now suppose that the following generalization of \eqref{eq: gamma 1/d bound} holds:
\begin{equation}\label{eq:general-d-bound}
    \sup_{\sigma \in \mathbb{F}} F(\tau^d , \sigma) \leq \frac{1}{d}.
\end{equation}
Then the ideas presented in  \cref{Sec: Randomness Distillation} can be extended to this more general scenario. 
In particular, we arrive at the following upper bound on the one-shot distillable resource: 
\begin{equation}
  G^\varepsilon(\rho) \leq   \inf_{\sigma \in \mathbb{F}} D^\varepsilon_{\min,F}(\rho \Vert \sigma),
\end{equation}
which follows because
\begin{align}
    \log_2 d  & \leq \inf_{\sigma' \in \mathbb{F}} -\log_2 F(\tau^d , \sigma') \\
    & \leq \inf_{\sigma' \in \mathbb{F}} D^\varepsilon_{\min,F}\!\left(\cM(\rho) \Vert \sigma'\right) \\
    & \leq \inf_{\sigma \in \mathbb{F}} D^\varepsilon_{\min,F}\!\left(\cM(\rho) \Vert \cM(\sigma)\right) \\
    & \leq \inf_{\sigma \in \mathbb{F}} D^\varepsilon_{\min,F}(\rho \Vert \sigma),
\end{align}
where the first inequality follows from \eqref{eq:general-d-bound}, the second from the definition of the smooth $F$-min-relative entropy, the third inequality from the assumption that~$\mathcal{M}$ is a free channel (thus preserving the set of free states), and the last inequality from the data-processing inequality for the smooth $F$-min-relative entropy (\cref{thm:data-processing}). 

We can also obtain a second-order upper bound on the asymptotic distillable resource, by following the approach from \cref{sec:2nd-order-expansions}. Indeed,  let us 
 define the following quantities: 
\begin{equation} \label{def:first-order-term-gen}
    D_{\mathbb{F}}(\rho) \coloneqq \inf_{\sigma\in \mathbb{F} } D(\rho \Vert \sigma),
\end{equation}
where $D(\rho \Vert \sigma)$ is the quantum relative entropy from~\eqref{eq:relative-entropy-def}.
Then, denote $\Pi_{\mathbb{F}} \subseteq \mathbb{F} $ as the set of states achieving the infimum in $D_{\mathbb{F}}(\rho)$. From that, we define the following variance quantity: 
\begin{equation}
\label{def:second-order-term-gen}
       V_{\mathbb{F}}^\varepsilon(\rho) \coloneqq 
  \begin{cases} \inf_{\sigma \in \Pi_{\mathbb{F}} } V(\rho \Vert \sigma)& \mbox{if} \ \varepsilon \geq \frac{1}{2},\\
 \sup_{\sigma \in \Pi_{\mathbb{F}}}  V(\rho \Vert \sigma) & \mbox{if} \ \varepsilon < \frac{1}{2}.
\end{cases}
\end{equation}
 Fix $\varepsilon \in (0,1)$.   Then the distillable resource of a  state $\rho$ is bounded from above as follows: 
    \begin{equation} 
         \frac{1}{n}    G^\varepsilon(\rho^{\otimes n})   \leq 
         D_{\mathbb{F}}(\rho) + \sqrt{\frac{1}{n}V_{\mathbb{F}}^\varepsilon(\rho)} \ \Phi^{-1}(\varepsilon) +  O\!\left( \frac{\log n}{n} \right).
    \end{equation}

\section{Computational Analysis} \label{Sec:Computational-analysis}

In this section, we provide techniques based on semi-definite programs (SDPs) to quantify the smooth $F$-min-relative entropy and other related quantities, including the smooth max-relative entropy and smooth conditional min-entropy. 

First, we present a bilinear program to evaluate the smooth $F$-min-relative entropy, which follows from the SDP formulation of fidelity in \cite{watrous2012simpler}. 

\begin{proposition}
\label{prop:compute-smooth-min-relative-entropy}
Given a state $\rho$, a PSD operator $\sigma$, and $\varepsilon\in(0,1)$, the smooth $F$-min-relative entropy can be written in terms of the following
optimization:
\begin{equation}
    D^\varepsilon_{\min,F}(\rho \Vert \sigma) = -2 \log_2 a^\star,
\end{equation}
where 
\begin{align}
a^\star 
& =\frac{1}{2}\inf_{\substack{\widetilde{\rho}\geq0,\\Y,Z\geq0,\\X\in
\mathcal{L}(\mathcal{H})}}\left\{
\begin{array}
[c]{c}%
\operatorname{Tr}[Y\widetilde{\rho}]+\operatorname{Tr}[Z\sigma]:\\%
\begin{bmatrix}
Y & I\\
I & Z
\end{bmatrix}
\geq0,\\
\operatorname{Re}[\operatorname{Tr}[X]]\geq\sqrt{1-\varepsilon},\\%
\begin{bmatrix}
\widetilde{\rho} & X\\
X^\dag & \rho
\end{bmatrix}
\geq0,\\
\operatorname{Tr}[\widetilde{\rho}]\leq1
\end{array}
\right\}  .
\label{eq:bilinear-prog-smooth-f-min}
\end{align}
\end{proposition}

\begin{IEEEproof}
    From the definition of the smooth $F$-min relative entropy, consider that
    \begin{equation}
    D^\varepsilon_{\min,F}(\rho \Vert \sigma) = -\log_2 (a^\star)^2,  
    \label{eq:to-compute-smooth-min-relative-entropy}
    \end{equation}
    where
    \begin{equation}
        a^\star  \coloneqq \inf_{\widetilde{\rho}\in\mathcal{D}_{\leq}}\left\{  \sqrt{F}(\widetilde{\rho
},\sigma):\sqrt{F}(\widetilde{\rho},\rho)\geq\sqrt{1-\varepsilon}\right\}.
\label{eq:a-star-def-orig}
    \end{equation}
    Recall that the root fidelity 
    $
       \sqrt{F}(\rho,\sigma) \coloneqq \left\| \sqrt{\rho} \sqrt{\sigma}\right\|_1  
    $
    has the following primal and dual SDP characterizations \cite{watrous2012simpler} (see also \cite[Proposition~6.6]{khatri2020principles}): 
    \begin{align}
         \sqrt{F}(\rho,\sigma) 
         & =\sup_{X\in
\mathcal{L}(\mathcal{H})}\left\{
\operatorname{Re}[\operatorname{Tr}[X]]:%
\begin{bmatrix}
\rho & X\\
X^\dag & \sigma
\end{bmatrix}
\geq0
\right\} \\ 
&= \frac{1}{2}\inf_{Y,Z \geq 0}\left\{
\Tr[Y \rho]+ \Tr[Z \sigma]:%
\begin{bmatrix}
Y & I \\
I & Z
\end{bmatrix}
\geq0
\right\}.
\label{eq:dual-SDP-fid-inf}
\end{align}
We then find the following expression for the  negative root fidelity:
\begin{align}
 \!\!\!\!\!\!\!\! -\sqrt{F}(\rho,\widetilde{\rho}) &  =-\sup_{X\in\mathcal{L}(\mathcal{H}%
)}\left\{  \operatorname{Re}[\operatorname{Tr}[X]]:%
\begin{bmatrix}
\rho & X\\
X^{\dag} & \widetilde{\rho}%
\end{bmatrix}
\geq0\right\}  \\
&  =\inf_{X\in\mathcal{L}(\mathcal{H})}\left\{  -\operatorname{Re}%
[\operatorname{Tr}[X]]:%
\begin{bmatrix}
\rho & X\\
X^{\dag} & \widetilde{\rho}%
\end{bmatrix}
\geq0\right\} . \label{eq:dualSDP_fidelity-negative}
\end{align}
With these expressions in hand, consider from \eqref{eq:a-star-def-orig} and \eqref{eq:dual-SDP-fid-inf} that
\begin{align}
 a^\star 
& =
\frac{1}{2}\inf_{\widetilde{\rho}%
\geq0,Y,Z \geq 0}\left\{
\begin{array}
[c]{c}%
\Tr[Y\widetilde{\rho}] +\Tr[Z \rho] : \\
\sqrt{F}(\widetilde{\rho},\rho) \geq \sqrt{1-\varepsilon}, \\
\begin{bmatrix}
Y & I\\
I & Z%
\end{bmatrix}
\geq0,\\
\Tr[\widetilde{\rho}]\leq 1
\end{array}
\right\}  \\
 &  =
\frac{1}{2}\inf_{\widetilde{\rho}%
\geq0,Y,Z \geq 0}\left\{
\begin{array}
[c]{c}%
\Tr[Y\widetilde{\rho}] +\Tr[Z \rho] : \\
-\sqrt{F}(\widetilde{\rho},\rho) \leq  -\sqrt{1-\varepsilon}, \\
\begin{bmatrix}
Y & I\\
I & Z%
\end{bmatrix}
\geq0,\\
\Tr[\widetilde{\rho}]\leq 1
\end{array}
\right\} \\
& = \frac{1}{2}\inf_{\substack{\widetilde{\rho}\geq0,\\Y,Z\geq0,\\X\in
\mathcal{L}(\mathcal{H})}}\left\{
\begin{array}
[c]{c}%
\operatorname{Tr}[Y\widetilde{\rho}]+\operatorname{Tr}[Z\sigma]:\\%
\begin{bmatrix}
Y & I\\
I & Z
\end{bmatrix}
\geq0,\\
-\operatorname{Re}[\operatorname{Tr}[X]]\leq-\sqrt{1-\varepsilon},\\%
\begin{bmatrix}
\widetilde{\rho} & X\\
X^\dag & \rho
\end{bmatrix}
\geq0,\\
\operatorname{Tr}[\widetilde{\rho}]\leq1
\end{array}
\right\} .
\end{align}
In the last line, we replaced the inequality $-\sqrt{F}(\widetilde{\rho},\rho) \leq -\sqrt{1-\varepsilon}$ with the optimization in  \eqref{eq:dualSDP_fidelity-negative}, to conclude the proof. 
\end{IEEEproof}

Note that the optimization problem in \eqref{eq:bilinear-prog-smooth-f-min} is not an SDP, but it is rather a bilinear program, due to the bilinear term $\Tr[Y\widetilde{\rho}]$ in the objective function in \eqref{eq:bilinear-prog-smooth-f-min}. Thus, to estimate the smooth $F$-min relative entropy, we can employ the seesaw or mountain-climbing algorithm \cite{SeeSawkonno1976cutting}, which results in a lower bound on the smooth $F$-min-relative entropy.

We apply the seesaw algorithm as follows. Set $k=0$. For fixed $\widetilde{\rho}=\widetilde{\rho}_k$, the optimization in \eqref{eq:bilinear-prog-smooth-f-min} is an SDP with the additional constraint $\widetilde{\rho}=\widetilde{\rho}_k$. By solving that SDP, we can find a $Y$ that achieves the optimum for the objective function, which we denote as $Y_k$. Next, by fixing $Y=Y_k$, we solve the respective SDP and find the optimum $\widetilde{\rho}_{k+1}$. Then, this iterative process is continued for a fixed number of iterations. In a finite (yet not necessarily polynomial) number of iterations, the algorithm converges to the optimum value \cite{SeeSawkonno1976cutting}. In practice, it is guaranteed that we arrive at an upper bound on \eqref{eq:bilinear-prog-smooth-f-min} by this method, which will in turn provide a lower bound on the smooth $F$-min-relative entropy. Note that one can employ the more advanced algorithm from \cite{Bilinear_huber2019jointly} as well for this purpose.

We use the seesaw method described above to investigate the tightness of the lower bound from \cref{prop:connect-smooth-max}. In particular, we obtain various upper bounds for the smooth $F$-min-relative entropy by varying $\delta$ in 
\begin{equation}
\label{eq:connection-Dmax-delta}
    D^\varepsilon_{\min,F} (\rho  \Vert \sigma) 
\leq  D_{\max}^{1-\varepsilon-\delta}(\rho\Vert\sigma)+ \log_{2} \!\left(
\frac{1}{1-f(\varepsilon,\delta)}\right), 
\end{equation} 
where $f(\varepsilon,\delta)$ is given in \eqref{eq:f(eps,delta)}. 

To compute the upper bound in \eqref{eq:connection-Dmax-delta}, it is required to compute the smooth max-relative entropy. For this purpose, we derive an SDP\ for the smooth max-relative entropy with
fidelity smoothing, which may be of independent interest. We consider two variants of that: $\widehat{D}_{\max}^{\varepsilon}(\rho\Vert\sigma)$ with smoothing over normalized states and ${D}_{\max}^{\varepsilon}(\rho\Vert\sigma)$ with smoothing over sub-normalized states. Then, we use the latter variant to compute the right-hand side of \eqref{eq:connection-Dmax-delta}.

\begin{proposition} \label{prop: SDP for smooth max normalized}
The smooth max-relative entropy with fidelity smoothing over normalized states has the following
primal and dual SDP characterizations for a state $\rho$ and PSD operator~$\sigma$:
\begin{align}
\label{eq:SDP-for-smooth-max-fidelity}
& \widehat{D}_{\max}^{\varepsilon}(\rho\Vert\sigma) \cr
& =
\log_{2}\inf_{\widetilde{\rho}%
\geq0,\lambda\geq0,X\in\mathcal{L}(\mathcal{H})}\left\{
\begin{array}
[c]{c}%
\lambda:\widetilde{\rho}\leq\lambda\sigma,\operatorname{Tr}[\widetilde{\rho
}]=1,\\
\operatorname{Re}[\operatorname{Tr}[X]]\geq\sqrt{1-\varepsilon},\\%
\begin{bmatrix}
\rho & X\\
X^{\dag} & \widetilde{\rho}%
\end{bmatrix}
\geq0
\end{array}
\right\}  \\
& = \log_{2}\sup_{W,\nu,Z\geq0,\mu
\in\mathbb{R}}\left\{
\begin{array}
[c]{c}%
\mu+2\nu\sqrt{1-\varepsilon}-\operatorname{Tr}[Z\rho]:\\
\operatorname{Tr}[W\sigma]\leq1,\\%
\begin{bmatrix}
Z & \nu I\\
\nu I & W-\mu I
\end{bmatrix}
\geq0
\end{array}
\right\}  .
\end{align}

\end{proposition}

\begin{IEEEproof}
    See Appendix~\ref{proof: SDP for smooth max normalized}.
\end{IEEEproof}

\medskip

Following a similar approach, we obtain an SDP for ${D}_{\max}^{\varepsilon}(\rho\Vert\sigma) $. 

\begin{proposition} \label{prop:SDP-for-smooth-max-with-sub-normalized}
The smooth max-relative entropy with fidelity smoothing over sub-normalized states has the following
primal and dual SDP characterizations for a state $\rho$, PSD operator~$\sigma$, and $\varepsilon\in(0,1)$:
\begin{align}
&{D}_{\max}^{\varepsilon}(\rho\Vert\sigma) \cr
& =
\log_{2}\inf_{\widetilde{\rho}%
\geq0,\lambda\geq0,X\in\mathcal{L}(\mathcal{H})}\left\{
\begin{array}
[c]{c}%
\lambda:\widetilde{\rho}\leq\lambda\sigma,\operatorname{Tr}[\widetilde{\rho
}] \leq 1,\\
\operatorname{Re}[\operatorname{Tr}[X]]\geq\sqrt{1-\varepsilon},\\%
\begin{bmatrix}
\rho & X\\
X^{\dag} & \widetilde{\rho}%
\end{bmatrix}
\geq0
\end{array}
\right\}  ,\\
&  = \log_{2}\sup_{W,\nu,Z,\mu \geq0}\left\{
\begin{array}
[c]{c}%
-\mu+2\nu\sqrt{1-\varepsilon}-\operatorname{Tr}[Z\rho]:\\
\operatorname{Tr}[W\sigma]\leq1,\\%
\begin{bmatrix}
Z & \nu I\\
\nu I & W+\mu I
\end{bmatrix}
\geq0
\end{array}
\right\}  .
\end{align}
\end{proposition}

\begin{IEEEproof}
    We omit the proof since it is similar to the proof of \cref{prop: SDP for smooth max normalized}, except that we remove the constraint $\Tr[\tilde{\rho}] \geq 1$. 
\end{IEEEproof}

\medskip
Utilizing the methods explained above (specifically in Propositions~\ref{prop:compute-smooth-min-relative-entropy} and \ref{prop:SDP-for-smooth-max-with-sub-normalized}), Figure~\ref{fig:UpperBoundLowerBound_D_2} shows the tightness of the lower bound on the smooth $F$-min-relative entropy of two random quantum states in a Hilbert space of dimension two, by employing the seesaw method for ten iterations. 
In addition, Figure~\ref{fig:UpperBoundLowerBound_D_4} presents the obtained results under the same setting, for two random quantum states in a Hilbert space of dimension four.

Due to the close link between smooth max-relative entropy and smooth conditional min-entropy, our SDP results apply to this latter quantity as well. We highlight this observation below, as it may be of independent interest for future work.

\begin{figure}
    \centering
    \includegraphics[width=\linewidth]{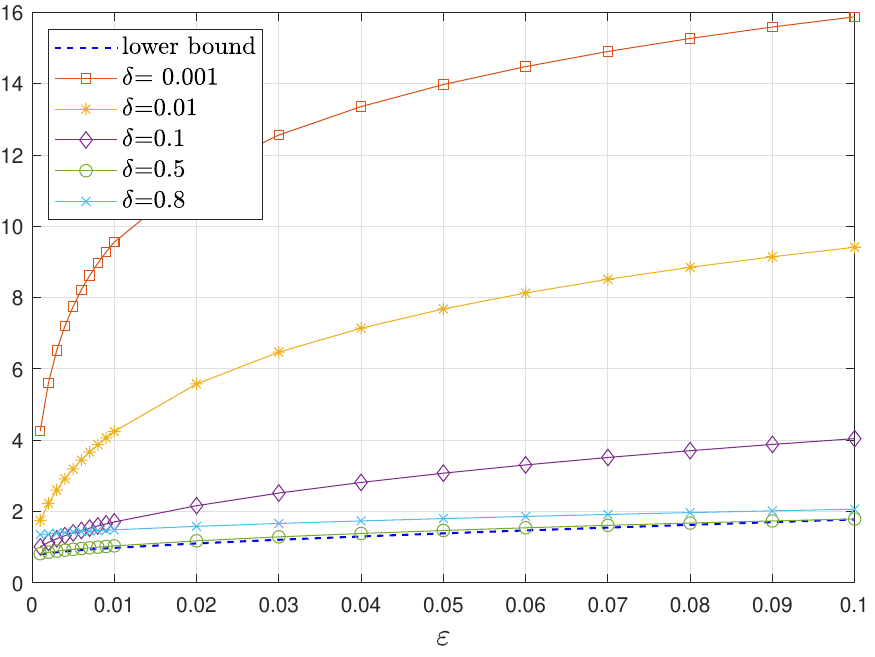}
    \caption{This plot shows the tightness of the lower bound on $D^\varepsilon_{\min,F}(\rho \Vert \sigma)$ obtained from running the seesaw algorithm for 10 iterations for each $\varepsilon$, where $\rho,\sigma$ are random quantum states in a Hilbert space of dimension two.  
    By changing $\delta$, the upper bound $D_{\max}^{1-\varepsilon-\delta}(\rho\Vert\sigma)+ \log_{2} \!\left(
\frac{1}{1-f(\varepsilon,\delta)}\right)$ is also shown. }
    \label{fig:UpperBoundLowerBound_D_2}
\end{figure}

\begin{figure}
    \centering
    \includegraphics[width=\linewidth]{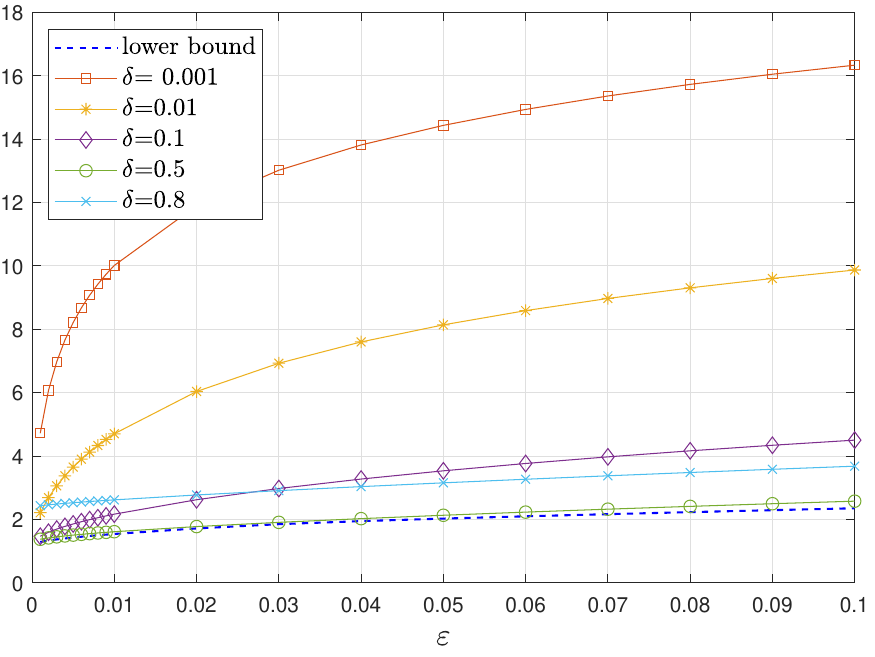}
    \caption{This plot shows the tightness of the lower bound on $D^\varepsilon_{\min,F}(\rho \Vert \sigma)$ obtained from running the seesaw algorithm for ten iterations for each $\varepsilon$, where $\rho,\sigma$ are random quantum states in a Hilbert space of dimension four.
    By changing $\delta$, the upper bound $D_{\max}^{1-\varepsilon-\delta}(\rho\Vert\sigma)+ \log_{2} \!\left(
\frac{1}{1-f(\varepsilon,\delta)}\right)$ is also shown. }
    \label{fig:UpperBoundLowerBound_D_4}
\end{figure}

\begin{remark}[Conditional smooth min-entropy]
   By Definitions~11 and 12 of \cite{tomamichel2010duality}, for $\rho_{AB} \in \cD(\cH_{AB})$, the smooth conditional min-entropy is defined as 
\begin{multline}
    H^\varepsilon_{\min}(A|B)_\rho \coloneqq  \\ \max_{\widetilde{\rho}_{AB} \in \cB^\varepsilon(\rho_{AB})} \max_{\sigma_B \in \cD(\cH_B)} - D_{\max}(\widetilde{\rho}_{AB} \Vert I_A \otimes \sigma_B),
\end{multline}
where 
\begin{equation}
    \cB^\varepsilon(\rho) \coloneqq \{ \tau \in \cD_{\leq}: P(\tau,\rho) \leq \varepsilon \},
\end{equation}
with $P(\tau,\rho) \coloneqq \sqrt{1-F(\tau,\rho)}$.
Note that we can equivalently define smooth conditional  min-entropy by 
\begin{equation} \label{eq: Hmin smooth using Dmax}
     H^\varepsilon_{\min}(A|B)_\rho \coloneqq - \min_{\sigma_B \in \cD(\cH_B)} D^{\varepsilon^2}_{\max}(\rho_{AB} \Vert I_A \otimes \sigma_B),
\end{equation}
with the association of smooth max-relative entropy in \eqref{eq:smooth-dmax-def}.

An SDP\ for $ H_{\min}^{\varepsilon}(A|B)_{\rho}$ was first given in \cite{schaffner2008robust} if a purification of
$\rho_{AB}$ is available. In this work, \cref{prop:SDP-for-smooth-max-with-sub-normalized}  together with \eqref{eq: Hmin smooth using Dmax} lead to an SDP for smooth conditional min-entropy directly in terms of the state $\rho_{AB}$, as presented in the next corollary.
\end{remark}

\begin{corollary}\label{cor: SDP for smooth conditional min entropy}
    Let $\rho_{AB}$ be a state and fix $\varepsilon\in(0,1)$. The smooth conditional min-entropy has the following SDP characterization:
    \begin{align}
& H_{\min}^{\varepsilon}(A|B)_{\rho} \cr
&  =-\log_{2} \hspace{-2mm}\inf_{\substack{\widetilde{\rho}_{AB}\geq 0, \\ S_{B}\geq0, \\ X_{AB}\in\mathcal{L}%
(\mathcal{H})}}\left\{
\begin{array}
[c]{c}%
\operatorname{Tr}[S_{B}]:\widetilde{\rho}_{AB}\leq I_{A}\otimes S_{B}%
, \\ \operatorname{Tr}[\widetilde{\rho}_{AB}] \leq 1,\\
\operatorname{Re}[\operatorname{Tr}[X_{AB}]]\geq\sqrt{1-\varepsilon^2},\\%
\begin{bmatrix}
\rho_{AB} & X_{AB}\\
X_{AB}^{\dag} & \widetilde{\rho}_{AB}%
\end{bmatrix}
\geq0
\end{array}
\right\} \label{eq: second primal smooth min entropy conditional} \\
& = -\log_2 \sup_{\substack{W_{AB}\geq0,\\
Z_{AB} \geq 0,\\ \nu\geq0,\mu \geq 0}}\left\{
\begin{array}
[c]{c}%
-\mu+2\nu\sqrt{1-\varepsilon^2}-\operatorname{Tr}[Z_{AB}\rho_{AB}]:\\
\operatorname{Tr}_{A}[W_{AB}]\leq I_{B},\\%
\begin{bmatrix}
Z_{AB} & \nu I_{AB}\\
\nu I_{AB} & W_{AB}+\mu I_{AB}%
\end{bmatrix}
\geq0
\end{array}
\right\} \nonumber .\\
&  \label{eq:dual-of-smooth-conditional-entropy-main}
\end{align}
\end{corollary}
\begin{IEEEproof}
    We can write
    \begin{multline}
 H_{\min}^{\varepsilon}(A|B)_{\rho} \\
 =-\log_{2}\inf_{\substack{\widetilde{\rho}_{AB}\geq 0, \\ \sigma_{B}\geq0, \\ \lambda\geq
0, \\ X_{AB}\in\mathcal{L}(\mathcal{H})}}\left\{
\begin{array}
[c]{c}%
\lambda:\widetilde{\rho}_{AB}\leq\lambda I_{A}\otimes\sigma_{B}%
, \\ \operatorname{Tr}[\widetilde{\rho}_{AB}] \leq 1,\\
\operatorname{Re}[\operatorname{Tr}[X_{AB}]]\geq\sqrt{1-\varepsilon^2},\\%
\begin{bmatrix}
\rho_{AB} & X_{AB} \\
X_{AB}^{\dag} & \widetilde{\rho}_{AB}%
\end{bmatrix}
\geq0,\\
\operatorname{Tr}[\sigma_{B}]=1
\end{array}
\right\} \quad \label{eq: primal smooth min entropy conditional} 
\end{multline}
    as a direct result of the SDP for smooth max-relative entropy in \cref{prop:SDP-for-smooth-max-with-sub-normalized} and its connection to smooth conditional min-entropy in \eqref{eq: Hmin smooth using Dmax}. The characterization in \eqref{eq: second primal smooth min entropy conditional} follows from the substitution $S_B = \lambda \sigma_B$ in \eqref{eq: primal smooth min entropy conditional}.
    The dual SDP for $H^\varepsilon_{\min}(A| B)_\rho$  in \eqref{eq:dual-of-smooth-conditional-entropy-main} is proved in Appendix~\ref{proof:dual-of-smooth-conditional-min-entropy}.
\end{IEEEproof}

\section{Conclusion}

\label{Sec:Conclusion}

In conclusion, we provided a comprehensive study of the fidelity-based smooth min-relative entropy. In particular, we proved some of its basic properties, including a proof of data processing independent from that of \cite[Theorem~3]{RT22}, which is of interest for several operational tasks. With the use of the derived relationships between smooth $F$-min-relative entropy and other quantum information-theoretic quantities, we derived the second-order asymptotics of the smooth $F$-min-relative entropy and all smooth sandwiched R\'enyi relative entropies. We explored applications where this quantity arises, with a focus on randomness distillation, establishing an upper bound on the one-shot distillable randomness that improves upon the prior one from \cite{lami2022upper}. Furthermore, we obtained a second-order expansion of this upper bound, as well as the precise second-order asymptotics of the distillable randomness of particular classical--quantum states.
Lastly, we designed methods to estimate the smooth $F$-min-relative entropy and showed that the estimates are sufficiently tight for some examples. To that end, we presented SDPs to compute the smooth max-relative entropy and smooth conditional min-entropy, which may be of independent interest. 

Some future research directions are as follows. It would be interesting to analyze the third-order asymptotics
and large deviations associated with the smooth $F$-min-relative entropy. It might also be possible to devise efficient computational methods for it and tighter connections to other information-theoretic quantities. To that end, understanding which parameter $\varepsilon_1$ achieves the tightest bound in \cref{prop:connect-smooth-max} when $\varepsilon_2$ is fixed is a possible direction. {We have shown here that the smooth $F$-min-relative entropy is the core quantity underlying an upper bound on the one-shot distillable randomness, but  finding an information-theoretic task that provides an operational interpretation of the smooth $F$-min-relative entropy itself is an open research question.} 

Another interesting future direction is to follow the observation made after \eqref{eq:smooth-min-relative-entropy} in the introduction. Indeed, we see that the main difference between the hypothesis testing relative entropy and the smooth $F$-min relative entropy is the replacement of the quantities $\operatorname{Tr}[\Lambda \rho]$ and $\operatorname{Tr}[\Lambda \sigma]$ with $F(\tilde{\rho}, \rho)$ and $F(\tilde{\rho},\sigma)$, respectively. Since $\operatorname{Tr}[\Lambda \rho]$ and $\operatorname{Tr}[\Lambda \sigma]$ are related to Type~I and Type~II error probabilities in quantum hypothesis testing, one could consider a variant of quantum hypothesis testing in which the ``error probabilities'' are then related to $F(\tilde{\rho}, \rho)$ and $F(\tilde{\rho},\sigma)$. Specifically, one could consider  $1- F(\tilde{\rho}, \rho)$ to be analogous to a Type~I error probability and $F(\tilde{\rho},\sigma)$ to be analogous to a Type~II error probability. Under this perspective, \cref{thm:second-order-Dmin} demonstrates that we have already determined the second-order asymptotics of this variant of the traditional asymmetric quantum hypothesis testing task and that they are the same as the second-order asymptotics in the standard setting of asymmetric quantum hypothesis testing. What remains open is to determine the asymptotics of a variant of symmetric hypothesis testing. That is, for fixed $\lambda\in(0,1)$, what is the following quantity equal to?
\begin{equation}
    \lim_{n\to \infty} - \frac{1}{n} \ln p_e^n \label{eq:sym-err-exp-conj}
\end{equation}
where
\begin{equation}
    p_e^n \coloneqq \inf_{\tilde{\rho}^{(n)} \in \mathcal{D}_{\leq}} \lambda(1-F(\tilde{\rho}^{(n)}, \rho^{\otimes n})) + (1-\lambda) F(\tilde{\rho}^{(n)}, \sigma^{\otimes n}).
\end{equation}
Based on known results in symmetric quantum hypothesis testing \cite{nussbaum_2009,audenaert_2007} and the aforementioned coincidence for the asymmetric setting, one might guess that \eqref{eq:sym-err-exp-conj} would be equal to the quantum Chernoff divergence, but this remains an intriguing open question for future work.

\section*{Acknowledgment}
The authors acknowledge helpful discussions with Mario Berta, Dhrumil Patel, Soorya Rethinasamy, and Aaron Wagner.

 \bibliographystyle{ieeetr}
\bibliography{reference}

\begin{thebibliography}{10}

\bibitem{khatri2020principles}
S.~Khatri and M.~M. Wilde, ``Principles of quantum communication theory: A modern approach,'' 2020.
\newblock arXiv:2011.04672v1.

\bibitem{chitambar2019quantum}
E.~Chitambar and G.~Gour, ``Quantum resource theories,'' {\em Reviews of Modern Physics}, vol.~91, no.~2, p.~025001, 2019.

\bibitem{datta2009min}
N.~Datta, ``Min-and max-relative entropies and a new entanglement monotone,'' {\em IEEE Transactions on Information Theory}, vol.~55, no.~6, pp.~2816--2826, 2009.

\bibitem{uhlmann1976transition}
A.~Uhlmann, ``The “transition probability” in the state space of a *-algebra,'' {\em Reports on Mathematical Physics}, vol.~9, no.~2, pp.~273--279, 1976.

\bibitem{wang2019resource}
X.~Wang and M.~M. Wilde, ``Resource theory of asymmetric distinguishability,'' {\em Physical Review Research}, vol.~1, no.~3, p.~033170, 2019.

\bibitem{blahut74}
R.~Blahut, ``Hypothesis testing and information theory,'' {\em IEEE Transactions on Information Theory}, vol.~20, no.~4, pp.~405--417, 1974.

\bibitem{wang19channels}
X.~Wang and M.~M. Wilde, ``Resource theory of asymmetric distinguishability for quantum channels,'' {\em Physical Review Research}, vol.~1, p.~033169, 2019.

\bibitem{dupuis2014generalizedDmin}
F.~Dupuis, L.~Kraemer, P.~Faist, J.~M. Renes, and R.~Renner, ``Generalized entropies,'' in {\em International Congress on Mathematical Physics}, pp.~134--153, World Scientific, 2014.

\bibitem{muller2013quantum}
M.~M{\"u}ller-Lennert, F.~Dupuis, O.~Szehr, S.~Fehr, and M.~Tomamichel, ``On quantum {R\'e}nyi entropies: A new generalization and some properties,'' {\em Journal of Mathematical Physics}, vol.~54, no.~12, p.~122203, 2013.

\bibitem{wilde2014strong}
M.~M. Wilde, A.~Winter, and D.~Yang, ``Strong converse for the classical capacity of entanglement-breaking and {H}adamard channels via a sandwiched {R\'e}nyi relative entropy,'' {\em Communications in Mathematical Physics}, vol.~331, pp.~593--622, 2014.

\bibitem{faist2016quantumDmin}
P.~Faist, {\em Quantum coarse-graining: An information-theoretic approach to thermodynamics}.
\newblock PhD thesis, ETH Zurich, 2016.
\newblock arXiv:1607.03104.

\bibitem{zhao2019oneDmin}
Q.~Zhao, Y.~Liu, X.~Yuan, E.~Chitambar, and A.~Winter, ``One-shot coherence distillation: towards completing the picture,'' {\em IEEE Transactions on Information Theory}, vol.~65, no.~10, pp.~6441--6453, 2019.

\bibitem{ramakrishnan2023moderate}
N.~Ramakrishnan, M.~Tomamichel, and M.~Berta, ``Moderate deviation expansion for fully quantum tasks,'' {\em IEEE Transactions on Information Theory}, vol.~69, no.~8, pp.~5041--5059, 2023.

\bibitem{RT22}
R.~Rubboli and M.~Tomamichel, ``Fundamental limits on correlated catalytic state transformations,'' {\em Physical Review Letters}, vol.~129, p.~120506, Sept. 2022.

\bibitem{RW04}
R.~Renner and S.~Wolf, ``Smooth {R}\'enyi entropy and applications,'' in {\em International Symposium on Information Theory}, p.~233, 2004.

\bibitem{RennerThesis}
R.~Renner, {\em Security of Quantum Key Distribution}.
\newblock PhD thesis, ETH Z\"urich, 2005.

\bibitem{BD10}
F.~Buscemi and N.~Datta, ``The quantum capacity of channels with arbitrarily correlated noise,'' {\em IEEE Transactions on Information Theory}, vol.~56, pp.~1447--1460, 2010.

\bibitem{brandao2011one}
F.~G. S.~L. Brandao and N.~Datta, ``One-shot rates for entanglement manipulation under non-entangling maps,'' {\em IEEE Transactions on Information Theory}, vol.~57, no.~3, pp.~1754--1760, 2011.

\bibitem{WR12}
L.~Wang and R.~Renner, ``One-shot classical-quantum capacity and hypothesis testing,'' {\em Physical Review Letters}, vol.~108, p.~200501, 2012.

\bibitem{bu2018max}
K.~Bu, U.~Singh, S.-M. Fei, A.~K. Pati, and J.~Wu, ``Maximum relative entropy of coherence: An operational coherence measure,'' {\em Physical Review Letters}, vol.~119, p.~150405, 2017.

\bibitem{LiuZi-WenPhd18}
Z.-W. Liu, {\em On quantum randomness and quantum resources}.
\newblock PhD thesis, Massachusetts Institute of Technology, 2018.

\bibitem{tomamichel2013hierarchy}
M.~Tomamichel and M.~Hayashi, ``A hierarchy of information quantities for finite block length analysis of quantum tasks,'' {\em IEEE Transactions on Information Theory}, vol.~59, no.~11, pp.~7693--7710, 2013.

\bibitem{li2014second}
K.~Li, ``Second-order asymptotics for quantum hypothesis testing,'' {\em Annals of Statistics}, vol.~42, no.~1, pp.~171--189, 2014.

\bibitem{datta2014second}
N.~Datta and F.~Leditzky, ``Second-order asymptotics for source coding, dense coding, and pure-state entanglement conversions,'' {\em IEEE Transactions on Information Theory}, vol.~61, no.~1, pp.~582--608, 2014.

\bibitem{DTW14}
N.~Datta, M.~Tomamichel, and M.~M. Wilde, ``On the second-order asymptotics for entanglement-assisted communication,'' {\em Quantum Information Processing}, vol.~15, pp.~2569--2591, 2016.

\bibitem{datta2016second}
N.~Datta, Y.~Pautrat, and C.~Rouz{\'e}, ``Second-order asymptotics for quantum hypothesis testing in settings beyond iid—quantum lattice systems and more,'' {\em Journal of Mathematical Physics}, vol.~57, no.~6, p.~062207, 2016.

\bibitem{lami2022upper}
L.~Lami, B.~Regula, X.~Wang, and M.~M. Wilde, ``Upper bounds on the distillable randomness of bipartite quantum states,'' in {\em Proceedings of the 2023 Information Theory Workshop}, 2023.
\newblock arXiv:2212.09073.

\bibitem{SW12}
N.~Sharma and N.~A. Warsi, ``On the strong converses for the quantum channel capacity theorems,'' {\em arXiv:1205.1712}, May 2012.

\bibitem{FL13}
R.~L. Frank and E.~H. Lieb, ``Monotonicity of a relative {R\'enyi} entropy,'' {\em Journal of Mathematical Physics}, vol.~54, p.~122201, 2013.

\bibitem{W18opt}
M.~M. Wilde, ``Optimized quantum $f$-divergences and data processing,'' {\em Journal of Physics A}, vol.~51, p.~374002, 2018.

\bibitem{W18optISIT}
M.~M. Wilde, ``Optimized quantum f-divergences,'' in {\em IEEE International Symposium on Information Theory (ISIT)}, pp.~2481--2485, 2018.

\bibitem{U62}
H.~Umegaki, ``Conditional expectations in an operator algebra {IV} (entropy and information),'' {\em Kodai Mathematical Seminar Reports}, vol.~14, pp.~59--85, 1962.

\bibitem{P85}
D.~Petz, ``{Quasi-entropies for States of a von Neumann Algebra},'' {\em Publications of the Research Institute for Mathematical Sciences}, vol.~21, pp.~787--800, 1985.

\bibitem{P86}
D.~Petz, ``Quasi-entropies for finite quantum systems,'' {\em Reports in Mathematical Physics}, vol.~23, pp.~57--65, 1986.

\bibitem{tomamichel2015quantum}
M.~Tomamichel, {\em Quantum Information Processing with Finite Resources: Mathematical Foundations}, vol.~5.
\newblock Springer, 2015.

\bibitem{kaur2017upper}
E.~Kaur and M.~M. Wilde, ``Upper bounds on secret-key agreement over lossy thermal bosonic channels,'' {\em Physical Review A}, vol.~96, no.~6, p.~062318, 2017.

\bibitem{Sti55}
W.~F. Stinespring, ``{Positive Functions on C*-Algebras},'' {\em Proceedings of the American Mathematical Society}, vol.~6, pp.~211--216, 1955.

\bibitem{W21second}
M.~M. Wilde, ``Second law of entanglement dynamics for the non-asymptotic regime,'' in {\em 2021 IEEE Information Theory Workshop}, pp.~1--6, 2021.

\bibitem{takagi2022one}
R.~Takagi, B.~Regula, and M.~M. Wilde, ``One-shot yield-cost relations in general quantum resource theories,'' {\em PRX Quantum}, vol.~3, no.~1, p.~010348, 2022.

\bibitem{Wbook17}
M.~M. Wilde, {\em Quantum Information Theory}.
\newblock Cambridge University Press, second~ed., 2017.

\bibitem{qi2018applications}
H.~Qi, Q.~Wang, and M.~M. Wilde, ``Applications of position-based coding to classical communication over quantum channels,'' {\em Journal of Physics A: Mathematical and Theoretical}, vol.~51, no.~44, p.~444002, 2018.

\bibitem{wilde2017position}
M.~M. Wilde, ``Position-based coding and convex splitting for private communication over quantum channels,'' {\em Quantum Information Processing}, vol.~16, p.~264, 2017.

\bibitem{khabbazi2019union}
S.~Khabbazi~Oskouei, S.~Mancini, and M.~M. Wilde, ``Union bound for quantum information processing,'' {\em Proceedings of the Royal Society A}, vol.~475, no.~2221, p.~20180612, 2019.

\bibitem{ABJT19}
A.~Anshu, M.~Berta, R.~Jain, and M.~Tomamichel, ``{A minimax approach to one-shot entropy inequalities},'' {\em Journal of Mathematical Physics}, vol.~60, p.~122201, 2019.

\bibitem{R02}
A.~E. Rastegin, ``Relative error of state-dependent cloning,'' {\em Physical Review A}, vol.~66, p.~042304, 2002.

\bibitem{R03}
A.~E. Rastegin, ``A lower bound on the relative error of mixed-state cloning and related operations,'' {\em Journal of Optics B: Quantum and Semiclassical Optics}, vol.~5, p.~S647, 2003.

\bibitem{GLN04}
A.~Gilchrist, N.~K. Langford, and M.~A. Nielsen, ``Distance measures to compare real and ideal quantum processes,'' {\em Physical Review A}, vol.~71, p.~062310, 2005.

\bibitem{R06}
A.~E. Rastegin, ``Sine distance for quantum states,'' {\em arXiv:quant-ph/0602112}, 2006.

\bibitem{chubb2017moderate}
C.~T. Chubb, V.~Y.~F. Tan, and M.~Tomamichel, ``Moderate deviation analysis for classical communication over quantum channels,'' {\em Communications in Mathematical Physics}, vol.~355, pp.~1283--1315, 2017.

\bibitem{devetak2004distilling}
I.~Devetak and A.~Winter, ``Distilling common randomness from bipartite quantum states,'' {\em IEEE Transactions on Information Theory}, vol.~50, no.~12, pp.~3183--3196, 2004.

\bibitem{OHHH02}
J.~Oppenheim, M.~Horodecki, P.~Horodecki, and R.~Horodecki, ``Thermodynamical approach to quantifying quantum correlations,'' {\em Physical Review Letters}, vol.~89, no.~18, p.~180402, 2002.

\bibitem{D05}
I.~Devetak, ``Distillation of local purity from quantum states,'' {\em Physical Review A}, vol.~71, no.~6, p.~062303, 2005.

\bibitem{DHW05RI}
I.~Devetak, A.~W. Harrow, and A.~Winter, ``A resource framework for quantum {Shannon} theory,'' {\em IEEE Transactions on Information Theory}, vol.~54, no.~10, pp.~4587--4618, 2008.

\bibitem{KD07}
H.~Krovi and I.~Devetak, ``Local purity distillation with bounded classical communication,'' {\em Physical Review A}, vol.~76, no.~1, p.~012321, 2007.

\bibitem{MPZ18}
G.~Manzano, F.~Plastina, and R.~Zambrini, ``Optimal work extraction and thermodynamics of quantum measurements and correlations,'' {\em Physical Review Letters}, vol.~121, p.~120602, Sept. 2018.

\bibitem{MLA19}
B.~Morris, L.~Lami, and G.~Adesso, ``Assisted work distillation,'' {\em Physical Review Letters}, vol.~122, p.~130601, Apr. 2019.

\bibitem{CNB22}
S.~Chakraborty, A.~Nema, and F.~Buscemi, ``One-shot purity distillation with local noisy operations and one-way classical communication,'' {\em arXiv:2208.05628}, Aug. 2022.

\bibitem{zhang2020expRand}
Y.~Zhang, L.~K. Shalm, J.~C. Bienfang, M.~J. Stevens, M.~D. Mazurek, S.~W. Nam, C.~Abell\'an, W.~Amaya, M.~W. Mitchell, H.~Fu, C.~A. Miller, A.~Mink, and E.~Knill, ``Experimental low-latency device-independent quantum randomness,'' {\em Physical Review Letters}, vol.~124, p.~010505, 2020.

\bibitem{shalm2021random}
L.~K. Shalm, Y.~Zhang, J.~C. Bienfang, C.~Schlager, M.~J. Stevens, M.~D. Mazurek, C.~Abellán, W.~Amaya, M.~W. Mitchell, M.~A. Alhejji, H.~Fu, J.~Ornstein, R.~P. Mirin, S.~W. Nam, and E.~Knill, ``Device-independent randomness expansion with entangled photons,'' {\em Nature Physics}, vol.~17, no.~4, pp.~452--456, 2021.

\bibitem{GPW05}
B.~Groisman, S.~Popescu, and A.~Winter, ``Quantum, classical, and total amount of correlations in a quantum state,'' {\em Physical Review A}, vol.~72, p.~032317, Sept. 2005.

\bibitem{CBR14}
N.~Ciganović, N.~J. Beaudry, and R.~Renner, ``Smooth max-information as one-shot generalization for mutual information,'' {\em IEEE Transactions on Information Theory}, vol.~60, no.~3, pp.~1573--1581, 2014.

\bibitem{ZHSL98}
K.~\ifmmode~\dot{Z}\else \.{Z}\fi{}yczkowski, P.~Horodecki, A.~Sanpera, and M.~Lewenstein, ``Volume of the set of separable states,'' {\em Physical Review A}, vol.~58, no.~2, pp.~883--892, 1998.

\bibitem{Vidal2002}
G.~Vidal and R.~F. Werner, ``Computable measure of entanglement,'' {\em Physical Review A}, vol.~65, p.~032314, Feb. 2002.
\newblock arXiv:quant-ph/0102117.

\bibitem{AJW17b}
A.~{Anshu}, R.~{Jain}, and N.~A. {Warsi}, ``Building blocks for communication over noisy quantum networks,'' {\em IEEE Transactions on Information Theory}, vol.~65, pp.~1287--1306, 2019.

\bibitem{khatri2019second}
S.~Khatri, E.~Kaur, S.~Guha, and M.~M. Wilde, ``Second-order coding rates for key distillation in quantum key distribution,'' 2019.
\newblock arXiv:1910.03883.

\bibitem{ChannelcodingRateYury}
Y.~Polyanskiy, H.~V. Poor, and S.~Verdu, ``Channel coding rate in the finite blocklength regime,'' {\em IEEE Transactions on Information Theory}, vol.~56, no.~5, pp.~2307--2359, 2010.

\bibitem{WNA20}
A.~B. Wagner, N.~V. Shende, and Y.~Altuğ, ``A new method for employing feedback to improve coding performance,'' {\em IEEE Transactions on Information Theory}, vol.~66, no.~11, pp.~6660--6681, 2020.

\bibitem{watrous2012simpler}
J.~Watrous, ``Simpler semidefinite programs for completely bounded norms,'' {\em Chicago Journal of Theoretical Computer Science}, no.~08, pp.~1--19, 2013.

\bibitem{SeeSawkonno1976cutting}
H.~Konno, ``A cutting plane algorithm for solving bilinear programs,'' {\em Mathematical Programming}, vol.~11, no.~1, pp.~14--27, 1976.

\bibitem{Bilinear_huber2019jointly}
S.~Huber, R.~K{\"o}nig, and M.~Tomamichel, ``Jointly constrained semidefinite bilinear programming with an application to {D}obrushin curves,'' {\em IEEE Transactions on Information Theory}, vol.~66, no.~5, pp.~2934--2950, 2019.

\bibitem{tomamichel2010duality}
M.~Tomamichel, R.~Colbeck, and R.~Renner, ``Duality between smooth min- and max-entropies,'' {\em IEEE Transactions on Information Theory}, vol.~56, no.~9, pp.~4674--4681, 2010.

\bibitem{schaffner2008robust}
C.~Schaffner, B.~Terhal, and S.~Wehner, ``Robust cryptography in the noisy-quantum-storage model,'' {\em Quantum Information and Computation}, vol.~9, no.~11, pp.~963--996, 2009.

\bibitem{nussbaum_2009}
M.~Nussbaum and A.~Szko{\l}a, ``The {{Chernoff}} lower bound for symmetric quantum hypothesis testing,'' {\em Annals of Statistics}, vol.~37, pp.~1040--1057, 2009.

\bibitem{audenaert_2007}
K.~M.~R. Audenaert, J.~Calsamiglia, R.~{Mu{\~n}oz-Tapia}, E.~Bagan, L.~Masanes, A.~Acin, and F.~Verstraete, ``Discriminating states: The quantum {C}hernoff bound,'' {\em Physical Review Letters}, vol.~98, p.~160501, 2007.

\end{thebibliography}

\appendices

\section{Supplementary Lemmas}

\label{app:supp-lems}

\begin{lemma}
\label{lem:gentle-meas-alt-proof}
Let $\rho$ be a state, $\Lambda$ a measurement operator, and set $\widetilde{\rho} \coloneqq \frac{\sqrt{\Lambda}\rho\sqrt{\Lambda}}{\operatorname{Tr}[\Lambda\rho]} $. Then
\begin{equation}
    F(\widetilde{\rho},\rho) \geq \operatorname{Tr}[\Lambda\rho].
\end{equation}
\end{lemma}

\begin{IEEEproof}
Consider that
\begin{align}
F(\widetilde{\rho},\rho) &  =\left(  \operatorname{Tr}\!\left[  \sqrt{\sqrt
{\rho}\widetilde{\rho}\sqrt{\rho}}\right]  \right)  ^{2}\\
&  =\frac{1}{\operatorname{Tr}[\Lambda\rho]}\left(  \operatorname{Tr}\!\left[
\sqrt{\sqrt{\rho}\sqrt{\Lambda}\rho\sqrt{\Lambda}\sqrt{\rho}}\right]  \right)
^{2}\\
&  =\frac{1}{\operatorname{Tr}[\Lambda\rho]}\left(  \operatorname{Tr}\!\left[
\sqrt{\sqrt{\rho}\sqrt{\Lambda}\sqrt{\rho}\sqrt{\rho}\sqrt{\Lambda}\sqrt{\rho
}}\right]  \right)  ^{2}\\
&  =\frac{1}{\operatorname{Tr}[\Lambda\rho]}\left(  \operatorname{Tr}\!\left[
\sqrt{\left(  \sqrt{\rho}\sqrt{\Lambda}\sqrt{\rho}\right)  ^{2}}\right]
\right)  ^{2}\\
&  =\frac{1}{\operatorname{Tr}[\Lambda\rho]}\left(  \operatorname{Tr}\!\left[
\sqrt{\rho}\sqrt{\Lambda}\sqrt{\rho}\right]  \right)  ^{2}\\
&  =\frac{1}{\operatorname{Tr}[\Lambda\rho]}\left(  \operatorname{Tr}\!\left[
\sqrt{\Lambda}\rho\right]  \right)  ^{2}\\
&  \geq\frac{1}{\operatorname{Tr}[\Lambda\rho]}\left(  \operatorname{Tr}%
\left[  \Lambda\rho\right]  \right)  ^{2}\\
&  =\operatorname{Tr}[\Lambda\rho],
\end{align}
concluding the proof.
\end{IEEEproof}

\begin{lemma}
\label{lem:alt-fid-flip-Lambda}
Let $\rho$ be a state, $\Lambda$ a measurement operator, and set $\widetilde{\rho} \coloneqq \frac{\sqrt{\Lambda}\rho\sqrt{\Lambda}}{\operatorname{Tr}[\Lambda\rho]} $. Then
\begin{equation}
     F(\widetilde{\rho},\sigma) = \frac{1}{\operatorname{Tr}[\Lambda\rho]}F(\rho,\sqrt{\Lambda}\sigma\sqrt{\Lambda}).
\end{equation}
\end{lemma}

\begin{IEEEproof}
Let $|\psi^{\sigma}\rangle$ purify $\sigma$, and let $|\psi^{\rho}\rangle$ purify $\rho$. Let $U$ denote a unitary acting on the purifying system.
Consider from Uhlmann's theorem \cite{uhlmann1976transition} that
\begin{align}
&  F(\widetilde{\rho},\sigma)\nonumber\\
&  =\sup_{U}\frac{1}{\operatorname{Tr}[\Lambda\rho]}\left\vert \langle
\psi^{\rho}|\left(  \sqrt{\Lambda}\otimes I\right)  \left(  I\otimes U\right)
|\psi^{\sigma}\rangle\right\vert ^{2}\\
&  =\frac{1}{\operatorname{Tr}[\Lambda\rho]}\sup_{U}\left\vert \langle
\psi^{\rho}|\sqrt{\Lambda}\otimes U|\psi^{\sigma}\rangle\right\vert ^{2}\\
&  =\frac{1}{\operatorname{Tr}[\Lambda\rho]}\sup_{U}\left\vert \langle\psi^{\rho}|\left(
I\otimes U\right)  \left(  \sqrt{\Lambda}\otimes I\right)  |\psi^{\sigma
}\rangle\right\vert ^{2}\\
&  =\frac{1}{\operatorname{Tr}[\Lambda\rho]}F(\rho,\sqrt{\Lambda}\sigma\sqrt{\Lambda}).
\end{align}
An alternative way of seeing this follows from%
\begin{align}
F(\widetilde{\rho},\sigma)  & =\left(  \operatorname{Tr}\!\left[  \sqrt
{\sqrt{\sigma}\widetilde{\rho}\sqrt{\sigma}}\right]  \right)  ^{2}\\
& =\frac{1}{\operatorname{Tr}[\Lambda\rho]}\left(  \operatorname{Tr}\!\left[
\sqrt{\sqrt{\sigma}\sqrt{\Lambda}\rho\sqrt{\Lambda}\sqrt{\sigma}}\right]
\right)  ^{2}\\
& =\frac{1}{\operatorname{Tr}[\Lambda\rho]}\left(  \operatorname{Tr}\!\left[
\sqrt{\sqrt{\sigma}\sqrt{\Lambda}\sqrt{\rho}\sqrt{\rho}\sqrt{\Lambda}%
\sqrt{\sigma}}\right]  \right)  ^{2}\\
& =\frac{1}{\operatorname{Tr}[\Lambda\rho]}\left(  \left\Vert \sqrt{\sigma
}\sqrt{\Lambda}\sqrt{\rho}\right\Vert _{1}\right)  ^{2}\\
& =\frac{1}{\operatorname{Tr}[\Lambda\rho]}\left(  \operatorname{Tr}\!\left[
\sqrt{\sqrt{\rho}\sqrt{\Lambda}\sqrt{\sigma}\sqrt{\sigma}\sqrt{\Lambda}%
\sqrt{\rho}}\right]  \right)  ^{2}\\
& =\frac{1}{\operatorname{Tr}[\Lambda\rho]}\left(  \operatorname{Tr}\!\left[
\sqrt{\sqrt{\rho}\sqrt{\Lambda}\sigma\sqrt{\Lambda}\sqrt{\rho}}\right]
\right)  ^{2}\\
& =\frac{1}{\operatorname{Tr}[\Lambda\rho]}F({\rho},\sqrt{\Lambda
}\sigma\sqrt{\Lambda}).
\end{align}
This concludes the proof.
\end{IEEEproof}

\section{Proof of Proposition~\ref{prop:connect-hypo-alt}}

\label{app:alt-proof-min-hypo-bnd}

The proof is quite similar to the proof of the Proposition~\ref{prop:connect-hypo}, but we give it here for completeness. Let $\Lambda$ be an arbitrary measurement operator satisfying
$\operatorname{Tr}[\Lambda\rho]\geq1-\varepsilon$. By the gentle measurement
lemma, we know that $\operatorname{Tr}[\Lambda\rho]\geq1-\varepsilon$ implies
that%
\begin{equation}
F(\widetilde{\rho},\rho)\geq1-\varepsilon(2-\varepsilon),
\label{eq:GML-fid-alt}
\end{equation}
where $\widetilde{\rho}=\sqrt{\Lambda}\rho\sqrt{\Lambda}$. 
To see the inequality in \eqref{eq:GML-fid-alt}, consider that
\begin{align}
F(\widetilde{\rho},\rho) &  =\left(  \operatorname{Tr}\!\left[  \sqrt{\sqrt
{\rho}\widetilde{\rho}\sqrt{\rho}}\right]  \right)  ^{2}\\
&  =\left(  \operatorname{Tr}\!\left[  \sqrt{\sqrt{\rho}\sqrt{\Lambda}\rho
\sqrt{\Lambda}\sqrt{\rho}}\right]  \right)  ^{2}\\
&  =\left(  \operatorname{Tr}\!\left[  \sqrt{\sqrt{\rho}\sqrt{\Lambda}\sqrt
{\rho}\sqrt{\rho}\sqrt{\Lambda}\sqrt{\rho}}\right]  \right)  ^{2}\\
&  =\left(  \operatorname{Tr}\!\left[  \sqrt{\left(  \sqrt{\rho}\sqrt{\Lambda
}\sqrt{\rho}\right)  ^{2}}\right]  \right)  ^{2}\\
&  =\left(  \operatorname{Tr}\!\left[  \sqrt{\rho}\sqrt{\Lambda}\sqrt{\rho
}\right]  \right)  ^{2}\\
&  =\left(  \operatorname{Tr}\!\left[  \sqrt{\Lambda}\rho\right]  \right)
^{2}\\
&  \geq\left(  \operatorname{Tr}\!\left[  \Lambda\rho\right]  \right)  ^{2}\\
&  \geq\left(  1-\varepsilon\right)  ^{2}\\
&  =1-\varepsilon\left(  2-\varepsilon\right)  .
\end{align}
Now we should relate $\operatorname{Tr}[\Lambda\sigma]$ to $F(\widetilde{\rho
},\sigma)$. Consider from Uhlmann's theorem \cite{uhlmann1976transition} that
\begin{align}
 F(\widetilde{\rho},\sigma)
&  =\sup_{U}\left\vert \langle\psi^{\rho}|\left(  \sqrt{\Lambda}\otimes
I\right)  \left(  I\otimes U\right)  |\psi^{\sigma}\rangle\right\vert ^{2}\\
&  =\sup_{U}\left\vert \langle\psi^{\rho}|\sqrt{\Lambda}\otimes U|\psi
^{\sigma}\rangle\right\vert ^{2}\\
&  =\sup_{U}\left\vert \langle\psi^{\rho}|\left(  I\otimes U\right)  \left(
\sqrt{\Lambda}\otimes I\right)  |\psi^{\sigma}\rangle\right\vert ^{2}\\
&  =F(\rho,\sqrt{\Lambda}\sigma\sqrt{\Lambda})\\
&  \leq\operatorname{Tr}[\Lambda\sigma].
\end{align}
The last inequality follows from data processing for fidelity under the trace
channel. Then%
\begin{align}
-\log_{2}\operatorname{Tr}[\Lambda\sigma] &  \leq-\log_{2}
F(\widetilde{\rho},\sigma)  \\
&  \leq D_{\min,F}^{\varepsilon\left(  2-\varepsilon\right)  }(\rho\Vert\sigma).
\end{align}
Since this holds for an arbitrary measurement operator satisfying
$\operatorname{Tr}[\Lambda\rho]\geq1-\varepsilon$, we conclude that%
\begin{equation}
D_{\min}^{\varepsilon}(\rho\Vert\sigma)\leq D_{\min,F}^{\varepsilon\left(
2-\varepsilon\right)  }(\rho\Vert\sigma),
\end{equation}
which is the desired statement.

\section{Properties of Max-Mutual Information}

\label{app:max-MI-props}

{In this appendix, we prove the equality in \eqref{eq:max-MI-rewrite}, we establish the bound
$I_{\max}(A;B)_{\rho}\leq2\log_{2}\min\left\{  d_{A},d_{B}\right\}  $, and we
prove the equality $I_{\max}(A;B)_{\Phi}=2\log_{2}d$ for a maximally entangled
state $\Phi_{AB}$ of Schmidt rank $d$. To begin, consider that%
\begin{align}
&  I_{\max}(A;B)_{\rho}\nonumber\\
&  =\inf_{\substack{\sigma_{A}\in\mathcal{D}(\mathcal{H}_{A}),\\\sigma_{B}%
\in\mathcal{D}(\mathcal{H}_{B})}}D_{\max}(\rho_{AB}\Vert\sigma_{A}%
\otimes\sigma_{B})\\
&  =\log_{2}\inf_{\substack{\lambda\geq0,\ \sigma_{A}\in\mathcal{D}%
(\mathcal{H}_{A}),\\\sigma_{B}\in\mathcal{D}(\mathcal{H}_{B})}}\left\{
\lambda:\rho_{AB}\leq\lambda\sigma_{A}\otimes\sigma_{B}\right\}  \\
&  =\log_{2}\inf_{K_{A},L_{B}\geq0}\{\operatorname{Tr}[K_{A}\otimes
L_{B}]:\rho_{AB}\leq K_{A}\otimes L_{B}\},
\end{align}
where the last equality follows from the substitution $\lambda\sigma
_{A}\otimes\sigma_{B}=K_{A}\otimes L_{B}$ and the fact that $\operatorname{Tr}%
[K_{A}\otimes L_{B}]=\operatorname{Tr}[\lambda\sigma_{A}\otimes\sigma
_{B}]=\lambda$. Now consider that, for every state $\rho_{AB}$, the operator
inequality $\rho_{AB}\leq d_{A}I_{A}\otimes\rho_{B}$ holds, because%
\begin{equation}
\frac{1}{d_{A}^{2}}\rho_{AB}\leq\frac{1}{d_{A}^{2}}\sum_{i=1}^{d_{A}^{2}}%
U_{A}^{i}\rho_{AB}\left(  U_{A}^{i}\right)  ^{\dag}=\frac{I_{A}}{d_{A}}%
\otimes\rho_{B},
\end{equation}
where $\left\{  U_{A}^{i}\right\}  _{i=1}^{d_{A}^{2}}$ is a set of
Heisenberg--Weyl unitaries. So this implies that the choices $K_{A}=d_{A}%
I_{A}$ and $L_{B}=\rho_{B}$ are feasible, leading to the claimed upper bound
$I_{\max}(A;B)_{\rho}\leq2\log_{2}\min d_{A}$. By a symmetric argument, the
following bound holds $I_{\max}(A;B)_{\rho}\leq2\log_{2}\min d_{B}$. Finally,
for a maximally entangled state, we have $I_{\max}(A;B)_{\Phi}\leq2\log_{2}d$
by the upper bound just derived. We also have%
\begin{align}
I_{\max}(A;B)_{\Phi}  & =\inf_{\substack{\sigma_{A}\in\mathcal{D}%
(\mathcal{H}_{A}),\\\sigma_{B}\in\mathcal{D}(\mathcal{H}_{B})}}D_{\max}%
(\Phi_{AB}\Vert\sigma_{A}\otimes\sigma_{B})\\
& \geq\inf_{\substack{\sigma_{A}\in\mathcal{D}(\mathcal{H}_{A}),\\\sigma
_{B}\in\mathcal{D}(\mathcal{H}_{B})}}D(\Phi_{AB}\Vert\sigma_{A}\otimes
\sigma_{B})\\
& =I(A;B)_{\Phi}\\
& =2\log_{2}d,
\end{align}
where the inequality follows from \eqref{eq:sandwiched-alpha-mono} and the second equality from \cite[Exercise~11.8.2]{Wbook17}. So we conclude that $I_{\max
}(A;B)_{\Phi}=2\log_{2}d$.}

\section{Proof of \cref{prop: SDP for smooth max normalized} } 

\label{proof: SDP for smooth max normalized}

First, let us verify that strong duality holds for the primal and dual SDPs. 
Consider the following feasible choices for the primal SDP: $\widetilde{\rho}=\rho$, $\lambda$ 
 such that $\rho \leq \lambda \sigma$ (one possible choice would be $\lambda= 2^{\widehat{D}_{\max}(\rho \Vert \sigma)}$), and $ X= \sqrt{1-\varepsilon} \  \rho$. 
  This follows because
  \begin{align}
  \begin{bmatrix}
\rho & X\\
X^{\dag} & {\widetilde{\rho}}%
\end{bmatrix} & = 
\begin{bmatrix}
\rho & \sqrt{1-\varepsilon} \  \rho\\
\sqrt{1-\varepsilon} \ \rho & {{\rho}}
\end{bmatrix} \\
& = \begin{bmatrix}
1 & \sqrt{1-\varepsilon} \\
\sqrt{1-\varepsilon}  & 1
\end{bmatrix} \otimes \rho\\
& \geq 0   .
 \end{align}
In addition, choosing $\mu$, $\nu$, $W$, and $Z$ to satisfy $\mu > 0$, $\nu >0$, $ \mu + \nu < 1/\Tr[\sigma]$, $W=(\mu+\nu) I$, and $Z = \nu I$ leads to strictly feasible choices for the dual program. Thus, strong duality holds, due to Slater's condition.

 Considering the known SDP for negative root fidelity \cite{watrous2012simpler}, we have 
\begin{align}
-\sqrt{F}(\rho,\widetilde{\rho}) &  =-\sup_{X\in\mathcal{L}(\mathcal{H}%
)}\left\{  \operatorname{Re}[\operatorname{Tr}[X]]:%
\begin{bmatrix}
\rho & X\\
X^{\dag} & \widetilde{\rho}%
\end{bmatrix}
\geq0\right\}  \\
&  =\inf_{X\in\mathcal{L}(\mathcal{H})}\left\{  -\operatorname{Re}%
[\operatorname{Tr}[X]]:%
\begin{bmatrix}
\rho & X\\
X^{\dag} & \widetilde{\rho}%
\end{bmatrix}
\geq0\right\} .
\end{align}
With that we find that%
\begin{align}
&\widehat{D}_{\max}^{\varepsilon}(\rho\Vert\sigma) \notag \\
&  =\log_{2}\inf_{\substack{\widetilde{\rho
}\geq0,\lambda\geq0, \\ X\in\mathcal{L}(\mathcal{H})}}\left\{
\begin{array}
[c]{c}%
\lambda:\widetilde{\rho}\leq\lambda\sigma,\operatorname{Tr}[\widetilde{\rho
}]=1,\\
-\operatorname{Re}[\operatorname{Tr}[X]]\leq-\sqrt{1-\varepsilon},\\%
\begin{bmatrix}
\rho & X\\
X^{\dag} & \widetilde{\rho}%
\end{bmatrix}
\geq0
\end{array}
\right\}  \\
&  =\log_{2}\inf_{\substack{\widetilde{\rho}\geq0,\lambda\geq0, \\ X\in\mathcal{L}%
(\mathcal{H})}}\left\{
\begin{array}
[c]{c}%
\lambda:\widetilde{\rho}\leq\lambda\sigma,\operatorname{Tr}[\widetilde{\rho
}]=1,\\
\operatorname{Re}[\operatorname{Tr}[X]]\geq\sqrt{1-\varepsilon},\\%
\begin{bmatrix}
\rho & X\\
X^{\dag} & \widetilde{\rho}%
\end{bmatrix}
\geq0
\end{array}
\right\}  .
\end{align}

Then, recall the standard form of dual and primal SDPs:
\begin{align}
\sup_{Z\geq0}\left\{  \operatorname{Tr}[AZ]:\Phi(Z)\leq B\right\},\\
\inf_{Y\geq0}\left\{  \operatorname{Tr}[BY]:\Phi^{\dag}(Y)\geq A\right\}  .
\end{align}
To that end, we find that%
\begin{align}
Y  &  =%
\begin{bmatrix}
\lambda & 0 & 0\\
0 & Z & X\\
0 & X^{\dag} & \widetilde{\rho}%
\end{bmatrix}
,\quad B=%
\begin{bmatrix}
1 & 0 & 0\\
0 & 0 & 0\\
0 & 0 & 0
\end{bmatrix}
,\\
\Phi^{\dag}(Y)  &  =%
\begin{bmatrix}
\lambda\sigma-\widetilde{\rho} & 0 & 0 & 0 & 0 & 0\\
0 & \operatorname{Tr}[\widetilde{\rho}] & 0 & 0 & 0 & 0\\
0 & 0 & -\operatorname{Tr}[\widetilde{\rho}] & 0 & 0 & 0\\
0 & 0 & 0 & \operatorname{Re}[\operatorname{Tr}[X]] & 0 & 0\\
0 & 0 & 0 & 0 & 0 & X\\
0 & 0 & 0 & 0 & X^{\dag} & \widetilde{\rho}%
\end{bmatrix}
,\\
A  &  =%
\begin{bmatrix}
0 & 0 & 0 & 0 & 0 & 0\\
0 & 1 & 0 & 0 & 0 & 0\\
0 & 0 & -1 & 0 & 0 & 0\\
0 & 0 & 0 & \sqrt{1-\varepsilon} & 0 & 0\\
0 & 0 & 0 & 0 & -\rho & 0\\
0 & 0 & 0 & 0 & 0 & 0
\end{bmatrix}
.
\end{align}
Setting%
\begin{equation}
Z=%
\begin{bmatrix}
W & 0 & 0 & 0 & 0 & 0\\
0 & \mu_{1} & 0 & 0 & 0 & 0\\
0 & 0 & \mu_{2} & 0 & 0 & 0\\
0 & 0 & 0 & \nu & 0 & 0\\
0 & 0 & 0 & 0 & Z_{1} & V\\
0 & 0 & 0 & 0 & V^{\dag} & Z_{2}%
\end{bmatrix}
,
\end{equation}
we find that%
\begin{align}
&  \operatorname{Tr}[Z\Phi^{\dag}(Y)]\cr
&  =\operatorname{Tr}[W(\lambda\sigma-\widetilde{\rho})]+\left(  \mu_{1}%
-\mu_{2}\right)  \operatorname{Tr}[\widetilde{\rho}]+\nu\operatorname{Re}%
[\operatorname{Tr}[X]] \nonumber \\ 
& \hspace{10mm} + \operatorname{Tr}\!\left[
\begin{bmatrix}
Z_{1} & V\\
V^{\dag} & Z_{2}%
\end{bmatrix}%
\begin{bmatrix}
0 & X\\
X^{\dag} & \widetilde{\rho}%
\end{bmatrix}
\right] \\
&  =\lambda\operatorname{Tr}[W\sigma]+\operatorname{Tr}[\left(  \left(
\mu_{1}-\mu_{2}\right)  I-W\right)  \widetilde{\rho}] \nonumber  
 \\ &\hspace{10mm} +\nu\operatorname{Re}%
[\operatorname{Tr}[X]]+\operatorname{Tr}\!\left[  VX^{\dag}+V^{\dag}%
X+Z_{2}\widetilde{\rho}\right] \\
&  =\lambda\operatorname{Tr}[W\sigma]+\operatorname{Tr}[\left(  \left(
\mu_{1}-\mu_{2}\right)  I-W+Z_{2}\right)  \widetilde{\rho}] \nonumber \\
& \hspace{10mm}+\operatorname{Re}%
[\operatorname{Tr}[\left(  \nu I+2V^{\dag}\right)  X]]\\
&  =\operatorname{Tr}\!\left[
\begin{bmatrix}
0 & 0\\
0 & \left(  \mu_{1}-\mu_{2}\right)  I-W+Z_{2}%
\end{bmatrix}%
\begin{bmatrix}
Z & X\\
X^{\dag} & \widetilde{\rho}%
\end{bmatrix}
\right]  \nonumber \\
&  +\operatorname{Tr}\!\left[
\begin{bmatrix}
0 & \frac{\nu}{2}I+V\\
\frac{\nu}{2}I+V^{\dag} & 0
\end{bmatrix}%
\begin{bmatrix}
Z & X\\
X^{\dag} & \widetilde{\rho}%
\end{bmatrix}
\right] +\lambda\operatorname{Tr}[W\sigma] 
\\
& = \operatorname{Tr}\!\left[
\begin{bmatrix}
0 & \frac{\nu}{2}I+V\\
\frac{\nu}{2}I+V^{\dag} & \left(  \mu_{1}-\mu_{2}\right)  I-W+Z_{2}%
\end{bmatrix}%
\begin{bmatrix}
Z & X\\
X^{\dag} & \widetilde{\rho}%
\end{bmatrix}
\right] \nonumber \\ 
& \qquad + \lambda\operatorname{Tr}[W\sigma].
\end{align}
This implies that%
\begin{equation}
\Phi(Z)=%
\begin{bmatrix}
\operatorname{Tr}[W\sigma] & 0\\
0 &
\begin{bmatrix}
0 & \frac{\nu}{2}I+V\\
\frac{\nu}{2}I+V^{\dag} & \left(  \mu_{1}-\mu_{2}\right)  I-W+Z_{2}%
\end{bmatrix}
\end{bmatrix}
\end{equation}
and we find that the dual SDP\ is given by%
\begin{align}
&  \sup_{Z\geq0}\left\{  \operatorname{Tr}[AZ]:\Phi(Z)\leq B\right\} \notag\\
& = \hspace{-3mm} \sup_{W,\mu_{1},\mu_{2},\nu\geq0}\left\{ \hspace{-2mm}
\begin{array}
[c]{c}%
\left(  \mu_{1}-\mu_{2}\right)  +\nu\sqrt{1-\varepsilon}-\operatorname{Tr}%
[Z_{1}\rho]:\\
\operatorname{Tr}[W\sigma]\leq1,\\%
\begin{bmatrix}
0 & \frac{\nu}{2}I+V\\
\frac{\nu}{2}I+V^{\dag} & \left(  \mu_{1}-\mu_{2}\right)  I-W+Z_{2}%
\end{bmatrix}
\leq0,\\%
\begin{bmatrix}
Z_{1} & V\\
V^{\dag} & Z_{2}%
\end{bmatrix}
\geq0
\end{array}
\hspace{-2mm}\right\} \\
&  =\sup_{W,\nu\geq0,\mu\in\mathbb{R}}\left\{
\begin{array}
[c]{c}%
\mu+\nu\sqrt{1-\varepsilon}-\operatorname{Tr}[Z_{1}\rho]:\\
\operatorname{Tr}[W\sigma]\leq1,\\%
\begin{bmatrix}
0 & \frac{\nu}{2}I+V\\
\frac{\nu}{2}I+V^{\dag} & \mu I-W+Z_{2}%
\end{bmatrix}
\leq0,\\%
\begin{bmatrix}
Z_{1} & V\\
V^{\dag} & Z_{2}%
\end{bmatrix}
\geq0
\end{array}
\right\} \\
&  =\sup_{W,\nu\geq0,\mu\in\mathbb{R}}\left\{
\begin{array}
[c]{c}%
\mu+2\nu\sqrt{1-\varepsilon}-\operatorname{Tr}[Z_{1}\rho]:\\
\operatorname{Tr}[W\sigma]\leq1,\\%
\begin{bmatrix}
0 & \nu I+V\\
\nu I+V^{\dag} & \mu I-W+Z_{2}%
\end{bmatrix}
\leq0,\\%
\begin{bmatrix}
Z_{1} & V\\
V^{\dag} & Z_{2}%
\end{bmatrix}
\geq0
\end{array}
\right\}
\end{align}
Consider that%
\begin{align}
\notag 
\begin{bmatrix}
0 & \nu I+V\\
\nu I+V^{\dag} & \mu I-W+Z_{2}%
\end{bmatrix}
&  \leq0 \label{eq: conversion of psd}\\
\Leftrightarrow\qquad%
\begin{bmatrix}
0 & V\\
V^{\dag} & Z_{2}%
\end{bmatrix}
&  \leq%
\begin{bmatrix}
0 & -\nu I\\
-\nu I & W-\mu I
\end{bmatrix}
\\
\Leftrightarrow\qquad%
\begin{bmatrix}
Z_{1} & V\\
V^{\dag} & Z_{2}%
\end{bmatrix}
&  \leq%
\begin{bmatrix}
Z_{1} & -\nu I\\
-\nu I & W-\mu I
\end{bmatrix}
.
\end{align}
Thus, we can eliminate the matrix variable $%
\begin{bmatrix}
Z_{1} & V\\
V^{\dag} & Z_{2}%
\end{bmatrix}
$, and the SDP reduces to
\begin{align}
&  \sup_{W,\nu,Z_{1}\geq0,\mu\in\mathbb{R}}\left\{
\begin{array}
[c]{c}%
\mu+2\nu\sqrt{1-\varepsilon}-\operatorname{Tr}[Z_{1}\rho]:\\
\operatorname{Tr}[W\sigma]\leq1,\\%
\begin{bmatrix}
Z_{1} & -\nu I\\
-\nu I & W-\mu I
\end{bmatrix}
\geq0
\end{array}
\right\} \notag \\
&  =\sup_{W,\nu,Z\geq0,\mu\in\mathbb{R}}\left\{
\begin{array}
[c]{c}%
\mu+2\nu\sqrt{1-\varepsilon}-\operatorname{Tr}[Z\rho]:\\
\operatorname{Tr}[W\sigma]\leq1,\\%
\begin{bmatrix}
Z & -\nu I\\
-\nu I & W-\mu I
\end{bmatrix}
\geq0
\end{array}
\right\} \\
&  =\sup_{W,\nu,Z\geq0,\mu\in\mathbb{R}}\left\{
\begin{array}
[c]{c}%
\mu+2\nu\sqrt{1-\varepsilon}-\operatorname{Tr}[Z\rho]:\\
\operatorname{Tr}[W\sigma]\leq1,\\%
\begin{bmatrix}
Z & \nu I\\
\nu I & W-\mu I
\end{bmatrix}
\geq0
\end{array}
\right\}  .
\end{align}

\section{SDP Dual of Smooth Conditional Min-Entropy}

\label{proof:dual-of-smooth-conditional-min-entropy}

In \cref{cor: SDP for smooth conditional min entropy}, we presented the primal SDP of the smooth conditional  min-entropy. 
The dual SDP for smooth conditional  min-entropy is as follows:
\begin{multline}
   H_{\min}^{{\varepsilon}}(A|B)_{\rho}= \\
   -\log_2 \sup_{\substack{W_{AB}\geq 0,\\Z_{AB} \geq 0,\\ \nu\geq0,\mu \geq 0}}\left\{
\begin{array}
[c]{c}%
-\mu+2\nu\sqrt{1-\varepsilon^2}-\operatorname{Tr}[Z_{AB}\rho_{AB}]:\\
\operatorname{Tr}_{A}[W_{AB}]\leq I_{B},\\%
\begin{bmatrix}
Z_{AB} & \nu I_{AB}\\
\nu I_{AB} & W_{AB}+\mu I_{AB}%
\end{bmatrix}
\geq0
\end{array}
\right\}  ,
\label{eq:dual-of-smooth-conditional-entropy}
\end{multline} 

Before proving the dual, let us verify that strong duality holds for the primal and dual SDPs. 
Consider the following feasible choices for the primal SDP: $\widetilde{\rho}_{AB}=\rho_{AB}$, $S_B = d_B \rho_B$, and $ X_{AB}= (\sqrt{1-\varepsilon^2}) \rho_{AB}$, where $d_B$ is the dimension of the $B$ system. This follows because the operator inequality $\rho_{AB} \leq I_A \otimes d_B \rho_B$ holds for every state $\rho_{AB}$ (see, e.g., just after Eq.~(34) in \cite{dupuis2014generalizedDmin}). Furthermore,
 \begin{align}
  \begin{bmatrix}
\rho_{AB} & X_{AB}\\
X_{AB}^{\dag} & {\widetilde{\rho}_{AB}}%
\end{bmatrix} & = 
\begin{bmatrix}
\rho_{AB} & \sqrt{1-\varepsilon^2} \rho_{AB}\\
\sqrt{1-\varepsilon^2} \rho_{AB} & {{\rho}_{AB}}
\end{bmatrix} \\
& = \begin{bmatrix}
1 & \sqrt{1-\varepsilon^2} \\
\sqrt{1-\varepsilon^2}  & 1
\end{bmatrix} \otimes \rho_{AB}\\
& \geq 0   .
 \end{align}
In addition, choosing $\mu$, $\nu$, $W_{AB}$, and $Z_{AB}$ to satisfy $\nu > \mu > 0$,  $d_A (\nu -\mu)< 1$, $W_{AB} = (\nu-\mu)I_{AB}$, $Z_{AB} = \nu I_{AB}$, and $W_{AB}=(\nu-\mu) I_{AB}$ leads to strictly feasible choices for the dual program. Thus, strong duality holds, due to Slater's condition.

Recall the primal of smooth conditional min-entropy given in \eqref{eq: primal smooth min entropy conditional}
and the standard form of SDPs:
\begin{align}
 \sup_{Z\geq0}\left\{  \operatorname{Tr}[AZ]:\Phi(Z)\leq B\right\} ,\\
 \inf_{Y\geq0}\left\{  \operatorname{Tr}[BY]:\Phi^{\dag}(Y)\geq A\right\}  .
\end{align}
In standard form, this SDP\ is given by%
\begin{align}
Y  &  =%
\begin{bmatrix}
S_{B} & 0 & 0\\
0 & W_{AB} & X_{AB}\\
0 & X_{AB}^{\dag} & \widetilde{\rho}_{AB}%
\end{bmatrix}
,\quad B=%
\begin{bmatrix}
I_{B} & 0 & 0\\
0 & 0 & 0\\
0 & 0 & 0
\end{bmatrix}
,\\
\Phi^{\dag}(Y)  &  =%
\begin{bmatrix}
L_{AB} & 0 & 0 & 0 & 0\\
0 & -\operatorname{Tr}[\widetilde{\rho}_{AB}] & 0 & 0 & 0\\
0 & 0 & \operatorname{Re}[\operatorname{Tr}[X_{AB}]] & 0 & 0\\
0  & 0 & 0 & 0 & X_{AB}\\
0 & 0 & 0 & X_{AB}^{\dag} & \widetilde{\rho}_{AB}%
\end{bmatrix}
,\\
\textnormal{where }  & L_{AB}=I_{A}\otimes S_{B}-\widetilde{\rho}_{AB}, \textnormal{ and} \nonumber \\
A  &  =%
\begin{bmatrix}
0  & 0 & 0 & 0 & 0\\
0 & -1 & 0 & 0 & 0\\
0  & 0 & \sqrt{1-\varepsilon^2} & 0 & 0\\
0 & 0 & 0 & -\rho_{AB} & 0\\
0 & 0 & 0 & 0 & 0
\end{bmatrix}
.
\end{align}
Setting%
\begin{equation}
Z=%
\begin{bmatrix}
Z_{AB}^{1}  & 0 & 0 & 0 & 0\\

0 & \mu & 0 & 0 & 0\\
0  & 0 & \nu & 0 & 0\\
0  & 0 & 0 & Z_{AB}^{2} & V_{AB}\\
0 & 0 & 0 & V_{AB}^{\dag} & Z_{AB}^{3}%
\end{bmatrix}
,
\end{equation}
we find that%
\begin{align}
&  \operatorname{Tr}[Z\Phi^{\dag}(Y)]\notag \\
&  =\operatorname{Tr}[Z_{AB}^{1}(I_{A}\otimes S_{B}-\widetilde{\rho}%
_{AB})]-\mu \operatorname{Tr}[\widetilde{\rho
}_{AB}] \nonumber \\
&\qquad +\operatorname{Tr}\!\left[
\begin{bmatrix}
Z_{AB}^{2} & V_{AB}\\
V_{AB}^{\dag} & Z_{AB}^{3}%
\end{bmatrix}%
\begin{bmatrix}
0 & X_{AB}\\
X_{AB}^{\dag} & \widetilde{\rho}_{AB}%
\end{bmatrix}
\right]  \notag\\
& \qquad +\nu\operatorname{Re}[\operatorname{Tr}[X_{AB}]] \\
&  =\operatorname{Tr}[\operatorname{Tr}_{A}[Z_{AB}^{1}]S_{B}%
]+\nu\operatorname{Re}[\operatorname{Tr}[X_{AB}]]+2\operatorname{Re}%
[\operatorname{Tr}[V_{AB}^{\dag}X_{AB}]] \nonumber \\
& \quad +\operatorname{Tr}[\left(  -\mu  I_{AB}-Z_{AB}%
^{1}+Z_{AB}^{3}\right)  \widetilde{\rho}_{AB}]\\
&  =\operatorname{Tr}[\operatorname{Tr}_{A}[Z_{AB}^{1}]S_{B}%
] \nonumber \\ 
&+\operatorname{Tr}\!\left[
\begin{bmatrix}
0 & L_{AB}\\
L_{AB}^{\dag} & -\mu
I_{AB}-Z_{AB}^{1}+Z_{AB}^{3}%
\end{bmatrix}%
\begin{bmatrix}
W_{AB} & X_{AB}\\
X_{AB}^{\dag} & \widetilde{\rho}_{AB}%
\end{bmatrix}
\right],
\end{align}
with the shorthand 
\begin{equation}
    L_{AB} \coloneqq \frac{\nu}{2}I_{AB}+V_{AB}.
\end{equation}
This implies that%
\begin{equation}
\Phi(Z)=%
\begin{bmatrix}
\operatorname{Tr}_{A}[Z_{AB}^{1}] & 0 & 0\\
0 & 0 & L_{AB}\\
0 & L_{AB}^{\dag} & -\mu
I_{AB}-Z_{AB}^{1}+Z_{AB}^{3}%
\end{bmatrix}
.
\end{equation}

Then we find that the dual SDP\ is given by%
\begin{align}
&  \sup_{Z\geq0}\left\{  \operatorname{Tr}[AZ]:\Phi(Z)\leq B\right\} \\
&  =\sup_{\substack{Z_{AB}^{1} \geq 0, \\ \mu \geq 0,\\ \nu\geq0}}\left\{
\begin{array}
[c]{c}%
-\mu+\nu\sqrt{1-\varepsilon^2}-\operatorname{Tr}[Z_{AB}^{2}\rho
_{AB}]:\\
\operatorname{Tr}_{A}[Z_{AB}^{1}]\leq I_{B},\\%
\begin{bmatrix}
0 & L_{AB}\\
L_{AB}^{\dag} & -\mu
I_{AB}-Z_{AB}^{1}+Z_{AB}^{3}%
\end{bmatrix}
\leq0\\%
\begin{bmatrix}
Z_{AB}^{2} & V_{AB}\\
V_{AB}^{\dag} & Z_{AB}^{3}%
\end{bmatrix}
\geq0
\end{array}
\right\} \\
&  =\sup_{\substack{Z_{AB}^{1} \geq 0,\\ \nu\geq0, \\ \mu \geq 0}}\left\{
\begin{array}
[c]{c}%
-\mu+\nu\sqrt{1-\varepsilon^2}-\operatorname{Tr}[Z_{AB}^{2}\rho_{AB}]:\\
\operatorname{Tr}_{A}[Z_{AB}^{1}]\leq I_{B},\\%
\begin{bmatrix}
0 & L_{AB}\\
L_{AB}^{\dag} & -\mu I_{AB}-Z_{AB}^{1}+Z_{AB}^{3}%
\end{bmatrix}
\leq0\\%
\begin{bmatrix}
Z_{AB}^{2} & V_{AB}\\
V_{AB}^{\dag} & Z_{AB}^{3}%
\end{bmatrix}
\geq0
\end{array}
\right\} \\
&  =\sup_{\substack{Z_{AB}^{1} \geq 0, \\ \nu\geq0, \\ \mu \geq 0}}\left\{
\begin{array}
[c]{c}%
-\mu+2\nu\sqrt{1-\varepsilon^2}-\operatorname{Tr}[Z_{AB}^{2}\rho_{AB}]:\\
\operatorname{Tr}_{A}[Z_{AB}^{1}]\leq I_{B},\\%
\begin{bmatrix}
0 & L'_{AB}\\
(L'_{AB})^{\dag} & -\mu I_{AB}-Z_{AB}^{1}+Z_{AB}^{3}%
\end{bmatrix}
\leq0\\%
\begin{bmatrix}
Z_{AB}^{2} & V_{AB}\\
V_{AB}^{\dag} & Z_{AB}^{3}%
\end{bmatrix}
\geq0
\end{array}
\right\}  ,
\end{align}
where 
\begin{equation}
    L'_{AB} \coloneqq \nu I_{AB}+V_{AB}.
\end{equation}
Then, as in derivations related to \eqref{eq: conversion of psd}, we find that%
\begin{align}%
\begin{bmatrix}
0 & \nu I_{AB}+V_{AB}\\
\nu I_{AB}+V_{AB}^{\dag} & -\mu I_{AB}-Z_{AB}^{1}+Z_{AB}^{3}%
\end{bmatrix}
 \leq0\\
\hspace{-10mm} \Leftrightarrow%
\begin{bmatrix}
Z_{AB}^{2} & V_{AB}\\
V_{AB}^{\dag} & Z_{AB}^{3}%
\end{bmatrix}
  \leq%
\begin{bmatrix}
Z_{AB}^{2} & -\nu I_{AB}\\
-\nu I_{AB} & Z_{AB}^{1}+\mu I_{AB}%
\end{bmatrix}
,
\end{align}
and we can again eliminate variables to reduce the SDP to%
\begin{align}
&  \sup_{Z_{AB}^{1},\nu\geq0,\mu \geq 0}\left\{
\begin{array}
[c]{c}%
-\mu+2\nu\sqrt{1-\varepsilon^2}-\operatorname{Tr}[Z_{AB}^{2}\rho_{AB}]:\\
\operatorname{Tr}_{A}[Z_{AB}^{1}]\leq I_{B},\\%
\begin{bmatrix}
Z_{AB}^{2} & -\nu I_{AB}\\
-\nu I_{AB} & Z_{AB}^{1}+\mu I_{AB}%
\end{bmatrix}
\geq0
\end{array}
\right\} \\
&  =\sup_{Z_{AB}^{1},\nu\geq0,\mu \geq 0}\left\{
\begin{array}
[c]{c}%
-\mu+2\nu\sqrt{1-\varepsilon^2}-\operatorname{Tr}[Z_{AB}^{2}\rho_{AB}]:\\
\operatorname{Tr}_{A}[Z_{AB}^{1}]\leq I_{B},\\%
\begin{bmatrix}
Z_{AB}^{2} & \nu I_{AB}\\
\nu I_{AB} & Z_{AB}^{1}+\mu I_{AB}%
\end{bmatrix}
\geq0
\end{array}
\right\}
\end{align}

We then do a final rewriting as%
\begin{equation}
\sup_{\substack{W_{AB},Z_{AB} \geq 0,\\ \nu\geq0,\mu \geq 0}}\left\{
\begin{array}
[c]{c}%
-\mu+2\nu\sqrt{1-\varepsilon^2}-\operatorname{Tr}[Z_{AB}\rho_{AB}]:\\
\operatorname{Tr}_{A}[W_{AB}]\leq I_{B},\\%
\begin{bmatrix}
Z_{AB} & \nu I_{AB}\\
\nu I_{AB} & W_{AB}+\mu I_{AB}%
\end{bmatrix}
\geq0
\end{array}
\right\}  .
\end{equation}

\section{Second-order Asymptotics of Smooth Sandwiched R\'enyi Relative Entropy} \label{APP:Second-order-smooth-renyi}

\begin{proposition}
Fix $\varepsilon\in(0,1)$. For a state $\rho$ and a PSD operator $\sigma$, the following second-order expansions hold:
\begin{enumerate}
    \item For $\alpha>1$:
  \begin{multline}
    \frac{1}{n} \widetilde{D}_{\alpha}^{\varepsilon}\!\left(\rho^{\otimes n} \Vert\sigma^{\otimes n}\right) = \\  D(\rho \Vert \sigma) -\sqrt{\frac{1}{n} V(\rho \Vert \sigma)} \ \Phi^{-1}(\varepsilon) + O\!\left( \frac{\log n}{n} \right).
\end{multline} 
    \item For $\alpha \in [1/2,1)$: 
    \begin{multline}
    \frac{1}{n} \widetilde{D}_{\alpha}^{\varepsilon}\!\left(\rho^{\otimes n} \Vert\sigma^{\otimes n}\right) = \\  D(\rho \Vert \sigma) + \sqrt{\frac{1}{n} V(\rho \Vert \sigma)} \ \Phi^{-1}(\varepsilon) + O\!\left( \frac{\log n}{n} \right).
\end{multline} 
\end{enumerate}
\end{proposition}

\begin{IEEEproof}
   \underline{Part (1):}

   For the upper bound, due to the $\beta$-monotonicity of sandwiched R\'enyi relative entropy for $\beta >1$, we have 
   \begin{align}
    & \frac{1}{n} \widetilde{D}_{\beta}^{\varepsilon}\!\left(\rho^{\otimes n} \Vert\sigma^{\otimes n}\right) \notag \\ 
    &\leq   \frac{1}{n} {D}_{\max}^{\varepsilon}\!\left(\rho^{\otimes n} \Vert\sigma^{\otimes n}\right) \\
    & \stackrel{(a)}\leq \frac{1}{n} {D}_{\min}^{1-\varepsilon}\!\left(\rho^{\otimes n} \Vert\sigma^{\otimes n}\right) + \frac{1}{n}\log_2\!\left( \frac{1}{1-\varepsilon}\right) \\ 
    & \stackrel{(b)} = D(\rho \Vert \sigma) + \sqrt{\frac{1}{n} V(\rho \Vert \sigma)} \ \Phi^{-1}(1-\varepsilon) + O\!\left( \frac{\log n}{n} \right) \\
    & \stackrel{(c)}= D(\rho \Vert \sigma) - \sqrt{\frac{1}{n} V(\rho \Vert \sigma)} \ \Phi^{-1}(\varepsilon) + O\!\left( \frac{\log n}{n} \right) \label{eq:-upperbound-smooth-renyi-higher}
   \end{align}
   where: (a) follows from \eqref{eq:smooth-max-to-hypothesis}; (b) from \eqref{eq:hypothesis-testing-second-order}; and (c) from $\Phi^{-1}(1-\varepsilon)=-\Phi^{-1}(\varepsilon) $.

   For the lower bound,
   we use \cref{thm:connection-to-smooth-renyi}  for $\varepsilon
,\delta\in\left(  0,1\right)  $ such that $\varepsilon+\delta\in\left(
0,1\right)$, as well as the $\alpha \to 1/2$ limit of the lower bound, by $\alpha$-monotonicity of $\widetilde{D}_{\alpha}$, and then consider 
\begin{multline}
 \frac{1}{n} \widetilde{D}_{\beta}^{\varepsilon}(\rho^{\otimes n}\Vert\sigma^{\otimes n}) \\
 \geq 
    \frac{1}{n}D^{1-\varepsilon-\delta}_{\min,F} (\rho^{\otimes n}  \Vert \sigma^{\otimes n}) - \frac{1} {n} \frac{\beta}{\beta-1}\log_{2} \!\left(
\frac{1}{1-f(\varepsilon,\delta)}\right), \label{eq:smooth-renyi-and-min-n-tensor}
\end{multline}
where $f(\varepsilon,\delta)$ is given in \eqref{eq:f(eps,delta)},
and we know from \eqref{eq: f function taylor} that $-\log\!\left(1-f(\varepsilon,\delta)\right)= O(\log n)$  for $\delta=1/ \sqrt{n}$ and sufficiently large $n$.

Now, using \cref{thm:second-order-Dmin}, consider that
\begin{align}
     &\frac{1}{n}D^{1-\varepsilon-\delta}_{\min,F} (\rho^{\otimes n}  \Vert \sigma^{\otimes n}) \notag \\
     &= D(\rho \Vert \sigma) + \sqrt{\frac{1}{n} V(\rho \Vert \sigma)} \ \Phi^{-1}(1-\varepsilon-\delta) + O\!\left( \frac{\log n}{n} \right) \\
     &= D(\rho \Vert \sigma) -\sqrt{\frac{1}{n} V(\rho \Vert \sigma)} \ \Phi^{-1}(\varepsilon+\delta) + O\!\left( \frac{\log n}{n} \right) \\
     & = D(\rho \Vert \sigma) -\sqrt{\frac{1}{n} V(\rho \Vert \sigma)} \ \Phi^{-1}(\varepsilon) + O\!\left( \frac{\log n}{n} \right),     
\end{align}
where the last equality holds from similar reasoning used to arrive at \eqref{eq:upperbound-second-order}, which was used in the proof of \cref{thm:second-order-Dmin} by the choice $\delta=1/ \sqrt{n}$.

Then, combining the above inequality with \eqref{eq:smooth-renyi-and-min-n-tensor}, we obtain the desired lower bound. Finally together with \eqref{eq:-upperbound-smooth-renyi-higher}, we complete the proof for the case $\beta >1$.

\medskip 
\underline{Part (2):}
For $\alpha \in [1/2,1)$, from the $\alpha$-monotonicity of the sandwiched R\'enyi relative entropy, we have 
\begin{align}
   &  \frac{1}{n} \widetilde{D}_{\alpha}^{\varepsilon}\!\left(\rho^{\otimes n} \Vert\sigma^{\otimes n}\right)   \notag \\
   & \geq  \frac{1}{n} {D}_{\min,F}^{\varepsilon}\!\left(\rho^{\otimes n} \Vert\sigma^{\otimes n}\right) \\
   & = D(\rho \Vert \sigma) + \sqrt{\frac{1}{n} V(\rho \Vert \sigma)} \ \Phi^{-1}(\varepsilon) + O\!\left( \frac{\log n}{n} \right) \label{eq: lowe bound smooth renyi lower},
\end{align}
where the equality follows from \cref{thm:second-order-Dmin}.

For the upper bound, using \cref{thm:connection-to-smooth-renyi}, consider that 
\begin{align}
    &\frac{1}{n} \widetilde{D}_{\alpha}^{\varepsilon}\!\left(\rho^{\otimes n} \Vert\sigma^{\otimes n}\right)  \notag \\
    &\leq \frac{1}{n} \widetilde{D}_{\beta}^{1-\varepsilon-\delta}\!\left(\rho^{\otimes n} \Vert\sigma^{\otimes n}\right) + \frac{1}{n}\frac{\beta}{\beta-1}\log\!\left( \frac{1}{1- f(\varepsilon,\delta)} \right) \\ 
    &= D(\rho \Vert \sigma) -\sqrt{\frac{1}{n} V(\rho \Vert \sigma)} \ \Phi^{-1}(1-\varepsilon-\delta) + O\!\left( \frac{\log n}{n} \right) \notag  \\ 
    & \qquad +\frac{1}{n}\frac{\beta}{\beta-1}\log\!\left( \frac{1}{1- f(\varepsilon,\delta)} \right),
\end{align}
where the last inequality follows from the second-order expansion obtained in Part (1). 
Now observing that $-\log\!\left(1-f(\varepsilon,\delta)\right)= O(\log n)$ for $\delta=1/ \sqrt{n}$ and sufficiently large $n$, and following the same reasoning, we used to arrive at \eqref{eq:upperbound-second-order}, in the proof of \cref{thm:second-order-Dmin}, we get the matching upper bound. 

Combining the obtained upper bound and \eqref{eq: lowe bound smooth renyi lower}, we conclude the proof. 
\end{IEEEproof}

\section{Moderate Deviation Analysis for Smooth \texorpdfstring{$F$}{F}-Min-Relative Entropy (Proof of \cref{prop:mod-dev-smooth-F-min})}

\label{App: moderate deviation}

In this appendix, we prove \cref{prop:mod-dev-smooth-F-min}. Recall again that a sequence $\{a_n\}_n$ is called a moderate sequence if $a_n \to 0$ and $\sqrt{n} a_n \to \infty$ when $n \to \infty$.
From \cite{chubb2017moderate}, we have the following scaling of the smooth min-relative entropy under moderate deviations, where $\varepsilon_n \coloneqq e^{-n a^2_n}$,
\begin{equation}\label{eq: hypothesis testing moderate}
   \frac{1}{n} D^{\varepsilon_n}_{\min}\!\left(\rho^{\otimes n} \Vert \sigma^{\otimes n }\right) =
    D(\rho \Vert \sigma) - \sqrt{2 V(\rho \Vert \sigma)} \ a_n +o(a_n).
\end{equation}

\begin{IEEEproof}[Proof of \cref{prop:mod-dev-smooth-F-min}]
 For the lower bound, we employ \cref{prop:connect-hypo} to find that
    \begin{align}
  & \frac{1}{n}D^{\varepsilon_n}_{\min,F} (\rho^{\otimes n} \Vert \sigma^{\otimes n})     \cr
  & \geq \frac{1}{n} D^{\varepsilon_n}_{\min}(\rho^{\otimes n} \Vert \sigma^{\otimes n})
  - \frac{1}{n}\log_{2} \!\left(\frac{1}{1-\varepsilon_n}\right) \\ 
  & = D(\rho \Vert \sigma) - \sqrt{2 V(\rho\Vert \sigma)} \ a_n + o(a_n),
\end{align}
where the last equality holds by \eqref{eq: hypothesis testing moderate}.
Then, we arrive at 
\begin{equation}
    \frac{1}{n}D^{\varepsilon_n}_{\min,F} (\rho^{\otimes n} \Vert \sigma^{\otimes n}) \\
    \geq  D(\rho \Vert \sigma) -\sqrt{2V(\rho\Vert \sigma)} \ a_n + o(a_n), \label{eq: lower bound moderate D min}
\end{equation}
along with $\frac{1}{n} \log_2\!\left(\frac{1}{1-\varepsilon_n}\right)=o\!\left(\frac{1}{n}\right) $ leading to $\frac{1}{n} \log_2\!\left(\frac{1}{1-\varepsilon_n}\right)=o(a_n)$ since $na^2_n \to \infty$.

For the upper bound, similar to the proof of \cref{thm:second-order-Dmin}, specifically using the relationship derived in \eqref{eq: connecting with hypothesis eq}, we have 
\begin{align}
 &  \frac{1}{n}D^{\varepsilon_n}_{\min,F} (\rho^{\otimes n}  \Vert \sigma^{\otimes n}) \notag \\
&  \leq \frac{1}{n} D^{\varepsilon_n+ \delta}_{\min}(\rho^{\otimes n} \Vert \sigma^{\otimes n})  +  \frac{1}{n} \log_{2}\!\left(\frac{1}{\varepsilon_n+\delta}\right)  \notag \\
& \qquad +  \frac{1}{n}\log_{2} \!\left(  \frac{1}{1-f(\varepsilon_n,\delta)}  \right).
\end{align}
Since $\delta \in(0,1)$, by choosing $\delta=\varepsilon_n$, 
$ \frac{1}{n}\log_{2} \!\left(  \frac{1}{1-f(\varepsilon_n,\varepsilon_n)} \right)= o(a_n)$ and $\frac{1}{n}\log_2\!\left( \frac{1}{2\varepsilon_n }\right)=o(a_n)$. Also observe that 
\begin{equation}
    2 e^{-n a^2_n} = \exp\!\left( -n\left( a^2_n - \frac{\ln 2}{n} \right)\right)= e^{-n b^2_n},
\end{equation}
where $b_n \coloneqq a_n + o(a_n)$.
Collecting these observations together, we obtain the upper bound 
\begin{equation}
    \frac{1}{n}D^{\varepsilon_n}_{\min,F} (\rho^{\otimes n} \Vert \sigma^{\otimes n}) \\
    \leq   D(\rho \Vert \sigma) -\sqrt{2V(\rho\Vert \sigma)} \ a_n + o(a_n).\label{eq: upper bound moderate D min}
\end{equation}
We conclude the proof by combining \eqref{eq: lower bound moderate D min} and \eqref{eq: upper bound moderate D min}.
\end{IEEEproof}

\end{document}